\documentclass[preprint,3p,sort&compress]{elsarticle}
\journal{Applied and Computational Harmonic Analysis}

\usepackage{graphicx}  
\usepackage{bm}        
\usepackage{pgfplots}
\usepackage{epstopdf}
\usepackage{placeins}
\usepackage{float}\usepackage{amsmath}

\usepackage{amsthm}
\usepackage{graphicx}  
\usepackage{dcolumn}   
\usepackage{bm}        
\usepackage{amssymb}   

\usepackage{hyperref}
\usepackage{placeins}
\usepackage{tabularx}
\usepackage{float}
\usepackage{algorithm}
\usepackage{algorithmicx}
\usepackage{algpseudocode}
\usepackage[caption=false]{subfig}
\usepackage{tikz,ifthen,pgfplots}

\newtheorem{mythm}{Theorem}
\newtheorem{myprop}{Proposition}
\newtheorem{mylemma}{Lemma}
\newtheorem{mydef}{Definition}
\newtheorem{mycorol}{Corollary}
\DeclareMathOperator{\diag}{diag}
\DeclareMathOperator{\vol}{vol}
\DeclareMathOperator{\Var}{Var}

\let\oldReturn\Return
\renewcommand{\Return}{\State\oldReturn}

\pgfplotsset{compat=1.3, every axis label/.append style={font=\scriptsize}, tick label style={font=\tiny}, legend style={font=\scriptsize}}
\newlength\figureheight 
\newlength\figurewidth 
\definecolor{lavander}{cmyk}{0,0.48,0,0}
\definecolor{violet}{cmyk}{0.79,0.88,0,0}
\definecolor{burntorange}{cmyk}{0,0.52,1,0}
\def\lav{lavander!90}
\def\oran{orange!30}
\tikzstyle{Node2}=[draw,circle,burntorange, left color=\oran,
                       text=violet,minimum width=18pt]
\tikzstyle{Node1}=[draw,circle,lavander, left color=\lav,
                       text=violet,minimum width=10pt]
\tikzstyle{light}=[draw,circle,violet, left color=violet,
                      text=violet,minimum width=10pt]

\tikzset{
  LabelStyle/.style = { rectangle, rounded corners, draw,
                        minimum width = 2em, fill = yellow!50,
                        text = red, font = \bfseries },
  VertexStyle/.append style = { inner sep=5pt,
                                font = \Large\bfseries},
  EdgeStyle/.append style = {->, bend left} }   

\begin{document}

\title{Deformed Laplacians and spectral ranking in directed networks}
\author{M. Fanuel\footnote{michael.fanuel@esat.kuleuven.be} \ and J.A.K. Suykens\footnote{johan.suykens@esat.kuleuven.be},\\
\emph{\small KU Leuven, Department of Electrical Engineering (ESAT),}\\
\emph{\small STADIUS Center for Dynamical Systems, Signal Processing and Data Analytics,}\\
\emph{\small Kasteelpark Arenberg 10, B-3001 Leuven, Belgium}
}

\date{August 24, 2016}

\begin{abstract} 

Deformations of the combinatorial Laplacian are proposed, which generalize several existing Laplacians. As particular cases of this construction, the dilation Laplacians are shown to be useful tools for ranking in directed networks of pairwise comparisons. In the case of a connected graph, the entries eigenvector with the smallest eigenvalue of the dilation Laplacians have all the same sign, and provide directly a ranking score of its nodes. The ranking method, phrased in terms of a group synchronization problem, is applied to artificial and real data, and its performance is compared with other ranking strategies. A main feature of this approach is the presence of a deformation parameter enabling the emphasis of the top-$k$ objects in the ranking. Furthermore, inspired by these results, a family of random walks  interpolating between the undirected random walk and the Pagerank random walk is also proposed.
\end{abstract}
\maketitle

\section{Introduction}
The combinatorial Laplacian has been used extensively in applied mathematics and machine learning over the years. Indeed, discrete Laplacians are of fundamental importance whenever the available data can be organized as a graph. In the case of undirected graphs, the combinatorial Laplacian is directly used, for instance, for clustering \cite{ChungBook}, data visualization \cite{Coifman05geometricdiffusions}, or semisupervised learning \citep{BelkinManifold}. 
This paper deals with directed networks, and more specifically, graphs arising from a set of pairwise comparisons between objects, i.e., if object $i$ is preferred to object $j$ ($i\succ j$), there is a directed edge $i\to j$. Applications of the Laplacian introduced in this paper to the problem of ranking from a set of pairwise comparisons are considered.

Ranking from pairwise comparisons was addressed in the literature as a least-squares problem, from a random walk viewpoint, or in terms of a spectral problem. Let us mention some of these approaches.  Hodgerank \cite{YaoRanking} deals with the issue of finding a consistent ranking in the least square sense, by taking advantage of a discrete Hodge theory. A random walk-based approach was recently presented in \cite{Negahban}.
Alternatively, spectral ranking algorithms were also proposed recently in \cite{FogelSpectralRanking,Vigna}. More precisely, in \cite{FogelSpectralRanking}, the ranking is obtained by computing the second least eigenvector of a combinatorial Laplacian associated to a similarity matrix, yielding the algorithm Serialrank. Hence, because the so-called Fiedler vector has both positive and negative elements, the choice of the sign of this vector has to be done in order to minimize the number of upsets, i.e., the number of times an object is preferred to another object with higher ranking. The main assumption used to build this similarity matrix is that two objects preferred to the same objects are similar. Therefore, this algorithm relies on the hypothesis that many repeated comparisons are available for each object.
Another spectral ranking method, \mbox{Sync-Rank}, proposed in the paper \cite{Cucuringu,YuStella}, is based on the computation of the complex phases of an eigenvector of a Hermitian Matrix, related to the so-called connection Laplacian for $SO(2)$ (a semi-definite programming method is also discussed in the same paper). In the latter paper, a rotation is associated to each comparison, so that the objects are positioned by the algorithm on the circle. Then, the ranking is obtained by finding the cyclic permutation of the objects minimizing the number of upsets.

The method proposed in this work can be understood as an algorithm for finding a unique positive eigenvector of the dilation Laplacian with minimal eigenvalue. In fact, this spectral problem is a constrained least-squares problem which can be phrased as a constrained synchronization of dilations.  A deformation of the combinatorial Laplacian is introduced, which appears as a particular case of a construction generalizing many existing discrete Laplacians.  Two novel deformed Laplacians are proposed: the \emph{infinitesimal dilation Laplacian} and the  \emph{dilation Laplacian}, both depending on a parameter $g$ controlling their deformation from the combinatorial Laplacian. In other words, each of them is a one parameter family of Laplacians. The parameter $g$ can be interpreted as a coupling constant in the jargon of quantum physics.  
In the case of a weak coupling constant $g$, we show that the score given by the eigenvector with the smallest eigenvalue of the dilation Laplacian yields a ranking which coincides with the ranking obtained in the least-squares approach, up to a $\mathcal{O}(g^2)$ correction. For a larger value of $g$, the score ranks more accurately the objects in the top part of the ranking.
In this paper, we compare empirically the ranking obtained using the dilation Laplacian with other methods with respect to two viewpoints:
\begin{itemize}
\item[-]\emph{Criterion 1}. Given a known ranking, the efficiency of the method is assessed by evaluating how well the ranking is retrieved from missing or corrupted pairwise information. In this context, we assume that a true ranking exists. The accuracy of the ranking is then measured by using Kendall's $\tau$-distance with the ground truth. 
\item[-]\emph{Criterion 2}. In a practical perspective, where often  no ``true'' ranking  exists, and given only pairwise information, the accuracy of the method is evaluated by counting the number of disagreements or  ``upsets'' between the ranking of the objects and the known comparisons in the top part of this ranking. For instance, a low number of upsets is of course desirable in the context of sport tournaments.
\end{itemize}

Let us outline the organization and the main contributions of this paper.
In section~\ref{sec:DilationLaplacians}, the least-squares ranking problem (Hodgerank) is firstly reviewed and used as a motivation for the definition of the dilation Laplacians, yielding a connection between additive and multiplicative pairwise comparisons. After the definition of the dilation Laplacians, a spectral algorithm for ranking from pairwise comparisons is proposed, where the ranking score has a clear meaning, i.e., if the comparisons are seen as ``exchange rates'', then the score can be intuitively understood as a ``universal currency''.

As it was already mentioned, the dilation ranking is very similar to a least-squares problem and can be related to a synchronization problem over the dilation group as explained in Section~\ref{sec:AbsErrors}, where constraints on the sum of absolute errors are also given in the form of a Cheeger type inequality. Furthermore, for a larger value of the deformation parameter, a ranking score with fewer upsets in the top part is obtained, as illustrated for instance in Figure~\ref{Fig:League}. Furthermore, in artificial data sets with missing comparisons, we observe in the simulations the same effect which can be understood from Lemma~\ref{LemmaMajorization} in Section~\ref{sec:AbsErrors}.

The normalization of the combinatorial Laplacian is a customary technique in view of applications. Hence, normalized dilation Laplacians are also discussed in Section~\ref{sec:Robust} where a family of random walks interpolating between the undirected and Pagerank random walks on directed graphs is also proposed. 

Then, a generalization of a series of existing Laplacians is put forward in Section~\ref{sec:DeformedLaplacians} and further connections with several existing Laplacian are outlined in~\ref{AppDefLap}.
Finally, numerical simulations are presented in Section~\ref{sec:NumSim}, whereas all the proofs are given in~\ref{AppProof}.

\section{Dilation Laplacians and ranking from pairwise comparisons \label{sec:DilationLaplacians}}
\subsection{Preliminaries}
Consider a connected graph $\mathcal{G} = (V,E)$ with a set of $N$ nodes $V$ and a set of oriented edges $E$. For simplicity, we identify the set of nodes $V$ to the set of integers $\{1,\dots, N\}$.
An edge corresponds to an unordered pair of nodes $\{i,j\}\in E_u$, and an oriented edge to an ordered pair of nodes $e = [i,j]\in E$. The edge with the opposite orientation is denoted by $\bar{e} = [j,i]$. A symmetric weight $w_{ij}> 0$ is also associated to any pair of nodes $i$ and $j \in V$ connected by an edge, whereas $w_{ij}=0$ otherwise. It is common to consider the weight matrix $W$ with $w_{ij}$ as matrix elements, as well as the diagonal degree matrix given by $D_{ii} = \sum_{j\in V}w_{ij}$, which represents the volume taken by each node in the graph. Furthermore, the volume of the weighted graph is commonly defined as the sum of the degrees: $\vol(\mathcal{G}) =\sum_{i,j\in V}w_{ij}$.

We will also consider the skew-symmetric functions of the oriented edges $\Omega_E = \{X:E\to \mathbb{R} | X(e) = -X(\bar{e})\}$ and identify it with the set of skew-symmetric matrices of $\mathbb{R}^{N\times N}$ such that $X_{ij} =- X_{ji}$ if $[i,j]\in E$ and $X_{ij} =0$ otherwise. In the context of ranking~\cite{Tran,ELSNER}, both additive pairwise measurements \[\{a_{ij}\in \mathbb{R}| a_{ji}=-a_{ij} , \ {\rm for \ all}\  [i,j]\in E\},\] and multiplicative pairwise measurements given by \[\{s_{ij}\in \mathbb{R}_\star^+| s_{ji}=1/s_{ij} , \ {\rm for \ all}\ [i,j]\in E\},\]
will be discussed. Notice that, given multiplicative comparisons $s_{ji}=1/s_{ij}$, we can find a skew symmetric matrix $a$ such that $a_{ij}=-a_{ji}$, for all $[i,j]\in E$ such that $s_{ij} = \exp (a_{ij})$.
\subsection{Motivations, Hodgerank and Dilation Laplacians}

In order to develop an intuition, the case of ordinal comparisons is first considered, i.e., for each comparison $[i,j]\in E$, we have $w_{ij} = 1$ and 
\[
a_{ij} =\begin{cases}
 1 & \text{ if } i\succ j,\\
-1 & \text{ if } j\succ i,\\
 0 & \text{ if } i\sim j.
\end{cases}
\]
Therefore, the edge flow $a$ is skew-symmetric under a change of orientation of the edges.  In the case of cardinal comparisons, the definition is similar except that $a_{ij}\in\mathbb{R}$ rather than being valued in $\{-1,0,1\}$.
The least-squares ranking problem discussed in \cite{YaoRanking} consists in finding the score $f\in \mathbb{R}^{N}$ solving
\begin{equation}
\mathop{\rm minimize}_{f\in \mathbb{R}^N}\frac{1}{2}\sum_{\{i,j|[i,j]\in E\}}w_{ij}\Big(a_{ij} - [{\rm d}f]_{ij} \Big)^{2},\label{Hodegerank}
\end{equation}
where 
\begin{equation}
[{\rm d}f]_{ij} = f_j-f_i,\label{eq:discretegradient}
\end{equation}
is a discrete gradient, which can be implemented thanks to the incidence matrix ${\rm d}\in \mathbb{R}^{|E|\times N}$. By construction, the objective function includes one term for each undirected edge, that is, for each comparison. Furthermore, each node in the graph is treated  in the same way and each comparison has \emph{a priori} the same weight. Considering now the solution of this problem, Theorem $3$ of \cite{YaoRanking} shows that solutions to (\ref{Hodegerank}) have to satisfy the linear system
\begin{equation}
\mathcal{L}_0f = -{\rm div}\ a,\label{eq:LinearSystem}
\end{equation} 
with the combinatorial Laplacian $\mathcal{L}_0 = {\rm d}^\intercal{\rm d} = D-W$. The divergence operator is given in terms of the adjoint (or transposed) of the incidence matrix~(\ref{eq:discretegradient}) as follows:
$[- {\rm div} X]_i = [{\rm d}^\intercal X]_i = \sum_{j\in V}w_{ij} X_{ji}$. Obviously, there are infinitely many solutions to (\ref{eq:LinearSystem}) yielding equivalent ranking scores.

As shown in \cite{YaoRanking}, the minimal norm solution of the least-squares ranking problem is simply
\begin{equation}
f^\star_H = -\mathcal{L}_0^{\dagger}{\rm div}\ a,\label{HodegerankSolution}
\end{equation} 
where $\mathcal{L}_0^{\dagger}$ is the Moore-Penrose pseudo-inverse of $\mathcal{L}_0$. 
The divergence of $a$ which counts the wins minus the losses,
\begin{equation}
w^{\rm diff}_i = [{\rm div}\ a]_i,\label{eq:wdiff}
\end{equation}
is in fact the number of times $i\in V$ is better than other objects in all known comparisons minus the number of times $i$ is defeated in all available comparisons. Indeed, the point score $w^{\rm diff}_i$ already provides us with a basic ranking score but does not yet encompass the complete structure of the comparisons incorporated in the graph $\mathcal{G}$. This remark motivates the optimization problem for obtaining a ranking score vector described in Proposition \ref{PropConvexOptim}.

Denote by $v_{0}$ the normalized eigenvector of eigenvalue zero of the combinatorial Laplacian $\mathcal{L}_{0}$, that is $[v_{0}]_i = 1/\sqrt{N}$ for all $i\in V$. For convenience, we further define the diagonal matrix
$\mathcal{W}^{\rm diff} = \diag(w^{\rm diff})$. 
\begin{myprop}[Smoothing of a point score]\label{PropConvexOptim}
The minimization problem
\begin{equation}
\mathop{\rm minimize}_{f\in \mathbb{R}^N}\frac{1}{2} f^\intercal\mathcal{L}_{0} f - g f^\intercal\mathcal{W}^{\rm diff} v_{0},~{\rm s.t.}~ v_{0}^\intercal f = 0,\label{ConvexOptimizationProblem2}
\end{equation}
has a unique solution given by $f^{\star} = g\mathcal{L}_{0}^{\dagger}\mathcal{W}^{\rm diff}v_{0}$.
\end{myprop}
The objective function~(\ref{ConvexOptimizationProblem2}) includes a (smoothing) regularization term with the combinatorial Laplacian and an inner product $-f^\intercal\mathcal{W}^{\rm diff} v_{0} = -f^\intercal w^{\rm diff}/\sqrt{N}$ which enforces the alignment of the score $f$ with the point score $w^{\rm diff}$. Indeed, the first term in~(\ref{ConvexOptimizationProblem2}) is a sum of squared differences
\[f^\intercal \mathcal{L}_{0} f = \frac{1}{2}\sum_{i,j\in V}w_{ij}(f_i-f_j)^2,\]
and therefore, it tends to smoothen the solution since it promotes the components of $f^\star$ along the eigenvectors of small eigenvalues which corresponds to the most constant functions on the graph. This feature is fundamental in the context of manifold learning~\cite{BelkinManifold}. The constraint $v_{0}^\intercal f = 0$ selects a unique solution since any $f+\alpha v_{0}$ for all $\alpha\in\mathbb{R}$ would otherwise be an equivalent solution. Notice also that $v_0^\intercal w_{\rm diff} =0$.

More importantly, the minimum norm solution of Hodgerank (\ref{HodegerankSolution}) gives a ranking score proportional to the score obtained as the unique solution of a quadratic program given in Proposition \ref{PropConvexOptim}. Numerically, the ranking score of Hodgerank is then obtained as the solution of a linear system up to a shift by a constant vector. 

The formulation given in Proposition \ref{PropConvexOptim} renders explicit the role of the combinatorial Laplacian in Hodgerank and introduces a parameter $g>0$ whose role is only to scale the value of the score. We firstly notice that the same ranking score can be obtained (in approximation) from an eigenvector problem involving a deformation of the combinatorial Laplacian $\mathcal{L}_{g}$, where $g$ is now seen as a deformation parameter. A major feature of this matrix is that it is built in order to have the power series expansion
$
\mathcal{L}_{g} = \mathcal{L}_{0} - g\mathcal{W}^{\rm diff} + \mathcal{O}(g^2).
$
Let us now introduce the dilation Laplacian which is a symmetric, positive semi-definite matrix satisfying the previous power series expansion.
\begin{mydef}[Dilation Laplacian]\label{Def1}
Let $g>0$. The dilation Laplacian $\mathcal{L}_{g} = \mathcal{L}_{g}(a,W)$ is given by
\begin{equation}
\big[\mathcal{L}_{g}(a,W) v\big]_i \triangleq \big[(\mathcal{D}_{g}-W )v\big]_i  = \sum_{j \in V}w_{ij}\big(e^{g a_{ji}}v_i-v_j\big),\label{DilationLaplacian}
\end{equation}
with the diagonal matrix $\big[\mathcal{D}_{g}\big]_{ii} =\sum_{j\in V }w_{ij}e^{g a_{ji}}$, for all $i\in V$.
\end{mydef}
For convenience, in the absence of ambiguity we will write $\mathcal{L}_{g}$ rather than $\mathcal{L}_{g}(a,W)$. In contrast with the combinatorial Laplacian, notice that $\mathcal{L}_{g}$ is positive semi-definite but it is in general \emph{not}  diagonally dominant as explained in the next section (see, e.g.(\ref{eq:ObjReOrg})).
A trivial consequence of Definition~\ref{Def1} given in Lemma~\ref{LemmaMoorePenrose} is that the least eigenvector of the dilation Laplacian gives the same ranking score as Hodgerank up to an irrelevant constant and $\mathcal{O}(g^2)$ corrections.
\begin{mylemma}[Perturbative expansion]\label{LemmaMoorePenrose}
The eigenvector of the dilation Laplacian of the smallest eigenvalue satisfying $\mathcal{L}_{g}v^{(g)}_{0} = \lambda^{(g)}_0 v^{(g)}_{0}$ admits the following series expansion in powers of $g$,
\begin{equation}
v^{(g)}_{0} = v_{0} +g \mathcal{L}_{0}^{\dagger}\mathcal{W}^{\rm diff}v_{0} + \mathcal{O}(g^2).\label{eq:SeriesEigenVec}
\end{equation}
\end{mylemma}
The dilation Laplacian~(\ref{DilationLaplacian}) is one possible choice of deformed Laplacian with the desirable property of positive definiteness and symmetry.
Clearly, it is possible to define another symmetric positive semi-definite Laplacian which includes only $g^2$ terms and no higher powers, while its least eigenvector also satisfies~(\ref{eq:SeriesEigenVec}).
\begin{mydef}[Infinitesimal dilation Laplacian]
Let $0<g<1$.  The infinitesimal dilation Laplacian is defined by
\begin{equation}
\big[{L}^{\rm inf}_{g}v\big]_{i}
 \triangleq\sum_{j\in V}w_{ij}\Big\{\big(1+ga_{ij}/2\big)^2v_i -\big(1-(ga_{ij}/2)^2\big) v_j\Big\},\label{eq:InfLap}
\end{equation}
for all $v\in \mathbb{R}^N$ and $i\in V$.
\end{mydef}
Notwithstanding, the dilation Laplacian has an interesting additional property whenever $a_{ij} = h_i-h_j$, i.e, the data provide us with a globally consistent ranking given by the score function $h\in \mathbb{R}^N$, which is associated to a discrete potential $-h$. Indeed, when the ranking can be exactly found because $a = -{\rm d}h$ with ${\rm d}$ defined in~(\ref{eq:discretegradient}), the output of Hodgerank is the solution of a linear system yielding the vector $h$, while the least eigenvector of the dilation Laplacian provides us with $\exp(g h)$, where the exponential is taken element-wise. This fact is summarized in Proposition~\ref{Prop2}.
\begin{myprop}[Existence of a zero eigenvalue]\label{Prop2}
The dilation Laplacian (\ref{DilationLaplacian}) has an eigenvector with zero eigenvalue if and only if $a = -{\rm d}h$, i.e., if there exists $h\in\mathbb{R}^{N}$ such that  \mbox{$a_{ij} = h_i-h_j$} for all $[i,j]\in E$. Then, this eigenvector is given by
$v^{(g)}_{0} =  c \times \exp(g h)$, where $c\in \mathbb{R}$.
\end{myprop}
Let us discuss the connection between the dilation Laplacian and Hodgerank when the edge flow $a$ has inconsistencies.
The decomposition $a_{ij} = h_i-h_j+\epsilon_{ij}$ can be chosen to be the Hodge decomposition as given by the minimal norm solution of Hodgerank, where $h$ is the solution to~(\ref{HodegerankSolution}) and $\epsilon_{ij} =-\epsilon_{ji}$ is an inconsistent edge flow. Then, Lemma~\ref{LemmaInconsistent} states that the dilation Laplacian associated with $a_{ij}$ is related to the dilation Laplacian associated with the inconsistent part of $a$ and with rescaled weights $w_{ij}e^{g(h_i+h_j)}$.
\begin{mylemma} \label{LemmaInconsistent} Let $a\in \Omega_E$ be an edge flow satisfying the Hodge decomposition $a_{ij} = h_i-h_j+\epsilon_{ij}$  with $\epsilon_{ij} =-\epsilon_{ji}$ for all  $[i,j]\in E$. Then, we have
$
\diag(e^{gh})\mathcal{L}_g(a,W)\diag(e^{gh}) = \mathcal{L}_g\big(\epsilon,\diag(e^{gh})W\diag(e^{gh})\big),
$
with $g\in \mathbb{R}$.
\end{mylemma}
The renormalization of the weight matrix $\diag(e^{gh})W\diag(e^{gh})$ gives an increased importance to the edges $\{i,j\}$ involving comparisons between alternatives $i,j\in V$ with a large score $h$ as given by Hodgerank.

We provide now a generalization of the dilation and infinitesimal dilation Laplacians which has the feature that its least eigenvector has no sign change. Indeed, given the set of strictly positive reals $\{s_{ij} | \ s_{ij}>0 \ {\rm for \ all}\ [i,j]\in E\}$, we define the symmetric positive semi-definite matrix $L_{s}$ given by
\begin{equation}
\big[L_{s}v\big]_i\triangleq\sum_{j\in V}w_{ij}s_{ji}\Big(s_{ji}v_i-s_{ij}v_j\Big),\label{DeformedLaplacian}
\end{equation}
for all $i\in V$. Clearly, the matrix of~(\ref{DeformedLaplacian}) is in general not only the sum of the combinatorial Laplacian with a diagonal matrix. In order to adapt this definition in the context of ranking, the key idea of the construction is to choose a positive real number for each edge such that $s_{ij}>s_{ji}$ if the edge is $i\to j$.
In particular, if we choose $s_{ij}= 1+ga_{ij}/2$ and $0<g<1$, the deformed Laplacian~(\ref{DeformedLaplacian}) is the infinitesimal dilation Laplacian~(\ref{eq:InfLap}). Furthermore, in order to have the property $s_{ij} = 1/s_{ji}$, we can also consider the choice $s_{ij}= \exp( ga_{ij}/2)$, yielding the dilation Laplacian~(\ref{DilationLaplacian}). 

In the literature, several methods propose to obtain a ranking score from the top eigenvector of a matrix with non-negative entries. By Frobenius-Perron theorem, the entries of this top eigenvector have all the same sign. In view of this remark, the construction (\ref{DeformedLaplacian}) also has the attractive feature that the entries of its least eigenvector always have the same sign as stated in Theorem \ref{Thm1}. The proof of this property does not depend on Frobenius-Perron theorem.
\begin{mythm}[Positivity of the least eigenvector]\label{Thm1}
Let $\{s_{ij} | \ s_{ij}>0 \ {\rm for \ all}\ [i,j]\in E\}$ be a set of positive reals indexed by the oriented edges $[i,j]\in E$ in the connected graph $\mathcal{G}$ and let $L_s$ be the deformed Laplacian given in~(\ref{DeformedLaplacian}). Then, there exists a unique normalized eigenvector $v_0\in \mathbb{R}^N$ of $L_s$  associated to its smallest eigenvalue such that $[v_0]_i>0$ for all $i\in V$.
\end{mythm}
Theorem~\ref{Thm1} is useful since it guarantees that the components of the least eigenvector yield an unambiguous ranking score. In particular, if the deformed Laplacian $L_s$ is taken to be the dilation Laplacian~(\ref{DilationLaplacian}), it can be decomposed as $\mathcal{L}_g = \mathcal{L}_0 + V^{(g)}$, with $V^{(g)}$ diagonal, yielding a discrete analogue of the Schr\"odinger equation (see~\ref{AppWitten}).  Then, Theorem~\ref{Thm1} can be interpreted as a discrete version of the classical result stating that the lowest energy state of the Schr\"odinger equation in a confining potential has no sign changes.
Trivially, the ranking scores $v\in \mathbb{R}^{N}$ and $\alpha v \in \mathbb{R}^{N}$ with $\alpha>0$ yield the same ranking.

\subsection{Dilation ranking}
To sum up, the ranking of the compared objects are computed thanks to Algorithm~\ref{Alg1}. From the practical perspective, a natural question concerns the choice of the value for the deformation parameter $g$. Consider a simple example, illustrated in Figure \ref{Fig:Line}, of a directed line graph. In this example, the dilation ranking gives 
\begin{equation}
[v^{(g)}_0]_1 = e^{g}[v^{(g)}_0]_2 = \dots = e^{g(N-1)}[v^{(g)}_0]_N.\label{eq:Line}
\end{equation}
The scores of the first and the $k$-th objects are proportional, the proportionality factor being $\exp ( g(k-1))$.
Therefore, we require $g = \log(c)/(k-1)$, where $c>0$ and $1<k\leq N$ are constants defined by the user. Intuitively, $c$ is the factor between the first and the $k$-th objects in the ranking in the case of the line graph of Figure \ref{Fig:Line}. The values that we recommend are summarized in Table~\ref{Tab:Choiceg} and in particular, we consider $g>0$.

\begin{algorithm}[h]
\caption{Dilation Ranking \label{Alg1}}
\begin{algorithmic}[1]
\Require Pairwise comparisons $E$ and $\{a_{ij}\}_{[i,j]\in E}$; a real $g$ (see Table~\ref{Tab:Choiceg});
\State compute the dilation Laplacian matrix $\mathcal{L}_g$ given in (\ref{DilationLaplacian});
\State compute the least eigenvector $v_0^{(g)}$ of $\mathcal{L}_g$, such that $\|v_0^{(g)}\|_2 = 1$;
\Return the score vector ${\rm sign}(g)|v_0^{(g)}|$ and the ranking obtained by sorting $\{{\rm sign}(g)|[v_0^{(g)}]_i|\}_{i\in V}$.
\end{algorithmic}
\end{algorithm}
\begin{table}[h]
\centering
\begin{tabular}{c c}
Parameter & Feature\\
\hline
\rule{0pt}{3ex}$g =0$ & No ranking, constant score: $v_0 = 1/\sqrt{N}$.\\
$g \ll 1$ & Hodgerank: $v_0^{(g)} = v_0 +g \mathcal{L}_{0}^{\dagger}\mathcal{W}^{\rm diff}v_0 + \mathcal{O}(g^2)$.\\
$g\approx 0.1/(N-1)$ & Recommended value.\\
$g> 0.1/(N-1)$ & Fewer upsets in the top-$k$ objects.\\
\hline
\end{tabular}
\caption{Summary of the empirical effects of the choice of $g$ on the ranking given by $\mathcal{L}_g$ for ordinal comparisons.\label{Tab:Choiceg}}
\end{table} 
\begin{figure}[h]
\centering
\begin{tikzpicture}[auto, thick]
\foreach \place/\name in {{(1,1)/a}, {(2,1)/b}, {(3,1)/c}, {(5,1)/d}}
    \node[Node2] (\name) at \place {};
\foreach \source/\dest in {a/b,b/c}
    \path[->,very thick] (\source) edge (\dest);
\draw (1,1) node{$1$};
\draw (2,1) node{$2$};
\draw (3,1) node{$3$};
\draw (5,1) node{$N$};
\draw[->,very thick] (3+0.3,1) --(3+0.6,1);
\draw[->,very thick] (5-0.6,1) --(5-0.3,1);
\draw (4,1) node{$\cdots$};
\end{tikzpicture}
\caption{Line graph with $N$ vertices.\label{Fig:Line}}
\end{figure}
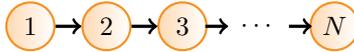
The dilation Laplacian can be built from pairwise comparisons, both in the case of ordinal and cardinal comparisons: 
\begin{itemize}
\item Ordinal comparisons: for each known comparison $[i,j]\in E$, define $w_{ij} = 1$. Furthermore, if $i\succ j$, define  $a_{ij} = 1 = -a_{ji}$ and $a_{ij} = 0$ if $i\sim j$. If there is no comparison available between $i$ and $j\in V$, define $w_{ij} =0$. Choose the Laplacian $\mathcal{L}_g$ with a deformation parameter $g$ choosing according to the criterion given in Table~\ref{Tab:Choiceg}.
\item Cardinal comparisons: for each known cardinal measurement corresponding to $[i,j]\in E$, define $w_{ij} = 1$. If the data provides us with an ``exchange rate'' $s_{ij}>0$ between $i$ and $j\in V$, define simply $\exp (a_{ij}) = s_{ij}$ (where $g = 1$) and use the dilation Laplacian $\mathcal{L}_g$. 
\end{itemize}

\subsection{Remarks about the information contained in the pairwise measurements}
In general, given the graphical nature of the available data, it can be expected that the quality of the retrieved ranking depends on the connectivity of the graph, i.e., how well the alternatives are compared with each other. Indeed, if the graph $\mathcal{G}$ has a bottleneck between two more connected regions, the retrieval problem will be made more difficult due to the lack of comparisons between the well-connected regions. In~\cite{JMLR:osting}(Proposition 4.1), it is shown that the variance of the least squares estimator for the ranking, called here Hodgerank, is proportional to the Moore-Penrose pseudo-inverse of the Laplacian.  This result holds for the data model given by
$
a_{ij} = h_j-h_i+X_{ij},
$
for each comparison $\{i,j\}\in E_u$ where $X_{ij}$ is a random variable such that $\mathbb{E}(X_{ij}) =0$ and $\Var(X_{ij}) = \sigma^2/w_{ij}$. More precisely, if we denote the least-squares estimator by $\hat{f}_H$, then the covariance matrix of this random vector is
$\Var(\hat{f}_H) = \sigma^2 \mathcal{L}_0^\dagger
$. Hence, the variance of the estimator can be made smaller by choosing well-connected measurement graphs so that their spectral gap, given by the second least eigenvalue of $\mathcal{L}_0$, is large. It is also proved in~\cite{JMLR:osting,Osting2016540,Howard} that the Fisher information matrix is proportional to the combinatorial Laplacian $\mathcal{L}_0$. Furthermore, the pseudo-inverse of the combinatorial Laplacian also has a central role in the context of the synchronization of rotations where it gives the Cram\'er-Rao lower bound on the variance of unbiased estimators~\cite{BoumalSyncRot}.

\section{Constraints on the squared and absolute errors\label{sec:AbsErrors}}
In this section, we consider only the case of the dilation Laplacian corresponding to the choice $s_{ij}= \exp (ga_{ij}/2)$ in (\ref{DeformedLaplacian}) and discuss the following aspects:
What is the connection between the eigenvector problem and the dilation group $\mathbb{R}_\star^+ = \{s\in \mathbb{R}|s>0\}$?
The eigenvector problem minimizes a sum of squared residuals. What can we tell about the sum of absolute residuals?

A possible interpretation of the ranking problem involves the group $\Gamma = (\mathbb{R}_\star^+,\times)$ of strictly positive reals for the multiplication, acting on the Hilbert space $\Gamma = \mathbb{R}$ with the scalar multiplication, i.e. for $\gamma = s \in \mathbb{R}_\star^+$ and $v = x\in \mathbb{R}$, we have the group action $\gamma\cdot v = sx$. Naturally, the group element mapped to an edge is merely $\gamma_{ij} = s_{ij}\in \mathbb{R}_\star^+$ with the property $s_{ji} = s_{ij}^{-1}$.
We choose $v_+ \in \mathbb{R}$ with $v_+>0$, so that the orbit of this element is $\Gamma\cdot v_+ = \mathbb{R}_\star^+$. Since $\mathbb{R}_\star^+$ is a non-compact group, we have to complement the problem with an additional constraint. The group synchronization problem is rephrased as the least-squares minimization problem
\begin{equation}
\min_{v\in (\mathbb{R}_\star^+)^N}\frac{1}{2}\sum_{i,j=1}^{N}w_{ij}s_{ji}\Big( v_i  -s_{ij}v_j\Big)^{2}, \ {\rm s.t.}\ \sum_{i=1}^{N} v_i^{2} = 1,\label{ObjectiveSynchroS}
\end{equation} 
where the objective function is a sum of squared residuals involving exactly one term for each edge $\{i,j\}$, since we have
$
s_{ji}\big(v_i  - s_{ij}v_j\big)^{2} = s_{ij}\big(v_j  - s_{ji}v_i\big)^{2}.
$
Recall that we can find a skew symmetric matrix $a$ such that $a_{ij}=-a_{ji}$, for all $[i,j]\in E$ such that $s_{ij} = \exp (g a_{ij})$.
Naturally, the quadratic form in the objective~(\ref{ObjectiveSynchroS}) is associated to the dilation Laplacian,
motivating the connection with the dilation group $\mathbb{R}_\star^+$.
Classically, the lowest eigenvector is obtained as the solution of the minimization of
\begin{equation}
v^\intercal\mathcal{L}_g v = \frac{1}{2}\sum_{i,j=1}^{N}w_{ij}s_{ji}\Big( v_i  -s_{ij}v_j\Big)^{2},\label{eq:ObjReOrg}
\end{equation} 
over the vectors $v\in \mathbb{R}^{N}$ subject to the constraint $v^\intercal v = 1$.
\subsection{Frustration and group potentials}
The least-squares ranking problem~(\ref{Hodegerank}), for additive pairwise comparisons, involves an objective function of $f\in\mathbb{R}^{N}$ which is invariant under a constant shift: $f_i\mapsto f_i+ c$ for all $i\in V$. For the case of multiplicative pairwise comparisons, it is necessary to define an objective function of the score vector $v\in (\mathbb{R}_\star^+)^{N}$ with an invariance with respect to a constant scaling $v\mapsto \alpha v$, with $\alpha>0$.
Inspired by \cite{Bandeira_OdCheeger}, the vector $v\in (\mathbb{R}_\star^+)^{N}$ is called a group potential and we introduce its frustration related to the corresponding Rayleigh quotient, which is interpreted as a normalized sum of squared residuals.
\begin{mydef}[$\ell^2$-frustration]
The $\ell^2$-frustration of a $\mathbb{R}_\star^+$-potential $v\in (\mathbb{R}_\star^+)^{N}$ is 
\begin{equation}
\eta^{[2]}_{a,W}(v) \triangleq\frac{1}{2} \frac{\sum_{i,j=1}^{N}w_{ij}s_{ji}\big(v_i  - s_{ij}v_j\big)^{2}}{ \sum_{i=1}^{N} v_i^{2}}.\label{eq:frustrationpotential}
\end{equation}
Where there is no ambiguity, we will simply write $\eta^{[2]}(v) = \eta^{[2]}_{a,W}(v)$.
\end{mydef}
Because of the connection with the dilation Laplacian, Proposition~\ref{Prop2} can be simply rephrased as a property of the frustration. Indeed, if the data provides us with $a_{ij} = h_i-h_j$ or $s_{ij}  = v_i/v_j$, then the frustration vanishes and conversely. Since, for a given $g>0$, we can write $s_{ij} = \exp (g a_{ij})$ for all oriented edge $[i,j]\in E$ and $v_i = \exp(g h_i)$ for all $i\in V$, the property $s_{ij}  = v_i/v_j$ for all $[i,j]\in E$ is equivalent to $a_{ij} = h_i-h_j$ for all $[i,j]\in E$. 
\begin{mylemma}\label{Lem:frustration}
The $\ell^2$-frustration vanishes for some $v\in (\mathbb{R}_\star^+)^N$, i.e., $\eta^{[2]}_{a,W}(v) = 0$ if and only if $s_{ij}v_j  = v_i$, with $s_{ij} = 1/s_{ji}$, for all oriented edge $[i,j]\in E$.
\end{mylemma}
Hence, for a given set of pairwise measurements and its corresponding graph, we define its $\ell^2$-frustration constant as the minimal frustration achieved over all possible choices of group potentials.
\begin{mydef}[$\ell^2$-frustration constant]
The $\ell^2$-frustration constant is defined by
\begin{equation}
\eta^{[2]}_{a,W} = \min\big\{\eta^{[2]}_{a,W}(v)\ | \ v\in (\mathbb{R}_\star^+)^N\big\},\label{eq:frustration}
\end{equation}
where the frustration of a $\mathbb{R}_\star^+$-potential is given in (\ref{eq:frustrationpotential}).
\end{mydef} 
We can now show that computing the $\ell^2$-frustration constant is equivalent to the computation of the least eigenvalue of the dilation Laplacian. Indeed, a spectral relaxation of this minimization problem is simply obtained by
\begin{equation}
\mathop{\rm minimize}_{\{v\in \mathbb{R}^{N}|v\neq 0\}}\eta^{[2]}_{a,W}(v),\label{eq:MinimizeFrustration}
\end{equation} 
where the feasible set has been enlarged to $\mathbb{R}^{N}_\star$.
Looking for a solution of this problem is equivalent to the computation of the lowest eigenvector of the  dilation Laplacian (\ref{DilationLaplacian}) because the objective
\[\eta^{[2]}_{a,W}(v) = \frac{v^\intercal \mathcal{L}_g(a,W) v}{v^\intercal v},
\]
is the Rayleigh quotient.
A straightforward consequence is that the frustration constant is given by the smallest eigenvalue of the dilation Laplacian. Thus, Proposition \ref{CheegerEqual} states that the smallest eigenvalue of the dilation Laplacian determines the frustration of the $\mathbb{R}^+_\star$-potential.
\begin{myprop}\label{CheegerEqual}
Let $\lambda_0^{(g)}\geq 0$ be the smallest eigenvalue of $\mathcal{L}_g$. We have
$\lambda_0^{(g)} = \eta^{[2]}_{a,W}$.
\end{myprop}
Again, ranking from pairwise comparisons is conveniently formulated in this context. Indeed, if $i$ and $j$ are compared, and for instance $i$ is better than $j$, then $s_{ij}>1$ (otherwise $s_{ij}<1$) and we wish the positive score $v_i$ to be larger than the positive score $v_j$. Hence, we ask $v_i \sim s_{ij} v_j $, which can be formulated more precisely as a synchronization of dilations.

The least-squares methods involve a sum of squared residuals which are known to amplify errors compared to a sum of absolute values of the residuals.
Actually, the sum of the absolute errors can also be considered in order to quantitatively evaluate the frustration of a group potential. Therefore, we now introduce  the $\ell^1$-frustration which gives less importance to larger residuals with respect to the $\ell^2$-frustration.
\begin{mydef}[$\ell^1$-frustration]
The $\ell^1$-frustration of a $\mathbb{R}_\star^+$-potential $v\in (\mathbb{R}_\star^+)^N$ is 
\begin{equation}
\eta^{[1]}_{a,W}(v) \triangleq\frac{1}{2} \frac{\sum_{i,j=1}^{N}w_{ij}|s^{1/2}_{ij}v_j  - s^{-1/2}_{ij}v_i|}{ \sum_{i=1}^{N} v_i}.\label{eq:frustrationl1potential}
\end{equation}
The  $\ell^1$-frustration constant is defined by
$\eta_{a,W}^{[1]} = \min\big\{\eta^{[1]}_{a,W}(v)\ |\ v\in (\mathbb{R}_\star^+)^N\big\}$.
\end{mydef}

By computing the lowest eigenvector of the dilation Laplacian (corresponding to a $\ell^2$-frustration constant), we obtain an approximate $\mathbb{R}_\star^+$-potential providing an upper bound on the $\ell^1$-frustration constant, interpreted as the performance of the spectral method.
We also provide a lower bound for the $\ell^1$-frustration constant. By analogy with~\cite{Bandeira_OdCheeger}, we call these bounds Cheeger type inequalities because of the resemblance with the Cheeger inequality relating the normalized cut of a graph and the second least eigenvalue of the combinatorial Laplacian.
\begin{mythm}[Cheeger type inequality]\label{Thm2}
Let $a\in \Omega_E$ and $s_{ij} = \exp(ga_{ij}/2)$ for all $[i,j]\in E$ with $g>0$. We have\begin{equation}
\frac{1}{N(1+\sqrt{s_{\rm max}})}\frac{\lambda^{(g)}_0}{{\rm vol}(\mathcal{G})}\leq \frac{\eta_{a,W}^{[1]}}{{\rm vol}(\mathcal{G})}\leq \sqrt{\frac{\lambda^{(g)}_0}{2{\rm vol}(\mathcal{G})}},\label{eq:frustrationCheeger}
\end{equation}
where $1\leq s_{\rm max} = \max\{s_{ij}\ | \ \{i,j\}\in E_u \}$.
\end{mythm}
The proof exploits the techniques used in \cite{Bandeira_OdCheeger}.
Being linear in the weights $w_{ij}$, the frustration can potentially rise if the weights increase or if the number of comparison increases.
Since the dilation Laplacian or the frustrations  are not normalized, the inequality (\ref{eq:frustrationCheeger}) involves ratios of the eigenvalues and frustration constants with the volume of the graph, $\vol(\mathcal{G}) =\sum_{i,j\in V}w_{ij}$. This fact ensures that the inequality keeps the same form if the all weights are rescaled by a positive factor, i.e., $w'_{ij} = \rho w_{ij}$ with $\rho>0$ for all  $\{i,j\}\in E_u$.
\subsection{Hodge decomposition, inconsistencies and emphasis of the top part of the ranking}

Taking into account the obstruction to a perfect ranking, an upper bound on $\lambda_0^{(g)}$ can be obtained starting from the solution of Hodgerank. Let $h$ be the solution of~(\ref{HodegerankSolution}) and let $a_{ij} = h_i-h_j + \epsilon_{ij}$ be the corresponding Hodge decomposition, where we remind that the skew-symmetric edgeflow  $\epsilon_{ij}$ represents the obstruction to obtain a perfect ranking, that is, it quantifies the inconsistency of the data. Then, the sum of squared residuals for the Hodgerank solution
\begin{equation}
\epsilon_H^2 = \frac{1}{2}\sum_{i,j\in V}w_{ij}\epsilon_{ij}^2,\label{eq:resHodge}
\end{equation}
gives the minimum error achieved by the least squares method. It is important to emphasize here that the sum of residuals in (\ref{eq:resHodge}) involves a priori the same weight for each error term. 
Based on the Hodgerank score $h$, the natural $\mathbb{R}_\star^+$-potential $e^{gh}$ has a frustration which is not too large provided that the error done by Hodgerank~(\ref{eq:resHodge}) is  small. This remark is stated more precisely by Proposition~\ref{UpperboundLeastEigHodge} which extends the result of Lemma~\ref{Lem:frustration}.
\begin{myprop}\label{UpperboundLeastEigHodge}
Let $0<g<1$ and $h$ be the solution of Hodgerank such that $a_{ij} = h_i-h_j + \epsilon_{ij}$. Then, we have
\[
\lambda_0^{(g)}\leq \frac{1}{2} \sum_{i,j\in V}w_{ij}\frac{e^{g(h_i+h_j)}}{\sum_{k\in V} e^{2gh_k}}4\sinh^{2}(\frac{g\epsilon_{ij}}{2})\leq g^2\epsilon_H^2/2+ \mathcal{O}(g^4),
\]
where the sum of squared residuals $\epsilon_H^2$ is independent of $g$ and given in~(\ref{eq:resHodge}).
\end{myprop}
Indeed, the potential $e^{gh}$ gives a good approximation to the least eigenvector $v_0^{(g)}$ whenever $g$ and $\epsilon_H^2$ are small.
In contrast with Hodgerank where the loss function is simply quadratic, we observe from Proposition~\ref{UpperboundLeastEigHodge} that it is the function $\sinh^2(\cdot/2)$ which appears and  that the errors are reweighted by a factor $e^{g(h_i+h_j)}/\sum_{k\in V} e^{2gh_k}$,
which tends to penalize more the errors involving comparisons between $i$ and $j$ when both $h_i$ and $h_j$ are large.

\begin{mylemma}\label{LemmaMajorization}  Let $h\in \mathbb{R}^N$ be the solution of Hodgerank such that $a = -{\rm d}h + \epsilon$ and let $v,\bar{v}_g\in (\mathbb{R}_\star^+)^N$ such that $v=\diag(e^{gh})\bar{v}_g$, as well as the renormalized weight matrix
$\bar{W} = \diag(\frac{e^{gh}}{\|e^{gh}\|_2})W\diag(\frac{e^{gh}}{\|e^{gh}\|_2})$.
Then, the following inequality holds
\[
\frac{v^\intercal\mathcal{L}_g(a,W)v}{v^\intercal v}\geq \frac{\bar{v}^\intercal_g\mathcal{L}_g(\epsilon,\bar{W})\bar{v}_g}{\bar{v}_g^\intercal\bar{v}_g},\]
and, in particular, the frustration constants satisfy
$
\eta_{a,W}^{[2]}\geq \eta_{\epsilon,\bar{W}}^{[2]}.
$
\end{mylemma}
In fact, the vector $\bar{v}_g=\diag(e^{-gh})v>0$ is an approximation of the $\mathbb{R}_\star^+$-potential minimizing the frustration given by $\epsilon$ on the same graph with renormalized  weights \[\bar{w}_{ij} = w_{ij}\frac{e^{g(h_i+h_j)}}{\sum_{k\in V}e^{2gh_k}}.\]
The weights $\bar{w}_{ij}$ give more importance to the comparisons involving objects $i$ and $j$  when the sum of their ranking scores $h_i+h_j$ is large. To summarize, given the solution to Hodgerank $h$, it is possible to improve the score in the top part of the ranking by calculating the least eigenvector of $\mathcal{L}_g(\epsilon,\bar{W})$, which minimizes 
$
\eta_{\epsilon,\bar{W}}^{[2]}(v) = v^\intercal\mathcal{L}_g(\epsilon,\bar{W})v/v^\intercal v
$ over $v\in (\mathbb{R}_\star^+)^N$. This procedure would involve first solving the linear system of Hodgerank and then solving an eigenvalue problem. Rather than computing the Hodgerank score $h$, we propose to minimize a surrogate function which majorizes the objective as explained in Lemma~\ref{LemmaMajorization}, so that only one eigenvector computation is needed. Indeed, the least eigenvector $v_0^{(g)}$ of $\mathcal{L}_g(a,W)$ provides the approximation $\bar{v}_0^{(g)} =\diag(e^{-gh})v_0^{(g)}$ of the least eigenvector of $\mathcal{L}_g(\epsilon,\bar{W})$. More precisely, we have
\[
\lambda_0^{(g)}\Big(\mathcal{L}_g(a,W)\Big) = \min_{v\in(\mathbb{R}_\star^+)^N}  \frac{v^\intercal\mathcal{L}_g(a,W)v}{v^\intercal v}\geq \frac{\bar{v}_0^{(g)\intercal}\mathcal{L}_g(\epsilon,\bar{W})\bar{v}_0^{(g)}}{\bar{v}_0^{(g)\intercal}\bar{v}_0^{(g)}}=\eta_{\epsilon,\bar{W}}^{[2]}(\bar{v}_0^{(g)}) \geq \eta_{\epsilon,\bar{W}}^{[2]}\geq 0.
\]
This gives a reason why we expect that the ranking score  $v_0^{(g)}$ has a reduced number of upsets in its top part.

 \section{Normalization of the dilation Laplacian and random walk ranking\label{sec:Robust}}
 
It is well-known in the context of spectral clustering that normalized versions of the combinatorial Laplacian yield improved results. Hence, we mention in this section how a normalization of the  dilation Laplacian can be defined.
Consider cardinal comparisons $a_{ij} = h_i-h_j + \epsilon_{ij}$ and a $\mathbb{R}_\star^{+}$-potential $v_i = e^{h_i}$. We define the following objective function
\begin{equation}
\frac{v^\intercal\mathcal{L}_g v}{v^\intercal\mathcal{D}_{g} v} =\frac{1}{2} \frac{\sum_{i,j\in V}w_{ij}e^{h_i+h_j}\mbox{ sinh}^2(\frac{h_i-h_j-g a_{ij}}{2})}{\sum_{i\in V}[\mathcal{D}_{g}]_{ii} e^{2h_i}},\label{eq:ObjGen}
\end{equation}
with the deformed degree $[\mathcal{D}_{g}]_{ii}$ introduced in (\ref{DilationLaplacian}).
Indeed, the objective (\ref{eq:ObjGen}) is associated to the generalized eigenvalue problem $\mathcal{L}_g f_0 = \lambda_0 \mathcal{D}_{g}f_0$.
For all $g\geq 0$, the sum of the eigenvalues of the dilation Laplacian is given by ${\rm Tr}(\mathcal{L}_{g}) = {\rm Tr}(\mathcal{D}_{g})$.
\begin{mythm}\label{Thm3}
There exists a solution $v^\star\in \mathbb{R}^{N}$ to the generalized eigenvalue problem
\begin{equation}
\mathop{\rm minimize}_{\{v\in \mathbb{R}^{N}|v\neq 0\}}\frac{v^\intercal\mathcal{L}_g v}{v^\intercal \mathcal{D}_{g} v},\label{eq:GeneralizedEigen}
\end{equation}
such that $v^\star_{i}>0$ for all $i\in V$.
\end{mythm}
In fact, Theorem~\ref{Thm3} justifies a posteriori the definition of the objective function (\ref{eq:ObjGen}). The generalized eigenvalue problem becomes equivalent to the normalization of the dilation Laplacian. In this spirit, we propose a random walk defined by the transition probability to jump from $j$ to $i$
\begin{equation}
P_g(i|j) = \frac{w_{ji}e^{g a_{ji}/2}}{\sum_{k\in V}w_{jk}e^{g a_{jk}/2}},\label{eq:TransitionMatrix}
\end{equation}
which defines a stochastic matrix. In the limit $g\to 0$, the standard random walk on undirected graphs is recovered. Interestingly, in the case of a graph with only directed edges, i.e. $a_{ij}\in\{-1,1\}$, and assuming that any node has at least one outgoing edge, we find the Pagerank random walk without teleportation, i.e. $\lim_{g\to +\infty} P_g(i|j) = w_{ji}/w^{\rm out}_{j}$,  if  $j\to i$
and $P_g(i|j)=0$ otherwise, where we have defined $w^{\rm out}_{i} = \sum_{\{j|i\succ j\}}w_{ij}a_{ij}$. For a finite $g$ the transition matrix defined by (\ref{eq:TransitionMatrix}) is irreducible. Therefore, there exist a unique Frobenius-Perron eigenvector satisfying $\sum_{j\in V}\pi_j P_g(i|j) = \pi_i$.
We observe empirically in Figure~\ref{Fig:MUN} that the ranking score given by $\pi$ in Algorithm~\ref{AlgRW} is more robust to noise in the case of cardinal comparisons. Incidentally, the Perron vector $\pi$ becomes a centrality score for large values of $g$.
On the contrary, if we define the positive vector 
\begin{equation}
\mu =\mathcal{D}_g^{-1}\pi,\label{eq:mu}
\end{equation}
we find that it satisfies the relations
$\sum_{j\in V}w_{ij}(e^{g a_{ji}/2}\mu_j  -e^{g a_{ij}/2}\mu_i)  = 0,$
for all $i\in V$.
Therefore, we expect the score $\mu$ to solve more accurately the ranking problem in the case of ordinal comparison with a low level of noise.
\begin{algorithm}[h]
\caption{Random walk based ranking }\label{AlgRW}
\begin{algorithmic}[1]
\Require Pairwise comparisons $E$ and $\{a_{ij}\}_{[i,j]\in E}$; a positive real $g$;
\State compute the transition matrix $P_g$ given in (\ref{eq:TransitionMatrix});
\State compute its top eigenvector $\pi$ and $\mu$ given in (\ref{eq:mu});
\Return the centrality score $|\pi_i|$  and the ranking score $|\mu_i|$ for all $i\in V$, as well as the rankings obtained by sorting the scores.
\end{algorithmic}
\end{algorithm}

\section{Deformed Laplacians \label{sec:DeformedLaplacians}}
The dilation Laplacians are in fact special cases of a general construction. Indeed, we can introduce a generalization of several existing Laplacians. Let $T_{[i,j]}\in \mathbb{C}^{n\times n}$ be a set of invertible complex matrix associated to the set of oriented edge $[i,j]\in E$. Notice that the positive integer $n>0$ is not the number of nodes in the graph. Each matrix $T_{[i,j]}$ is actually identified with a linear operator on the finite dimensional Hilbert space $\mathbb{C}^{n}$. We also define the set of functions $\Omega_V = \{f:V\to \mathbb{C}^n\}$, endowed with the canonical inner product $\langle f, f' \rangle_{V} = \sum_{i\in V}f^*_i f_i$, where $(\cdot)^*$ denotes the Hermitian conjugate of a complex matrix.
Similarly, we consider the $\mathbb{C}^{n}$-valued functions on the oriented edges which are skew-symmetric with respect to a change of the orientation, i.e., $X_{ji} = -X_{ij}$, for all edges $[i,j]$. The inner product on the space $\Omega_E = \{X:E\to \mathbb{C}^n | X(e) = -X(\bar{e})\}$ is defined by
\begin{equation}
\langle X, Y \rangle_{E} = \frac{1}{2}\sum_{e=[i,j]\in E} w(e)X^*(e) Y(e).\label{eq:EdgeInner}
\end{equation}
Notice the factor $1/2$ in (\ref{eq:EdgeInner}), which is present in order to avoid a double counting since $w_{s}(e)X^*(e) Y(e) = w_{s}(\bar{e})X^*(\bar{e}) Y(\bar{e})$ for all $e\in E$.
Based on this definition, we define a general deformed Laplacian.
\begin{mydef}[Deformed Laplacian]
For a given non-singular matrix $T_{[i,j]}\in \mathbb{C}^{n\times n}$ associated to each oriented edge $[i,j]$, the deformed Laplacian is defined by
\begin{equation}
\big[L_{T}v\big]_i\triangleq\sum_{j\in V}w_{ij}T^{*}_{[j,i]}\Big(T_{[j,i]}v_i-T_{[i,j]}v_j\Big),
\end{equation}
for all $i\in V$ and all functions $v\in \Omega_V$.
\end{mydef}
Noticeably, the combinatorial Laplacian $\mathcal{L}_{0}$ is obtained by choosing $n=1$ and  \mbox{$T_{[i,j]}=1$}.
\begin{mylemma}
The deformed Laplacian $L_{T}:\Omega_V\to\Omega_V$ is given by $L_{T} = D^*_T D_T $, where  $D_T: \Omega_V\to\Omega_E$ (deformed gradient) is given by 
\[\big[D_T v\big]_{ji}= T_{[j,i]}v_i-T_{[i,j]}v_j, {\rm~for~all~} [i,j]\in E {\rm~and~}v\in \Omega_V,\] and its adjoint $D^*_T: \Omega_E\to\Omega_V$ (deformed divergence) is
\[\big[D^*_T X\big]_i= \sum_{j\in V}w_{ij}T^{*}_{(j,i)}X_{ji}, {\rm~for~all~} i\in V {\rm~and~} X\in \Omega_E.\]
\end{mylemma}

\begin{mycorol}
The deformed Laplacian (\ref{DeformedLaplacian}) is positive semi-definite and self-adjoint. In particular, we have the energy functional, for all $v \in \Omega_V$,

\[\langle v,L_{T}v\rangle_{V}=\frac{1}{2}\sum_{i,j\in V}w_{ij}\Big\|T_{[j,i]}v_i-T_{[i,j]}v_j\Big\|^2_{2},
\]
where $\|\cdot \|_2$ is the $2$-norm of vectors of $\mathbb{C}^n$.
\end{mycorol}

\begin{table*}[b]\centering
\begin{tabularx}{\textwidth}{X c  c  c  c}
   Laplacian & Notation &  functions in $\Omega_V$   & edge data  & $T_{[i,j]}$ \\
   \hline
   combinatorial \cite{ChungBook} & $\mathcal{L}_{0}$ & $\mathbb{R}$-valued   & no &$1$ \\
   signed \cite{Kunegis}& $\mathcal{L}^{[\pi]}_{0}$  & $\mathbb{C}$-valued & $s_{ij} = \pm 1\in \mathbb{Z}_2$  &$\exp ( {\rm i}(1-s_{ij}) \pi /4)$ \\
   connection \cite{SingerWu}& $\nabla^2$ & $\mathbb{R}^d$-valued   & $O_{ij} =\exp (o_{ij})\in SO(d)$& $\exp (o_{ij}/2)$ \\
   magnetic \cite{Shubin}& $\mathcal{L}^{[\theta]}_{0}$ & $\mathbb{C}$-valued  & $a_{ij} = -a_{ji}$ &$\exp ( {\rm i} \theta a_{ij}/2 )$ \\
   \emph{dilation}  & $\mathcal{L}_{g}$&$\mathbb{R}$-valued &$a_{ij} = -a_{ji}$ &$\exp (ga_{ij}/2)$ \\
   \emph{infinitesimal dilation}  & $\mathcal{L}^{\rm inf}_{g}$ & $\mathbb{R}$-valued &$a_{ij} = -a_{ji}$ &$1+ga_{ij}/2$ \\
    \emph{dilation magnetic}  & $\mathcal{L}^{[\theta]}_{g}$ & $\mathbb{C}$-valued &$a_{ij} = -a_{ji}$ &$\exp ((g+{\rm i} \theta)a_{ij}/2)$ \\  
      \hline
\end{tabularx}

\caption{The self-adjoint Laplacians as particular cases of the deformed Laplacian $L_g$ (\ref{DeformedLaplacian}).\label{Table1}}
\end{table*}
The definition (\ref{DeformedLaplacian}) is actually a deformed Laplacian generalizing the combinatorial Laplacian  \cite{ChungBook}, the magnetic Laplacian \cite{deVerdiere,Berkolaiko,Shubin}, the signed Laplacian \cite{Kunegis} and the Connection Laplacian \cite{SingerWu} for $SO(d)$. These special cases are summarized in Table \ref{Table1}, whereas additional details can be found in \ref{AppDefLap}.
To some extent, the vector bundle Laplacian of \cite{KenyonVectorBundle} is analogous to (\ref{DeformedLaplacian}) and coincides with the magnetic Laplacian in the case of unitary parallel transporter and functions valued in $\mathbb{C}$. Some formal similarities with the work of \mbox{E. Witten} \cite{WittenMorse,CyconBook} and \mbox{R. Forman} \cite{FormanWittenMorse} are also explained in \ref{AppDefLap}.

\section{Numerical simulations \label{sec:NumSim}}
In this section, the ranking methods are discussed and illustrated both on artificial and real data.
The ranking process always requires to sort the ranking scores. In the case of two objects  with an equal ranking score, the tie is broken randomly.
 
\subsection{Ordinal comparisons}
Different artificial examples involving ordinal comparisons are discussed in the presence of missing comparisons and corruptions.
For clarity, we do not include in this section the results obtained with the random walk method of Algorithm~\ref{AlgRW}, since the ranking obtained from the ranking score $\mu$ are very similar to the dilation Laplacian method. The centrality score $\pi$ gives better results in the presence of corrupted cardinal comparisons.

\subsubsection{Missing comparisons}
We compare the rankings obtained from artificial pairwise comparisons datasets by the dilation Laplacian, following the empirical methodology of \cite{FogelSpectralRanking}. The idea is to construct an ordered comparison matrix associated to a known ranking, with ordinal comparisons. Then, it is modified by removing edges uniformly at random with the constraint that the graph remains connected. 
The rankings obtained are compared  with the ``correct'' ranking using Kendall's $\tau$-distance, which counts the number of pairwise differences between the rankings, divided by the total number of pairwise comparisons. 
\begin{figure}[t]
\setlength\figureheight{5cm} 
\setlength\figurewidth{0.4\textwidth}
%
%
\definecolor{mycolor1}{rgb}{1.00000,0.00000,1.00000}%
\begin{tikzpicture}

\begin{axis}[%
width=\figurewidth,
height=\figureheight,
at={(0.769in,0.47in)},
scale only axis,
xmin=0,
xmax=0.75,
xlabel={\% missing},
ymin=0,
ymax=0.035,
ylabel={Kendall's $\tau$-distance},
axis background/.style={fill=white},
axis x line*=bottom,
axis y line*=left
]
\addplot [color=red,solid,forget plot]
  table[row sep=crcr]{%
0.01	0.000131313131313131\\
0.06	0.00194141414141414\\
0.11	0.00362828282828283\\
0.16	0.00487878787878788\\
0.21	0.00608686868686869\\
0.26	0.00767474747474747\\
0.31	0.00867676767676768\\
0.36	0.00995555555555556\\
0.41	0.0112363636363636\\
0.46	0.0127454545454545\\
0.51	0.0141151515151515\\
0.56	0.016240404040404\\
0.61	0.0182909090909091\\
0.66	0.0200909090909091\\
0.71	0.0236565656565657\\
};
\addplot [color=red,only marks,mark=x,mark size=3,mark options={solid},forget plot]
plot [error bars/.cd, y dir = both, y explicit]
table[row sep=crcr, x index=0, y index=1, y error index=2]{%
0.01	0.000131313131313131	0.000114079280669855	0.000114079280669855\\
0.06	0.00194141414141414	0.000401605734042942	0.000401605734042942\\
0.11	0.00362828282828283	0.000432342112303691	0.000432342112303691\\
0.16	0.00487878787878788	0.000673203090645274	0.000673203090645274\\
0.21	0.00608686868686869	0.000858024238205688	0.000858024238205688\\
0.26	0.00767474747474747	0.000780686452261126	0.000780686452261126\\
0.31	0.00867676767676768	0.000914742277450885	0.000914742277450885\\
0.36	0.00995555555555556	0.00101936792238567	0.00101936792238567\\
0.41	0.0112363636363636	0.00112453427419595	0.00112453427419595\\
0.46	0.0127454545454545	0.00116261328682411	0.00116261328682411\\
0.51	0.0141151515151515	0.0013127046862579	0.0013127046862579\\
0.56	0.016240404040404	0.00162775295951159	0.00162775295951159\\
0.61	0.0182909090909091	0.00166388015656885	0.00166388015656885\\
0.66	0.0200909090909091	0.00178223822475925	0.00178223822475925\\
0.71	0.0236565656565657	0.00232031760372199	0.00232031760372199\\
};
\addplot [color=blue,solid,forget plot]
  table[row sep=crcr]{%
0.01	0.000127272727272727\\
0.06	0.00189292929292929\\
0.11	0.00361212121212121\\
0.16	0.00483434343434344\\
0.21	0.00601414141414141\\
0.26	0.00751313131313131\\
0.31	0.00854949494949495\\
0.36	0.00985050505050506\\
0.41	0.0111818181818182\\
0.46	0.0126181818181818\\
0.51	0.014020202020202\\
0.56	0.0164020202020202\\
0.61	0.0186222222222222\\
0.66	0.0223535353535354\\
0.71	0.0268828282828283\\
};
\addplot [color=blue,only marks,mark=o,mark size=3,mark options={solid},forget plot]
plot [error bars/.cd, y dir = both, y explicit]
table[row sep=crcr, x index=0, y index=1, y error index=2]{%
0.01	0.000127272727272727	9.74640391963538e-05	9.74640391963538e-05\\
0.06	0.00189292929292929	0.000397184672726895	0.000397184672726895\\
0.11	0.00361212121212121	0.000411475599898912	0.000411475599898912\\
0.16	0.00483434343434344	0.000643717752002129	0.000643717752002129\\
0.21	0.00601414141414141	0.000887700429244752	0.000887700429244752\\
0.26	0.00751313131313131	0.000846315118038099	0.000846315118038099\\
0.31	0.00854949494949495	0.000882654624417497	0.000882654624417497\\
0.36	0.00985050505050506	0.000920413241686702	0.000920413241686702\\
0.41	0.0111818181818182	0.00115191832569863	0.00115191832569863\\
0.46	0.0126181818181818	0.00129546494297961	0.00129546494297961\\
0.51	0.014020202020202	0.00142309501713583	0.00142309501713583\\
0.56	0.0164020202020202	0.00237187389847942	0.00237187389847942\\
0.61	0.0186222222222222	0.0026933350378941	0.0026933350378941\\
0.66	0.0223535353535354	0.00442774593976423	0.00442774593976423\\
0.71	0.0268828282828283	0.00522812825056164	0.00522812825056164\\
};
\addplot [color=mycolor1,solid,forget plot]
  table[row sep=crcr]{%
0.01	0.000268686868686869\\
0.06	0.00292525252525253\\
0.11	0.00507272727272728\\
0.16	0.00657777777777778\\
0.21	0.00809696969696969\\
0.26	0.0100141414141414\\
0.31	0.0114545454545455\\
0.36	0.012989898989899\\
0.41	0.0149030303030303\\
0.46	0.0165454545454546\\
0.51	0.0181636363636364\\
0.56	0.0210262626262626\\
0.61	0.0239010101010101\\
0.66	0.0260484848484848\\
0.71	0.0301959595959596\\
};
\addplot [color=mycolor1,only marks,mark=triangle,mark size=3,mark options={solid},forget plot]
plot [error bars/.cd, y dir = both, y explicit]
table[row sep=crcr, x index=0, y index=1, y error index=2]{%
0.01	0.000268686868686869	0.000152180221105163	0.000152180221105163\\
0.06	0.00292525252525253	0.000442867463323943	0.000442867463323943\\
0.11	0.00507272727272728	0.000612312289844248	0.000612312289844248\\
0.16	0.00657777777777778	0.000837845955318153	0.000837845955318153\\
0.21	0.00809696969696969	0.000955733310175703	0.000955733310175703\\
0.26	0.0100141414141414	0.000919682227689518	0.000919682227689518\\
0.31	0.0114545454545455	0.00121469526728629	0.00121469526728629\\
0.36	0.012989898989899	0.00130206394871349	0.00130206394871349\\
0.41	0.0149030303030303	0.00161528655964096	0.00161528655964096\\
0.46	0.0165454545454546	0.00150183109163148	0.00150183109163148\\
0.51	0.0181636363636364	0.00159876354882428	0.00159876354882428\\
0.56	0.0210262626262626	0.00214686028342357	0.00214686028342357\\
0.61	0.0239010101010101	0.00156761901467745	0.00156761901467745\\
0.66	0.0260484848484848	0.00235352075537616	0.00235352075537616\\
0.71	0.0301959595959596	0.00266798518779637	0.00266798518779637\\
};
\addplot [color=green,solid,forget plot]
  table[row sep=crcr]{%
0.01	0.000135353535353535\\
0.06	0.00192323232323232\\
0.11	0.00362626262626262\\
0.16	0.0048949494949495\\
0.21	0.00604848484848485\\
0.26	0.00758585858585858\\
0.31	0.00863232323232323\\
0.36	0.00987070707070708\\
0.41	0.0112121212121212\\
0.46	0.0126666666666667\\
0.51	0.0140545454545455\\
0.56	0.0162141414141414\\
0.61	0.0181676767676768\\
0.66	0.0200080808080808\\
0.71	0.0235393939393939\\
};
\addplot [color=green,only marks,mark=square,mark size=3,mark options={solid},forget plot]
plot [error bars/.cd, y dir = both, y explicit]
table[row sep=crcr, x index=0, y index=1, y error index=2]{%
0.01	0.000135353535353535	0.000116535209867648	0.000116535209867648\\
0.06	0.00192323232323232	0.000400395848681665	0.000400395848681665\\
0.11	0.00362626262626262	0.000427211294172798	0.000427211294172798\\
0.16	0.0048949494949495	0.0006280654456265	0.0006280654456265\\
0.21	0.00604848484848485	0.00087836956945094	0.00087836956945094\\
0.26	0.00758585858585858	0.000827243919599486	0.000827243919599486\\
0.31	0.00863232323232323	0.000926494581226226	0.000926494581226226\\
0.36	0.00987070707070708	0.000973366245497936	0.000973366245497936\\
0.41	0.0112121212121212	0.00108176997483355	0.00108176997483355\\
0.46	0.0126666666666667	0.001206956632977	0.001206956632977\\
0.51	0.0140545454545455	0.00140587299150875	0.00140587299150875\\
0.56	0.0162141414141414	0.00168036759050502	0.00168036759050502\\
0.61	0.0181676767676768	0.00171432093870653	0.00171432093870653\\
0.66	0.0200080808080808	0.00180989890499354	0.00180989890499354\\
0.71	0.0235393939393939	0.00243563597601269	0.00243563597601269\\
};
\end{axis}
\end{tikzpicture}%
%
%
\definecolor{mycolor1}{rgb}{1.00000,0.00000,1.00000}%
\begin{tikzpicture}

\begin{axis}[%
width=\figurewidth,
height=\figureheight,
at={(0.769in,0.47in)},
scale only axis,
xmin=0,
xmax=0.75,
ymin=-0.001,
ymax=0.5,
axis background/.style={fill=white},
ylabel={$\%$ upsets in top $10$},
xlabel={$\%$ missing},
axis x line*=bottom,
axis y line*=left
]
\addplot [color=red,solid,forget plot]
  table[row sep=crcr]{%
0.01	0.00181912144702842\\
0.06	0.0337448108508502\\
0.11	0.0562683124844591\\
0.16	0.0720982410222516\\
0.21	0.0997471200664035\\
0.26	0.123303963758984\\
0.31	0.126620595451609\\
0.36	0.130702110812749\\
0.41	0.148972213671404\\
0.46	0.178250873603048\\
0.51	0.224492394597072\\
0.56	0.224638627321438\\
0.61	0.264417302698498\\
0.66	0.234551297343867\\
0.71	0.274627356561567\\
};
\addplot [color=red,only marks,mark=x,mark size=3,mark options={solid},forget plot]
 plot [error bars/.cd, y dir = both, y explicit]
table[row sep=crcr, x index=0, y index=1, y error index=2]{%
0.01	0.00181912144702842	0.00623332050979586	0.00623332050979586\\
0.06	0.0337448108508502	0.0284338325212038	0.0284338325212038\\
0.11	0.0562683124844591	0.0393374786231545	0.0393374786231545\\
0.16	0.0720982410222516	0.0468579563076082	0.0468579563076082\\
0.21	0.0997471200664035	0.0501390476204148	0.0501390476204148\\
0.26	0.123303963758984	0.0664178118349512	0.0664178118349512\\
0.31	0.126620595451609	0.074189198984736	0.074189198984736\\
0.36	0.130702110812749	0.064333082267011	0.064333082267011\\
0.41	0.148972213671404	0.0840942109598517	0.0840942109598517\\
0.46	0.178250873603048	0.0957678832169431	0.0957678832169431\\
0.51	0.224492394597072	0.123908423550589	0.123908423550589\\
0.56	0.224638627321438	0.110697101282075	0.110697101282075\\
0.61	0.264417302698498	0.126840044962054	0.126840044962054\\
0.66	0.234551297343867	0.132454489355624	0.132454489355624\\
0.71	0.274627356561567	0.133283434087954	0.133283434087954\\
};
\addplot [color=green,solid,forget plot]
  table[row sep=crcr]{%
0.01	0.00181912144702842\\
0.06	0.0403038662783118\\
0.11	0.064224342218678\\
0.16	0.0816214849822165\\
0.21	0.11080766484789\\
0.26	0.136243763334382\\
0.31	0.144383934828576\\
0.36	0.155346025856219\\
0.41	0.155014996904956\\
0.46	0.188109419831916\\
0.51	0.247334321508248\\
0.56	0.238168972312592\\
0.61	0.285653146213652\\
0.66	0.241487972639675\\
0.71	0.295909118659119\\
};
\addplot [color=green,only marks,mark=square,mark size=3,mark options={solid},forget plot]
 plot [error bars/.cd, y dir = both, y explicit]
table[row sep=crcr, x index=0, y index=1, y error index=2]{%
0.01	0.00181912144702842	0.00623332050979586	0.00623332050979586\\
0.06	0.0403038662783118	0.0319767119008299	0.0319767119008299\\
0.11	0.064224342218678	0.0383083121967	0.0383083121967\\
0.16	0.0816214849822165	0.0469208259864053	0.0469208259864053\\
0.21	0.11080766484789	0.0571224057842225	0.0571224057842225\\
0.26	0.136243763334382	0.0675644603349876	0.0675644603349876\\
0.31	0.144383934828576	0.0770428186748816	0.0770428186748816\\
0.36	0.155346025856219	0.0721941504470028	0.0721941504470028\\
0.41	0.155014996904956	0.0854394876592115	0.0854394876592115\\
0.46	0.188109419831916	0.0951494655562224	0.0951494655562224\\
0.51	0.247334321508248	0.124572564390034	0.124572564390034\\
0.56	0.238168972312592	0.114524778053427	0.114524778053427\\
0.61	0.285653146213652	0.12719279860106	0.12719279860106\\
0.66	0.241487972639675	0.116435062661089	0.116435062661089\\
0.71	0.295909118659119	0.1532114501062	0.1532114501062\\
};
\addplot [color=mycolor1,solid,mark=triangle,mark size=3,mark options={solid},forget plot]
  table[row sep=crcr]{%
0.01	0.00317312661498708\\
0.06	0.0484678782489333\\
0.11	0.0806992606737325\\
0.16	0.0974408676186107\\
0.21	0.123293372160636\\
0.26	0.149836436391977\\
0.31	0.183346590907129\\
0.36	0.192314661466151\\
0.41	0.170737429511838\\
0.46	0.229182418804921\\
0.51	0.25261792017428\\
0.56	0.260154862886624\\
0.61	0.282872290494985\\
0.66	0.297119982477103\\
0.71	0.333610468307063\\
};
\addplot [color=mycolor1,only marks,mark=triangle,mark size=3,mark options={solid},forget plot]
 plot [error bars/.cd, y dir = both, y explicit]
table[row sep=crcr, x index=0, y index=1, y error index=2]{%
0.01	0.00317312661498708	0.00794672396132344	0.00794672396132344\\
0.06	0.0484678782489333	0.0302122818174844	0.0302122818174844\\
0.11	0.0806992606737325	0.040631577302081	0.040631577302081\\
0.16	0.0974408676186107	0.0499096420735051	0.0499096420735051\\
0.21	0.123293372160636	0.0606399975371254	0.0606399975371254\\
0.26	0.149836436391977	0.0688035219344232	0.0688035219344232\\
0.31	0.183346590907129	0.0762163593999396	0.0762163593999396\\
0.36	0.192314661466151	0.079179814072751	0.079179814072751\\
0.41	0.170737429511838	0.0838200930122928	0.0838200930122928\\
0.46	0.229182418804921	0.0979049735532107	0.0979049735532107\\
0.51	0.25261792017428	0.106689163015341	0.106689163015341\\
0.56	0.260154862886624	0.118755177503477	0.118755177503477\\
0.61	0.282872290494985	0.139376097021135	0.139376097021135\\
0.66	0.297119982477103	0.133677439086582	0.133677439086582\\
0.71	0.333610468307063	0.150642965097132	0.150642965097132\\
};
\addplot [color=blue,solid,forget plot]
  table[row sep=crcr]{%
0.01	0.00181912144702842\\
0.06	0.0407800567545022\\
0.11	0.0638054596463936\\
0.16	0.0792984224238599\\
0.21	0.115008682143765\\
0.26	0.143973365219802\\
0.31	0.146059547366936\\
0.36	0.154295394078845\\
0.41	0.160152051027456\\
0.46	0.191531874278785\\
0.51	0.245561775115598\\
0.56	0.236569570304211\\
0.61	0.264792511600662\\
0.66	0.237297645337119\\
0.71	0.28012735209407\\
};
\addplot [color=blue,only marks,mark=o,mark size=3,mark options={solid},forget plot]
 plot [error bars/.cd, y dir = both, y explicit]
table[row sep=crcr, x index=0, y index=1, y error index=2]{%
0.01	0.00181912144702842	0.00623332050979586	0.00623332050979586\\
0.06	0.0407800567545022	0.031903267624702	0.031903267624702\\
0.11	0.0638054596463936	0.0380541950731353	0.0380541950731353\\
0.16	0.0792984224238599	0.0466449216262757	0.0466449216262757\\
0.21	0.115008682143765	0.0540737530205649	0.0540737530205649\\
0.26	0.143973365219802	0.0746408467537126	0.0746408467537126\\
0.31	0.146059547366936	0.0771998513383786	0.0771998513383786\\
0.36	0.154295394078845	0.0656546418599027	0.0656546418599027\\
0.41	0.160152051027456	0.0892760797305568	0.0892760797305568\\
0.46	0.191531874278785	0.103812873757176	0.103812873757176\\
0.51	0.245561775115598	0.134710841263663	0.134710841263663\\
0.56	0.236569570304211	0.09195442057293	0.09195442057293\\
0.61	0.264792511600662	0.129388805775614	0.129388805775614\\
0.66	0.237297645337119	0.138086543509612	0.138086543509612\\
0.71	0.28012735209407	0.148124317381724	0.148124317381724\\
};
\end{axis}
\end{tikzpicture}%
\caption{On the left, Kendall's $\tau$-distances (the lower the better) between a known ranking and rankings obtained by various methods on this artificial dataset, i.e., least-squares (green squares \cite{YaoRanking}), Serialrank (magenta triangles \cite{FogelSpectralRanking}), Spectral Sync-Rank (blue circles \cite{Cucuringu}), dilation Laplacian (red crosses, $g = 0.1$). On the right, percentage of upsets in the top $10$ (the lower the better). The computations were repeated $50$ times, and the error bars are the standard deviations. \label{Fig:Missing_100_50repeats}}
\end{figure}
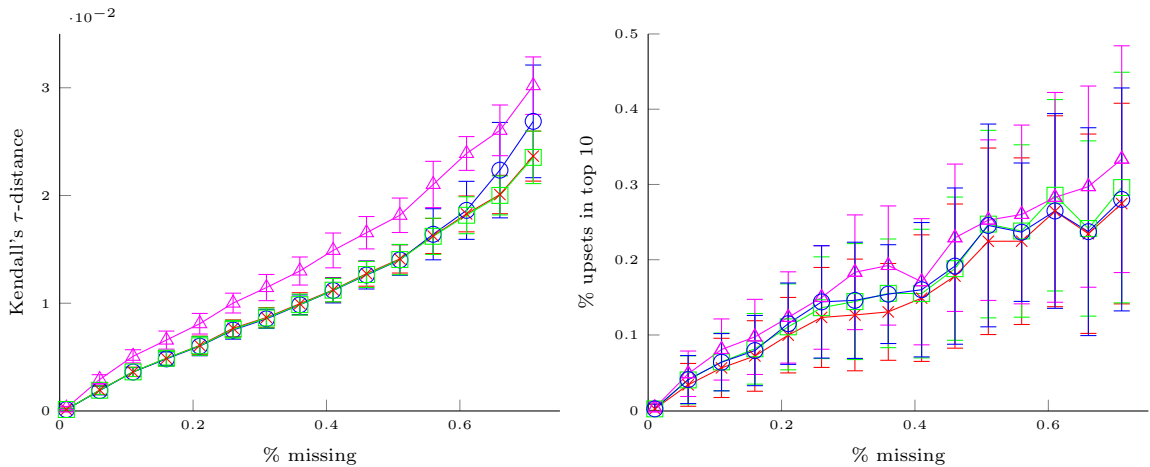
Moreover, our method based on the dilation Laplacian is  compared with the least-squares ranking \cite{YaoRanking} and two spectral methods, that is, SerialRank in the matching comparisons scheme \cite{FogelSpectralRanking}, the spectral version of Sync-Rank \cite{Cucuringu}. The features of these methods are outlined in Table~\ref{Tab:Methods}. In the presence of ties in the scores, the ranking is selected at random.
The results obtained on a set of artificial pairwise comparisons between $100$ objects are displayed in Figure \ref{Fig:Missing_100_50repeats}.
\begin{table}[h]\centering
\begin{tabular}{l | c c c c c}
  &  Sync-Rank~\cite{Cucuringu} & Hodgerank~\cite{YaoRanking}  & Serialrank~\cite{FogelSpectralRanking}  & Dilation & Random walk\\ 
\hline 
\rule{0pt}{3ex}Choice & $H_{ij} = \exp(\frac{\pi a_{ij}}{N-1})$ & \slash & matrix $S$ & $g>0$& $g>0$\\
 Problem & $H v_1 = \lambda_1 v_1$ & $\mathcal{L}_0 s+{\rm div}a = 0$ & $\mathcal{L}^{S}_0y =\lambda y$ & $\mathcal{L}_g v_0^{(g)} =\lambda_0^{(g)} v_0^{(g)}$ & $\pi P_g = \pi$\\
Type & Top eig. & least-squares & Fiedler eig. & Lowest eig.& Top eig.\\
Solution & $c\times v_1 \in \mathbb{C}^N$ & $s\in \mathbb{R}^N$ & $c\times y \in \mathbb{R}^N$ & $c\times v_0^{(g)}\in \mathbb{R}^N$& $c\times\pi^{(g)}\in \mathbb{R}^N$\\
Score of $i$ & $[v_{1}]_i/|[v_{1}]_i|$ & $s_i +{\rm cst}$ & $\pm y_i$ & $|c\times [v_{0}^{(g)}]_i|$ & $|c\times\pi_i|$\\
Rankings & $N$ & $1$ & $2$ & $1$ & $1$\\
Selection & min. upsets & \slash & min. upsets & \slash & \slash\\
 \hline
\end{tabular}
\caption{Features of the methods compared for ranking from pairwise comparisons.\label{Tab:Methods}}
\end{table}

On the one hand, in terms of Kendall's $\tau$-distance with respect to a known ranking (Criterion $1$), the numerical results are similar for Hodgerank, Sync-Rank and the dilation Laplacian method. Serialrank is known to be more robust to corrupted comparisons. Another accuracy measure (Criterion $2$), used for instance in \cite{FogelSpectralRanking}, is the number of upsets in the top-$k$, that is, the number of times an object in the top-$k$ is preferred to another object in the top-$k$ with a higher ranking. More precisely, if $C_k$ the subset of top-$k$ objects, the number of upsets in the top-$k$ is
$
l_k = \sum_{i,j\in C_k}1_{r_i>r_j}1_{a_{ij}<0},
$
where $1_{\cdot}$ is an indicator function and $r_i\in \{1,\dots,N\}$ is the ranking of the node $i\in V$.
Hence, the number of upsets is an efficiency measure that is independent  of any prior knowledge about the correct ranking.  This is meaningful since, in many applications, only the top-$k$ objects are interesting. 
 With respect to the number of upsets in the top $10$, the dilation ranking performs better than the other methods, as illustrated in Figure~\ref{Fig:Missing_100_50repeats}, while its Kendall's $\tau$-distance with the true ranking is the same of Hodgerank. However, the error bars (standard deviations) on the number of upsets on the right hand side of Figure~\ref{Fig:Missing_100_50repeats} are important so that the statistical significance of the reduced number of upsets is difficult to assess. Hence, in Figure~\ref{Hodge_vs_Dilation_Upsets}, the same simulation was repeated 200 times on graphs of 200 nodes in order to compare the mean number of upsets in the top-10 and top-20 for Hodgerank and the dilation ranking.
In order to determine the statistical significance of the difference between the two methods, the Wilcoxon signed rank test is used with significance level $0.05$ for the null hypothesis that the difference between the results comes from a distribution with zero median (In fact, if the test value is $1$, the null hypothesis is rejected). Concerning the upsets in the top-$10$, we observe on the \emph{lhs} of Figure~\ref{Hodge_vs_Dilation_Upsets} that the p-value of this test (dashed blue line) increases above $0.05$ for a large percentage of missing comparisons. Indeed, in this case, the difference between Hodgerank and Algorithm~\ref{Alg1} is significant only for a low number of missing comparisons. However, the number of upsets in the top-20 between the two methods is significantly different for all values of missing comparisons.

\begin{figure}[h]
\setlength\figureheight{4.5cm} 
\setlength\figurewidth{0.35\textwidth}
%
%
\begin{tikzpicture}

\begin{axis}[%
width=\figurewidth,
height=\figureheight,
at={(0.769in,0.47in)},
scale only axis,
xmin=0,
xmax=0.9,
xlabel={\% missing},
ymin=-0.01,
ymax=0.4,
ylabel={\% upsets top 10},
axis background/.style={fill=white},
axis x line*=bottom,
axis y line*=left,
xticklabel style={/pgf/number format/fixed},
yticklabel style={/pgf/number format/fixed}
]
\addplot [color=red,solid,forget plot]
  table[row sep=crcr]{%
0.01	0.00643036679083191\\
0.06	0.056424392164186\\
0.11	0.0905951565339159\\
0.16	0.115555524889872\\
0.21	0.14985211940044\\
0.26	0.170653229098055\\
0.31	0.186756559763413\\
0.36	0.208726875177313\\
0.41	0.231090827183207\\
0.46	0.251801534516895\\
0.51	0.270358880092638\\
0.56	0.281676622025312\\
0.61	0.284604370298542\\
0.66	0.334944783308367\\
0.71	0.338043425185066\\
0.76	0.362718879310404\\
0.81	0.36744489089342\\
};
\addplot [color=red,only marks,mark=x,mark options={solid},forget plot]
 table[row sep=crcr, y error plus index=2, y error minus index=3]{%
0.01	0.00643036679083191	0.0120665407235707	0.0120665407235707\\
0.06	0.056424392164186	0.0334020808345737	0.0334020808345737\\
0.11	0.0905951565339159	0.0486786813855953	0.0486786813855953\\
0.16	0.115555524889872	0.0536365796616268	0.0536365796616268\\
0.21	0.14985211940044	0.0699724658544126	0.0699724658544126\\
0.26	0.170653229098055	0.0783225385653001	0.0783225385653001\\
0.31	0.186756559763413	0.0781374331845478	0.0781374331845478\\
0.36	0.208726875177313	0.0872064125526041	0.0872064125526041\\
0.41	0.231090827183207	0.0978493585235164	0.0978493585235164\\
0.46	0.251801534516895	0.104281169745788	0.104281169745788\\
0.51	0.270358880092638	0.121557498570273	0.121557498570273\\
0.56	0.281676622025312	0.11896668688524	0.11896668688524\\
0.61	0.284604370298542	0.139500001661568	0.139500001661568\\
0.66	0.334944783308367	0.161383981347896	0.161383981347896\\
0.71	0.338043425185066	0.158393191187764	0.158393191187764\\
0.76	0.362718879310404	0.163907244017194	0.163907244017194\\
0.81	0.36744489089342	0.207042920638949	0.207042920638949\\
};
\addplot [color=green,solid,forget plot]
  table[row sep=crcr]{%
0.01	0.00879423806167992\\
0.06	0.0657921053788295\\
0.11	0.103408709147391\\
0.16	0.13017318089664\\
0.21	0.167751890546373\\
0.26	0.184536239187526\\
0.31	0.206331288809675\\
0.36	0.232057979576581\\
0.41	0.247779896195398\\
0.46	0.264668651487709\\
0.51	0.287925731610315\\
0.56	0.299706177925123\\
0.61	0.305174520449339\\
0.66	0.350285701816132\\
0.71	0.346325876028856\\
0.76	0.373597080722855\\
0.81	0.382046737984238\\
};
\addplot [color=green,only marks,mark=square,mark options={solid},forget plot]
 table[row sep=crcr, y error plus index=2, y error minus index=3]{%
0.01	0.00879423806167992	0.0137352497510695	0.0137352497510695\\
0.06	0.0657921053788295	0.0353478908944531	0.0353478908944531\\
0.11	0.103408709147391	0.051263272119493	0.051263272119493\\
0.16	0.13017318089664	0.0573552070714648	0.0573552070714648\\
0.21	0.167751890546373	0.0700448787126459	0.0700448787126459\\
0.26	0.184536239187526	0.0819536230629205	0.0819536230629205\\
0.31	0.206331288809675	0.0865992586145045	0.0865992586145045\\
0.36	0.232057979576581	0.093859913509995	0.093859913509995\\
0.41	0.247779896195398	0.102899764296586	0.102899764296586\\
0.46	0.264668651487709	0.113308448967173	0.113308448967173\\
0.51	0.287925731610315	0.128770469892	0.128770469892\\
0.56	0.299706177925123	0.13078003583588	0.13078003583588\\
0.61	0.305174520449339	0.144403085044137	0.144403085044137\\
0.66	0.350285701816132	0.156960244103059	0.156960244103059\\
0.71	0.346325876028856	0.162434681841662	0.162434681841662\\
0.76	0.373597080722855	0.175997695676631	0.175997695676631\\
0.81	0.382046737984238	0.195963429040985	0.195963429040985\\
};
\end{axis}

\begin{axis}[%
blue,
width=\figurewidth,
height=\figureheight,
at={(0.769in,0.47in)},
scale only axis,
axis y line*=right,
xmin=0,
xmax=0.9,
ymin=-0.01,
ymax=1.02,
ylabel={p-value Wilcoxon test},
]

\addplot [color=blue,dashed,forget plot]
  table[row sep=crcr]{%
0.01	6.55566199986336e-05\\
0.06	4.5573972232307e-10\\
0.11	1.77115658288588e-11\\
0.16	2.4559771159724e-14\\
0.21	1.57144370782495e-13\\
0.26	1.17517167624449e-07\\
0.31	8.60644029730102e-11\\
0.36	6.38255611987351e-10\\
0.41	4.05762848138017e-05\\
0.46	0.003486144883075\\
0.51	0.000366355364079371\\
0.56	0.00116390479350599\\
0.61	0.0038713971222178\\
0.66	0.0303046725824188\\
0.71	0.282161535213459\\
0.76	0.301926761033453\\
0.81	0.155816194654476\\
};
\addplot [color=blue,only marks,mark=asterisk,mark options={solid},forget plot]
  table[row sep=crcr]{%
0.01	1\\
0.06	1\\
0.11	1\\
0.16	1\\
0.21	1\\
0.26	1\\
0.31	1\\
0.36	1\\
0.41	1\\
0.46	1\\
0.51	1\\
0.56	1\\
0.61	1\\
0.66	1\\
0.71	0\\
0.76	0\\
0.81	0\\
};
\end{axis}
\end{tikzpicture}%
%
%
\begin{tikzpicture}

\begin{axis}[%
width=\figurewidth,
height=\figureheight,
at={(0.769in,0.47in)},
scale only axis,
xmin=0,
xmax=0.9,
xlabel={\% missing},
ymin=-0.01,
ymax=0.4,
ylabel={\% upsets top 20},
axis background/.style={fill=white},
axis x line*=bottom,
axis y line*=left,
xticklabel style={/pgf/number format/fixed},
yticklabel style={/pgf/number format/fixed}
]
\addplot [color=red,solid,forget plot]
  table[row sep=crcr]{%
0.01	0.00348776639287098\\
0.06	0.0280811305566467\\
0.11	0.0497623267595051\\
0.16	0.0654574738185867\\
0.21	0.0798394792623447\\
0.26	0.0938510768843916\\
0.31	0.103339928097464\\
0.36	0.114695224680003\\
0.41	0.135621225542951\\
0.46	0.150512793602637\\
0.51	0.159849511940438\\
0.56	0.177710073311175\\
0.61	0.193286184954217\\
0.66	0.211455915051672\\
0.71	0.237231231584294\\
0.76	0.254916896529695\\
0.81	0.277626405769326\\
};
\addplot [color=red,only marks,mark=x,mark options={solid},forget plot]
 table[row sep=crcr, y error plus index=2, y error minus index=3]{%
0.01	0.00348776639287098	0.00409216261166054	0.00409216261166054\\
0.06	0.0280811305566467	0.011511853821948	0.011511853821948\\
0.11	0.0497623267595051	0.0172546308251936	0.0172546308251936\\
0.16	0.0654574738185867	0.019681442922283	0.019681442922283\\
0.21	0.0798394792623447	0.0252707154243816	0.0252707154243816\\
0.26	0.0938510768843916	0.0286513460036957	0.0286513460036957\\
0.31	0.103339928097464	0.0279179971116468	0.0279179971116468\\
0.36	0.114695224680003	0.0319443860046913	0.0319443860046913\\
0.41	0.135621225542951	0.0371864665551645	0.0371864665551645\\
0.46	0.150512793602637	0.0457218311143162	0.0457218311143162\\
0.51	0.159849511940438	0.0450586629552084	0.0450586629552084\\
0.56	0.177710073311175	0.0505941370240969	0.0505941370240969\\
0.61	0.193286184954217	0.052778367584075	0.052778367584075\\
0.66	0.211455915051672	0.0638582301388556	0.0638582301388556\\
0.71	0.237231231584294	0.068366882869841	0.068366882869841\\
0.76	0.254916896529695	0.083186792384878	0.083186792384878\\
0.81	0.277626405769326	0.0976470017582946	0.0976470017582946\\
};
\addplot [color=green,solid,forget plot]
  table[row sep=crcr]{%
0.01	0.00444611030994886\\
0.06	0.0323715744797341\\
0.11	0.0553297146557418\\
0.16	0.0728809180651206\\
0.21	0.0876248593813235\\
0.26	0.103834252576805\\
0.31	0.113802552048541\\
0.36	0.126209379593276\\
0.41	0.146740646120016\\
0.46	0.160130449297195\\
0.51	0.174642024318892\\
0.56	0.192173294102788\\
0.61	0.207551874254853\\
0.66	0.228524954188806\\
0.71	0.254139515108222\\
0.76	0.270986735316796\\
0.81	0.29262319886907\\
};
\addplot [color=green,only marks,mark=square,mark options={solid},forget plot]
 table[row sep=crcr, y error plus index=2, y error minus index=3]{%
0.01	0.00444611030994886	0.0045288449856695	0.0045288449856695\\
0.06	0.0323715744797341	0.0126367722794317	0.0126367722794317\\
0.11	0.0553297146557418	0.0184241007803564	0.0184241007803564\\
0.16	0.0728809180651206	0.0215009446668178	0.0215009446668178\\
0.21	0.0876248593813235	0.0268116799541749	0.0268116799541749\\
0.26	0.103834252576805	0.0313185778161478	0.0313185778161478\\
0.31	0.113802552048541	0.0345942525156821	0.0345942525156821\\
0.36	0.126209379593276	0.0327443168080336	0.0327443168080336\\
0.41	0.146740646120016	0.0388539005158602	0.0388539005158602\\
0.46	0.160130449297195	0.0452768335648422	0.0452768335648422\\
0.51	0.174642024318892	0.0508772939002934	0.0508772939002934\\
0.56	0.192173294102788	0.0539227560092918	0.0539227560092918\\
0.61	0.207551874254853	0.0602752083069949	0.0602752083069949\\
0.66	0.228524954188806	0.0661950468479063	0.0661950468479063\\
0.71	0.254139515108222	0.0764607177720913	0.0764607177720913\\
0.76	0.270986735316796	0.0873520392908131	0.0873520392908131\\
0.81	0.29262319886907	0.0983512174359658	0.0983512174359658\\
};
\end{axis}
\begin{axis}[%
blue,       
width=\figurewidth,
height=\figureheight,
at={(0.769in,0.47in)},
scale only axis,
axis y line*=right,   
xmin=0,
xmax=0.9,
ymin=-0.01,
ymax=1.02,
ylabel={p-value Wilcoxon test}
]
\addplot [color=blue,dashed,forget plot]
  table[row sep=crcr]{%
0.01	2.16360195733814e-07\\
0.06	1.26479047202377e-15\\
0.11	4.55246556281392e-16\\
0.16	8.55323352631101e-19\\
0.21	1.24789172560119e-17\\
0.26	1.9302592210236e-16\\
0.31	1.62182487033744e-14\\
0.36	1.79078252973819e-15\\
0.41	4.55951141284057e-11\\
0.46	4.4740906634349e-07\\
0.51	1.63691943079451e-11\\
0.56	1.58083939691535e-10\\
0.61	1.36503686769943e-06\\
0.66	3.13196297917681e-06\\
0.71	3.57261531865163e-06\\
0.76	8.68151067412654e-05\\
0.81	0.00973081177644769\\
};
\addplot [color=blue,only marks,mark=asterisk,mark options={solid},forget plot]
  table[row sep=crcr]{%
0.01	1\\
0.06	1\\
0.11	1\\
0.16	1\\
0.21	1\\
0.26	1\\
0.31	1\\
0.36	1\\
0.41	1\\
0.46	1\\
0.51	1\\
0.56	1\\
0.61	1\\
0.66	1\\
0.71	1\\
0.76	1\\
0.81	1\\
};
\end{axis}
\end{tikzpicture}%
\caption{On the left, mean percentage of upsets in the top 10 for $g = 0.1$ (left axis). On the right, mean percentage of upsets in the top 20 (left axis). The least-squares (green squares \cite{YaoRanking}) and dilation Laplacian (red crosses) are compared. The simulation was repeated 200 times on a graph of 200 nodes. The dashed blue line is the $p$-value of the Wilcoxon signed rank test, while the blue stars indicate the result of the test at the 5$\%$ significance level (right axis).\label{Hodge_vs_Dilation_Upsets}}
\end{figure}
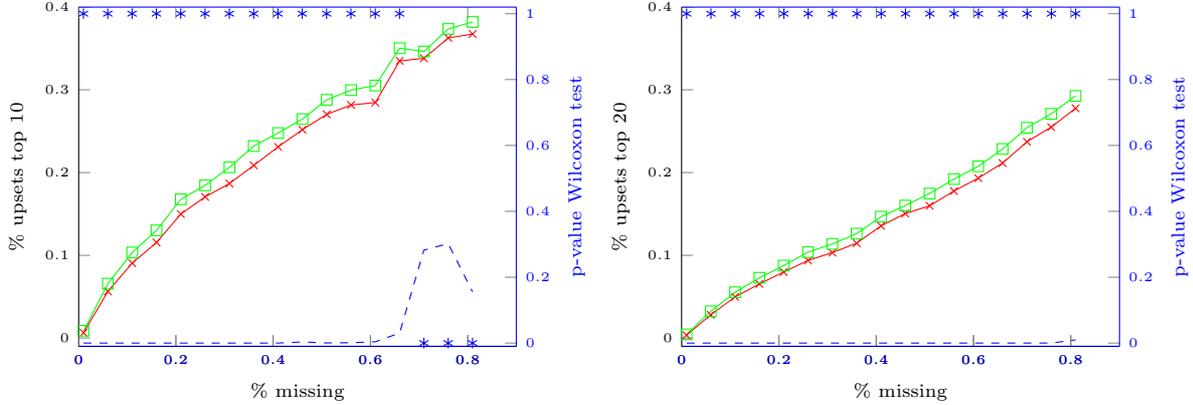

In Figure \ref{Fig:TopK}, the percentage of upsets is displayed with respect to the top-$k$ subset of a set of $50$ objects with $50$ percent of missing comparisons (the simulations were repeated $50$ times). The method based on the dilation Laplacian with $g=0.3$ gives a lower number of upsets in the top-$k$ objects, explained previously. The difference between the results of the various methods naturally decreases as $k$ increases.
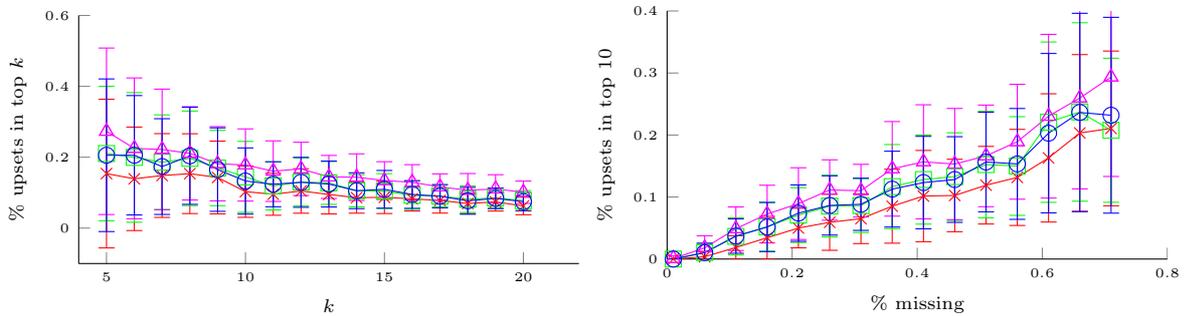
\begin{figure}[h]
\setlength\figureheight{0.2\textwidth} 
\setlength\figurewidth{0.4\textwidth}
%
%
\definecolor{mycolor1}{rgb}{1.00000,0.00000,1.00000}%
\begin{tikzpicture}

\begin{axis}[%
width=\figurewidth,
height=\figureheight,
at={(0.769in,0.47in)},
scale only axis,
xmin=4,
xmax=22,
ymin=-0.1,
ymax=0.6,
axis background/.style={fill=white},
ylabel={$\%$ upsets in top $k$},
xlabel={$k$},
axis x line*=bottom,
axis y line*=left
]
\addplot [color=red,solid,forget plot]
  table[row sep=crcr]{%
5	0.153857142857143\\
6	0.139127705627706\\
7	0.148402042402042\\
8	0.153641363452509\\
9	0.142862838967766\\
10	0.102923144468689\\
11	0.0954937541838265\\
12	0.104797925331553\\
13	0.0949564968765366\\
14	0.0858609310667537\\
15	0.0868596098487159\\
16	0.0825534487261211\\
17	0.0779173481736563\\
18	0.0694720806630826\\
19	0.0739674684210913\\
20	0.0623760137171565\\
};
\addplot [color=red,only marks,mark=x,mark size=3,mark options={solid},forget plot]
 plot [error bars/.cd, y dir = both, y explicit]
table[row sep=crcr, x index=0, y index=1, y error index=2]{%
5	0.153857142857143	0.209768768591875	0.209768768591875\\
6	0.139127705627706	0.14603858885392	0.14603858885392\\
7	0.148402042402042	0.118206141043864	0.118206141043864\\
8	0.153641363452509	0.112593590706975	0.112593590706975\\
9	0.142862838967766	0.102668945191584	0.102668945191584\\
10	0.102923144468689	0.0724995241035452	0.0724995241035452\\
11	0.0954937541838265	0.0592289527942784	0.0592289527942784\\
12	0.104797925331553	0.0627842872882403	0.0627842872882403\\
13	0.0949564968765366	0.0549801639700448	0.0549801639700448\\
14	0.0858609310667537	0.043718535617805	0.043718535617805\\
15	0.0868596098487159	0.0457512724348813	0.0457512724348813\\
16	0.0825534487261211	0.0341951049720744	0.0341951049720744\\
17	0.0779173481736563	0.0355619285790971	0.0355619285790971\\
18	0.0694720806630826	0.0284759996874953	0.0284759996874953\\
19	0.0739674684210913	0.0257728565984884	0.0257728565984884\\
20	0.0623760137171565	0.0246475101817928	0.0246475101817928\\
};
\addplot [color=green,solid,forget plot]
  table[row sep=crcr]{%
5	0.210476190476191\\
6	0.199453102453102\\
7	0.186100843600844\\
8	0.196951619759669\\
9	0.169509252254453\\
10	0.144581895419217\\
11	0.11901589035823\\
12	0.130094820364069\\
13	0.127669540953418\\
14	0.105443948481241\\
15	0.104561935876461\\
16	0.0925022686586917\\
17	0.090239102228292\\
18	0.0813853525942795\\
19	0.0851359553333645\\
20	0.0773888543780806\\
};
\addplot [color=green,only marks,mark=square,mark size=3,mark options={solid},forget plot]
 plot [error bars/.cd, y dir = both, y explicit]
table[row sep=crcr, x index=0, y index=1, y error index=2]{%
5	0.210476190476191	0.189768486057902	0.189768486057902\\
6	0.199453102453102	0.182915646044021	0.182915646044021\\
7	0.186100843600844	0.133447584167087	0.133447584167087\\
8	0.196951619759669	0.1333801798	0.1333801798\\
9	0.169509252254453	0.106227098060188	0.106227098060188\\
10	0.144581895419217	0.100137785643561	0.100137785643561\\
11	0.11901589035823	0.0677675300515031	0.0677675300515031\\
12	0.130094820364069	0.0678779089172993	0.0678779089172993\\
13	0.127669540953418	0.0607175562152856	0.0607175562152856\\
14	0.105443948481241	0.0454384818457862	0.0454384818457862\\
15	0.104561935876461	0.0490884181579129	0.0490884181579129\\
16	0.0925022686586917	0.0397593661274647	0.0397593661274647\\
17	0.090239102228292	0.0306381823200323	0.0306381823200323\\
18	0.0813853525942795	0.0373760500490483	0.0373760500490483\\
19	0.0851359553333645	0.028071531287243	0.028071531287243\\
20	0.0773888543780806	0.0283913930840874	0.0283913930840874\\
};
\addplot [color=mycolor1,solid,mark=triangle,mark size=3,mark options={solid},forget plot]
  table[row sep=crcr]{%
5	0.273285714285714\\
6	0.225002164502165\\
7	0.22200604950605\\
8	0.210940573959924\\
9	0.181628144650402\\
10	0.178142690079913\\
11	0.1607227162573\\
12	0.167109545826155\\
13	0.144380222290155\\
14	0.143760677318653\\
15	0.134151070147558\\
16	0.129443022849473\\
17	0.116244414642169\\
18	0.106916928725705\\
19	0.110412290861864\\
20	0.101236245194149\\
};
\addplot [color=mycolor1,only marks,mark=triangle,mark size=3,mark options={solid},forget plot]
 plot [error bars/.cd, y dir = both, y explicit]
table[row sep=crcr, x index=0, y index=1, y error index=2]{%
5	0.273285714285714	0.234979757857575	0.234979757857575\\
6	0.225002164502165	0.198843035585888	0.198843035585888\\
7	0.22200604950605	0.170295278217929	0.170295278217929\\
8	0.210940573959924	0.131595452661946	0.131595452661946\\
9	0.181628144650402	0.104788404201665	0.104788404201665\\
10	0.178142690079913	0.101876159576679	0.101876159576679\\
11	0.1607227162573	0.0855354577958106	0.0855354577958106\\
12	0.167109545826155	0.0752933384505441	0.0752933384505441\\
13	0.144380222290155	0.0607115113874827	0.0607115113874827\\
14	0.143760677318653	0.0658634602534468	0.0658634602534468\\
15	0.134151070147558	0.0532595292704468	0.0532595292704468\\
16	0.129443022849473	0.0496052300469768	0.0496052300469768\\
17	0.116244414642169	0.0368833837386959	0.0368833837386959\\
18	0.106916928725705	0.0470588730666851	0.0470588730666851\\
19	0.110412290861864	0.0402729828412439	0.0402729828412439\\
20	0.101236245194149	0.0315227510761049	0.0315227510761049\\
};
\addplot [color=blue,solid,forget plot]
  table[row sep=crcr]{%
5	0.205619047619048\\
6	0.205620490620491\\
7	0.173622988122988\\
8	0.204023146272372\\
9	0.165144206946192\\
10	0.132899442254475\\
11	0.123814617799351\\
12	0.129195013976326\\
13	0.12438607816165\\
14	0.105445372218399\\
15	0.109457651962838\\
16	0.094920784063567\\
17	0.0900184771219841\\
18	0.0770173475416762\\
19	0.0839202481509091\\
20	0.0745809176047989\\
};
\addplot [color=blue,only marks,mark=o,mark size=3,mark options={solid},forget plot]
 plot [error bars/.cd, y dir = both, y explicit]
table[row sep=crcr, x index=0, y index=1, y error index=2]{%
5	0.205619047619048	0.215351393935819	0.215351393935819\\
6	0.205620490620491	0.168434993957219	0.168434993957219\\
7	0.173622988122988	0.135113501550919	0.135113501550919\\
8	0.204023146272372	0.13697532056297	0.13697532056297\\
9	0.165144206946192	0.117758563046831	0.117758563046831\\
10	0.132899442254475	0.0934048588432162	0.0934048588432162\\
11	0.123814617799351	0.0635118330827559	0.0635118330827559\\
12	0.129195013976326	0.070526086643423	0.070526086643423\\
13	0.12438607816165	0.0646919723980314	0.0646919723980314\\
14	0.105445372218399	0.0506227516516955	0.0506227516516955\\
15	0.109457651962838	0.0532213188662203	0.0532213188662203\\
16	0.094920784063567	0.0391559973928336	0.0391559973928336\\
17	0.0900184771219841	0.0328645444461295	0.0328645444461295\\
18	0.0770173475416762	0.0378777630099729	0.0378777630099729\\
19	0.0839202481509091	0.0289599479212735	0.0289599479212735\\
20	0.0745809176047989	0.025499939078001	0.025499939078001\\
};
\end{axis}
\end{tikzpicture}%
%
%
\definecolor{mycolor1}{rgb}{1.00000,0.00000,1.00000}%
\begin{tikzpicture}

\begin{axis}[%
width=\figurewidth,
height=\figureheight,
at={(0.769in,0.47in)},
scale only axis,
xmin=0,
xmax=0.8,
ymin=0.,
ymax=0.4,
axis background/.style={fill=white},
ylabel={\% upsets in top $10$},
xlabel={\% missing},
axis x line*=bottom,
axis y line*=left
]
\addplot [color=red,solid,forget plot]
  table[row sep=crcr]{%
0.01	0\\
0.06	0.00376793172794307\\
0.11	0.0186175509480415\\
0.16	0.0343523457425896\\
0.21	0.0499732230984228\\
0.26	0.0592184073958551\\
0.31	0.0650632074887424\\
0.36	0.0851831306888789\\
0.41	0.101523297510757\\
0.46	0.10250247920971\\
0.51	0.119065991652066\\
0.56	0.131602849510859\\
0.61	0.16305922895474\\
0.66	0.203296381894679\\
0.71	0.210639370433488\\
};
\addplot [color=red,only marks,mark=x,mark size=3,mark options={solid},forget plot]
 plot [error bars/.cd, y dir = both, y explicit]
table[row sep=crcr, x index=0, y index=1, y error index=2]{%
0.01	0	0	0\\
0.06	0.00376793172794307	0.00872565050626263	0.00872565050626263\\
0.11	0.0186175509480415	0.0189846525306068	0.0189846525306068\\
0.16	0.0343523457425896	0.0345270439333097	0.0345270439333097\\
0.21	0.0499732230984228	0.0319015578357401	0.0319015578357401\\
0.26	0.0592184073958551	0.045145049942741	0.045145049942741\\
0.31	0.0650632074887424	0.0401395626278111	0.0401395626278111\\
0.36	0.0851831306888789	0.0596472372164651	0.0596472372164651\\
0.41	0.101523297510757	0.0735912303797404	0.0735912303797404\\
0.46	0.10250247920971	0.0584389567585467	0.0584389567585467\\
0.51	0.119065991652066	0.0625823415043009	0.0625823415043009\\
0.56	0.131602849510859	0.0774787901151756	0.0774787901151756\\
0.61	0.16305922895474	0.1034127238851	0.1034127238851\\
0.66	0.203296381894679	0.126374519216236	0.126374519216236\\
0.71	0.210639370433488	0.124775963576085	0.124775963576085\\
};
\addplot [color=green,solid,forget plot]
  table[row sep=crcr]{%
0.01	0\\
0.06	0.0103375984952843\\
0.11	0.0365066414839309\\
0.16	0.0521129184056391\\
0.21	0.0708587806448698\\
0.26	0.0853933829131391\\
0.31	0.0862231939896845\\
0.36	0.116797521568512\\
0.41	0.128151853325039\\
0.46	0.132942193609949\\
0.51	0.152309000240572\\
0.56	0.150006525021769\\
0.61	0.220690580403429\\
0.66	0.237159029051444\\
0.71	0.207635031635032\\
};
\addplot [color=green,only marks,mark=square,mark size=3,mark options={solid},forget plot]
 plot [error bars/.cd, y dir = both, y explicit]
table[row sep=crcr, x index=0, y index=1, y error index=2]{%
0.01	0	0	0\\
0.06	0.0103375984952843	0.0151101574660005	0.0151101574660005\\
0.11	0.0365066414839309	0.0298785856926278	0.0298785856926278\\
0.16	0.0521129184056391	0.0394362333361181	0.0394362333361181\\
0.21	0.0708587806448698	0.0441668631166169	0.0441668631166169\\
0.26	0.0853933829131391	0.0497581921247437	0.0497581921247437\\
0.31	0.0862231939896845	0.0434205334923525	0.0434205334923525\\
0.36	0.116797521568512	0.0678349676029757	0.0678349676029757\\
0.41	0.128151853325039	0.0718129357101763	0.0718129357101763\\
0.46	0.132942193609949	0.0706346363187994	0.0706346363187994\\
0.51	0.152309000240572	0.085885779601841	0.085885779601841\\
0.56	0.150006525021769	0.0795604436348846	0.0795604436348846\\
0.61	0.220690580403429	0.129311425390988	0.129311425390988\\
0.66	0.237159029051444	0.143903180001975	0.143903180001975\\
0.71	0.207635031635032	0.116187466228761	0.116187466228761\\
};
\addplot [color=mycolor1,solid,mark=triangle,mark size=3,mark options={solid},forget plot]
  table[row sep=crcr]{%
0.01	0.000898989898989899\\
0.06	0.017943662375313\\
0.11	0.0488615760114332\\
0.16	0.0725180135061487\\
0.21	0.0893442810974355\\
0.26	0.111332312983947\\
0.31	0.110236026741662\\
0.36	0.145440228893442\\
0.41	0.156773397854368\\
0.46	0.153136937354226\\
0.51	0.166328250278825\\
0.56	0.18913853335758\\
0.61	0.230281083983096\\
0.66	0.259698184075388\\
0.71	0.293496807114454\\
};
\addplot [color=mycolor1,only marks,mark=triangle,mark size=3,mark options={solid},forget plot]
 plot [error bars/.cd, y dir = both, y explicit]
table[row sep=crcr, x index=0, y index=1, y error index=2]{%
0.01	0.000898989898989899	0.00444913874763435	0.00444913874763435\\
0.06	0.017943662375313	0.0196197391387878	0.0196197391387878\\
0.11	0.0488615760114332	0.0351643486434482	0.0351643486434482\\
0.16	0.0725180135061487	0.0465009870096554	0.0465009870096554\\
0.21	0.0893442810974355	0.0578447052636508	0.0578447052636508\\
0.26	0.111332312983947	0.0486740532199607	0.0486740532199607\\
0.31	0.110236026741662	0.0424825745889469	0.0424825745889469\\
0.36	0.145440228893442	0.0762365220163278	0.0762365220163278\\
0.41	0.156773397854368	0.0920037741162472	0.0920037741162472\\
0.46	0.153136937354226	0.0897211793261853	0.0897211793261853\\
0.51	0.166328250278825	0.0819459178720505	0.0819459178720505\\
0.56	0.18913853335758	0.092636040643336	0.092636040643336\\
0.61	0.230281083983096	0.131992073728614	0.131992073728614\\
0.66	0.259698184075388	0.146541808145252	0.146541808145252\\
0.71	0.293496807114454	0.160343258798549	0.160343258798549\\
};
\addplot [color=blue,solid,forget plot]
  table[row sep=crcr]{%
0.01	0\\
0.06	0.0103375984952843\\
0.11	0.0370558330903981\\
0.16	0.0516187663588555\\
0.21	0.0740289755541888\\
0.26	0.0864329020875467\\
0.31	0.0882866344345864\\
0.36	0.112724606083063\\
0.41	0.123284852151923\\
0.46	0.128187480017111\\
0.51	0.156472179987195\\
0.56	0.153374405177091\\
0.61	0.202944588693041\\
0.66	0.236396651422193\\
0.71	0.231889855242796\\
};
\addplot [color=blue,only marks,mark=o,mark size=3,mark options={solid},forget plot]
 plot [error bars/.cd, y dir = both, y explicit]
table[row sep=crcr, x index=0, y index=1, y error index=2]{%
0.01	0	0	0\\
0.06	0.0103375984952843	0.0151101574660005	0.0151101574660005\\
0.11	0.0370558330903981	0.0277405576306795	0.0277405576306795\\
0.16	0.0516187663588555	0.0395909651879822	0.0395909651879822\\
0.21	0.0740289755541888	0.0454525922235851	0.0454525922235851\\
0.26	0.0864329020875467	0.0477815649251956	0.0477815649251956\\
0.31	0.0882866344345864	0.0420648147810222	0.0420648147810222\\
0.36	0.112724606083063	0.0611464412855034	0.0611464412855034\\
0.41	0.123284852151923	0.0745680472651365	0.0745680472651365\\
0.46	0.128187480017111	0.0690753164059837	0.0690753164059837\\
0.51	0.156472179987195	0.0803385361107129	0.0803385361107129\\
0.56	0.153374405177091	0.0893786253995217	0.0893786253995217\\
0.61	0.202944588693041	0.128523571950985	0.128523571950985\\
0.66	0.236396651422193	0.159877408753085	0.159877408753085\\
0.71	0.231889855242796	0.157885505863325	0.157885505863325\\
};
\end{axis}
\end{tikzpicture}%
\caption{Percentage of upsets in the top-$k$ (the lower the better) as a function of $k$ (left) and percentage of upsets in the top $10$ (the lower the better) as a function of the percentage of missing comparisons, for various recent methods, i.e., least-squares (green squares \cite{YaoRanking}), Serialrank (magenta triangles \cite{FogelSpectralRanking}), Spectral Sync-Rank (blue circles \cite{Cucuringu}), dilation Laplacian (red crosses, $g = 0.3$). \label{Fig:TopK}}
\end{figure}
For the same kind of artificial graph, the number of upsets in the top $10$ is displayed with respect the the percentage of missing comparisons in Figure~\ref{Fig:TopK}.

In order to assess the statistical significance of the results, the rankings given by the dilation Laplacian with $g = 0.3$ and Hodgerank are compared separately in Figure~\ref{Hodge_vs_Dilation_Upsets03} on a graph of $200$ nodes with randomly sampled missing comparisons. Indeed, for $g=0.3$, the dilation Laplacian yields significantly fewer upsets in the top 20, however we observe that Kendall's $\tau$-distance with the true ranking is sightly higher than Hodgerank. Hence, there is a trade-off between a low number of upsets in the top part and a low Kendall's $\tau$-distance.
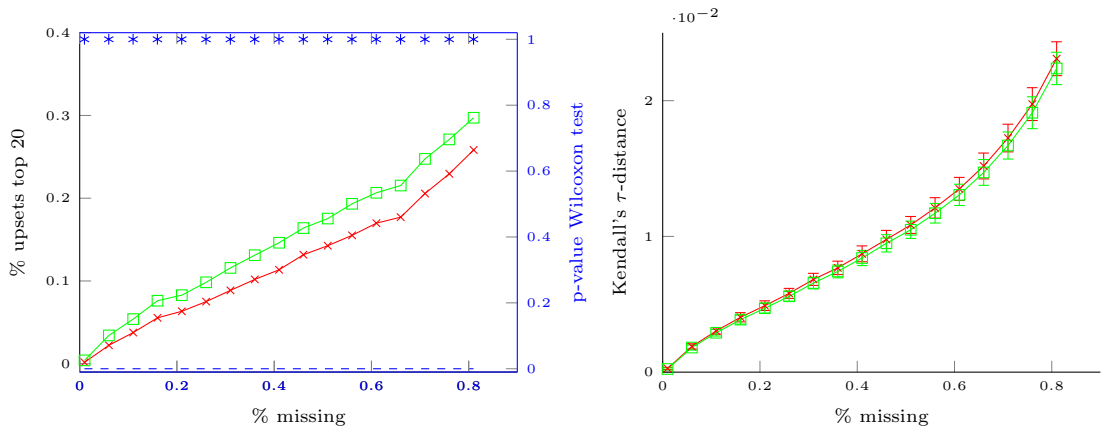
\begin{figure}[h]
\setlength\figureheight{4.5cm} 
\setlength\figurewidth{0.35\textwidth}
%
%
\begin{tikzpicture}

\begin{axis}[%
width=\figurewidth,
height=\figureheight,
at={(0.769in,0.47in)},
scale only axis,
xmin=0,
xmax=0.9,
xlabel={\% missing},
ymin=-0.01,
ymax=0.4,
ylabel={\% upsets top 20},
axis background/.style={fill=white},
axis x line*=bottom,
axis y line*=left
]
\addplot [color=red,solid,forget plot]
  table[row sep=crcr]{%
0.01	0.00168233763812507\\
0.06	0.0222836411305742\\
0.11	0.0376758257455462\\
0.16	0.0554851916672363\\
0.21	0.063391685031637\\
0.26	0.0750587237099405\\
0.31	0.0886389283595967\\
0.36	0.101907622921353\\
0.41	0.113511614636232\\
0.46	0.131729493982567\\
0.51	0.142605086338506\\
0.56	0.155114942716854\\
0.61	0.169826793798427\\
0.66	0.177099256491181\\
0.71	0.205688074214268\\
0.76	0.229521723574946\\
0.81	0.258271891612667\\
};
\addplot [color=red,only marks,mark=x,mark options={solid},forget plot]
 table[row sep=crcr, y error plus index=2, y error minus index=3]{%
0.01	0.00168233763812507	0.00311475822393703	0.00311475822393703\\
0.06	0.0222836411305742	0.0105685465236575	0.0105685465236575\\
0.11	0.0376758257455462	0.0147417417650764	0.0147417417650764\\
0.16	0.0554851916672363	0.0198596533503421	0.0198596533503421\\
0.21	0.063391685031637	0.0211686736327685	0.0211686736327685\\
0.26	0.0750587237099405	0.0221283089988737	0.0221283089988737\\
0.31	0.0886389283595967	0.025470695702589	0.025470695702589\\
0.36	0.101907622921353	0.0312587148765036	0.0312587148765036\\
0.41	0.113511614636232	0.0364029116190535	0.0364029116190535\\
0.46	0.131729493982567	0.0385947783505076	0.0385947783505076\\
0.51	0.142605086338506	0.0445477550484087	0.0445477550484087\\
0.56	0.155114942716854	0.0501257074747034	0.0501257074747034\\
0.61	0.169826793798427	0.0497183098460833	0.0497183098460833\\
0.66	0.177099256491181	0.0571436141533203	0.0571436141533203\\
0.71	0.205688074214268	0.0691597444114423	0.0691597444114423\\
0.76	0.229521723574946	0.0820119894988833	0.0820119894988833\\
0.81	0.258271891612667	0.0913488984801351	0.0913488984801351\\
};
\addplot [color=green,solid,forget plot]
  table[row sep=crcr]{%
0.01	0.0040766325549533\\
0.06	0.0342120318809598\\
0.11	0.0539485544524387\\
0.16	0.0761528458008495\\
0.21	0.0828665319744654\\
0.26	0.0983850432942352\\
0.31	0.115783526912681\\
0.36	0.131158795579975\\
0.41	0.146200159355363\\
0.46	0.163928744503091\\
0.51	0.175417976531832\\
0.56	0.193218769626461\\
0.61	0.206565655031674\\
0.66	0.21531962176492\\
0.71	0.247515508940072\\
0.76	0.27106147178917\\
0.81	0.297226138458661\\
};
\addplot [color=green,only marks,mark=square,mark options={solid},forget plot]
 table[row sep=crcr, y error plus index=2, y error minus index=3]{%
0.01	0.0040766325549533	0.00464014957351686	0.00464014957351686\\
0.06	0.0342120318809598	0.0128266009098805	0.0128266009098805\\
0.11	0.0539485544524387	0.0176631275310949	0.0176631275310949\\
0.16	0.0761528458008495	0.0230650520378204	0.0230650520378204\\
0.21	0.0828665319744654	0.0259485297167543	0.0259485297167543\\
0.26	0.0983850432942352	0.0261082672998957	0.0261082672998957\\
0.31	0.115783526912681	0.0310454052412304	0.0310454052412304\\
0.36	0.131158795579975	0.0384262410855951	0.0384262410855951\\
0.41	0.146200159355363	0.0410041699948762	0.0410041699948762\\
0.46	0.163928744503091	0.0478183819435901	0.0478183819435901\\
0.51	0.175417976531832	0.0485113071942895	0.0485113071942895\\
0.56	0.193218769626461	0.054241786920021	0.054241786920021\\
0.61	0.206565655031674	0.0590013462515479	0.0590013462515479\\
0.66	0.21531962176492	0.0656927499655219	0.0656927499655219\\
0.71	0.247515508940072	0.0725005029372706	0.0725005029372706\\
0.76	0.27106147178917	0.0931404684658391	0.0931404684658391\\
0.81	0.297226138458661	0.102197180137503	0.102197180137503\\
};
\end{axis}

\begin{axis}[%
blue,
width=\figurewidth,
height=\figureheight,
at={(0.769in,0.47in)},
scale only axis,
axis y line*=right,
xmin=0,
xmax=0.9,
ymin=-0.01,
ymax=1.02,
ylabel={p-value Wilcoxon test},
]
\addplot [color=blue,dashed,forget plot]
  table[row sep=crcr]{%
0.01	2.47307015262576e-13\\
0.06	5.38146906919442e-29\\
0.11	1.36435094620443e-31\\
0.16	1.16625112149376e-31\\
0.21	4.4512755056805e-29\\
0.26	4.70315417461643e-30\\
0.31	8.18254673959954e-29\\
0.36	3.38745833736857e-28\\
0.41	1.63313948772205e-26\\
0.46	1.79523770992387e-24\\
0.51	5.25825186668884e-21\\
0.56	2.24971242748843e-20\\
0.61	2.46175966042319e-20\\
0.66	2.40535671870766e-16\\
0.71	9.30398991547274e-15\\
0.76	1.2548435754353e-11\\
0.81	5.64299321774058e-08\\
};
\addplot [color=blue,only marks,mark=asterisk,mark options={solid},forget plot]
  table[row sep=crcr]{%
0.01	1\\
0.06	1\\
0.11	1\\
0.16	1\\
0.21	1\\
0.26	1\\
0.31	1\\
0.36	1\\
0.41	1\\
0.46	1\\
0.51	1\\
0.56	1\\
0.61	1\\
0.66	1\\
0.71	1\\
0.76	1\\
0.81	1\\
};
\end{axis}
\end{tikzpicture}%
%
%
\begin{tikzpicture}

\begin{axis}[%
width=\figurewidth,
height=\figureheight,
at={(0.758in,0.481in)},
scale only axis,
xmin=0,
xmax=0.9,
xlabel={\% missing},
ymin=0,
ymax=0.025,
ylabel={$\text{Kendall's }\tau\text{-distance}$},
axis background/.style={fill=white},
axis x line*=bottom,
axis y line*=left
]
\addplot [color=red,solid,forget plot]
  table[row sep=crcr]{%
0.01	0.000270226130653267\\
0.06	0.00189195979899498\\
0.11	0.00303492462311558\\
0.16	0.00404246231155779\\
0.21	0.00489761306532663\\
0.26	0.00580263819095477\\
0.31	0.00683379396984924\\
0.36	0.00768190954773869\\
0.41	0.00871695979899498\\
0.46	0.00977751256281407\\
0.51	0.0108430904522613\\
0.56	0.0120998743718593\\
0.61	0.0134983668341709\\
0.66	0.0151880653266332\\
0.71	0.0172464824120603\\
0.76	0.0197456030150754\\
0.81	0.0230905778894472\\
};
\addplot [color=red,only marks,mark=x,mark options={solid},forget plot]
 plot [error bars/.cd, y dir = both, y explicit]
 table[row sep=crcr, y error plus index=2, y error minus index=3]{%
0.01	0.000270226130653267	8.01497860378433e-05	8.01497860378433e-05\\
0.06	0.00189195979899498	0.000195305679829077	0.000195305679829077\\
0.11	0.00303492462311558	0.000236727887861604	0.000236727887861604\\
0.16	0.00404246231155779	0.000339625130795651	0.000339625130795651\\
0.21	0.00489761306532663	0.00035788261891456	0.00035788261891456\\
0.26	0.00580263819095477	0.000367920275098329	0.000367920275098329\\
0.31	0.00683379396984924	0.000445786084712163	0.000445786084712163\\
0.36	0.00768190954773869	0.000493519694286142	0.000493519694286142\\
0.41	0.00871695979899498	0.000583986778380399	0.000583986778380399\\
0.46	0.00977751256281407	0.000665864985768637	0.000665864985768637\\
0.51	0.0108430904522613	0.00062440340430585	0.00062440340430585\\
0.56	0.0120998743718593	0.000752072692023802	0.000752072692023802\\
0.61	0.0134983668341709	0.000848727577938287	0.000848727577938287\\
0.66	0.0151880653266332	0.000952831410447162	0.000952831410447162\\
0.71	0.0172464824120603	0.00101957595560404	0.00101957595560404\\
0.76	0.0197456030150754	0.00121023655950762	0.00121023655950762\\
0.81	0.0230905778894472	0.00124773868151769	0.00124773868151769\\
};
\addplot [color=green,solid,forget plot]
  table[row sep=crcr]{%
0.01	0.00023178391959799\\
0.06	0.00181155778894472\\
0.11	0.00290201005025126\\
0.16	0.00386896984924623\\
0.21	0.00470967336683417\\
0.26	0.00558869346733668\\
0.31	0.00656356783919598\\
0.36	0.00742022613065327\\
0.41	0.00840741206030151\\
0.46	0.00949057788944724\\
0.51	0.0104886934673367\\
0.56	0.0117072864321608\\
0.61	0.0130600502512563\\
0.66	0.0147131909547739\\
0.71	0.0166958542713568\\
0.76	0.0191130653266332\\
0.81	0.0223752512562814\\
};
\addplot [color=green,only marks,mark=square,mark options={solid},forget plot]
 plot [error bars/.cd, y dir = both, y explicit]
 table[row sep=crcr, y error plus index=2, y error minus index=3]{%
0.01	0.00023178391959799	7.86692684054201e-05	7.86692684054201e-05\\
0.06	0.00181155778894472	0.00018506395006968	0.00018506395006968\\
0.11	0.00290201005025126	0.000240548585088628	0.000240548585088628\\
0.16	0.00386896984924623	0.000326117790697218	0.000326117790697218\\
0.21	0.00470967336683417	0.000353079311323086	0.000353079311323086\\
0.26	0.00558869346733668	0.000361793174768157	0.000361793174768157\\
0.31	0.00656356783919598	0.000432187449812576	0.000432187449812576\\
0.36	0.00742022613065327	0.000473807935546208	0.000473807935546208\\
0.41	0.00840741206030151	0.000558954191338472	0.000558954191338472\\
0.46	0.00949057788944724	0.000639306022103855	0.000639306022103855\\
0.51	0.0104886934673367	0.000636604925115369	0.000636604925115369\\
0.56	0.0117072864321608	0.000723457924905005	0.000723457924905005\\
0.61	0.0130600502512563	0.000784471020203305	0.000784471020203305\\
0.66	0.0147131909547739	0.000949490503049446	0.000949490503049446\\
0.71	0.0166958542713568	0.000998206069981178	0.000998206069981178\\
0.76	0.0191130653266332	0.00117453542276173	0.00117453542276173\\
0.81	0.0223752512562814	0.00119005478870751	0.00119005478870751\\
};
\end{axis}
\end{tikzpicture}%
\caption{On the left, mean percentage of upsets in the top 20 for $g=0.3$ (left axis). The least-squares (green squares \cite{YaoRanking}) and dilation Laplacian (red crosses, $g = 0.3$) are compared. The simulation was repeated 200 times on a graph of 200 nodes. The dashed blue line is the $p$-value of the Wilcoxon signed rank test, while the blue stars indicate the result of the test at the 5$\%$ significance level (right axis). On the right, Kendall's $\tau$-distance with the known ranking (error bars are standard deviations).\label{Hodge_vs_Dilation_Upsets03}}
\end{figure}

To summarize, in the comparison between Hodgerank and the dilation ranking, we observe that:
\begin{itemize}
\item for $g=0.1$, both methods yield a similar Kendall's $\tau$-distance (see, \emph{lhs} of Figure~\ref{Fig:Missing_100_50repeats}) while the dilation ranking has fewer upsets in its top part (see, Figure~\ref{Hodge_vs_Dilation_Upsets}).
\item for $g=0.3$, the dilation ranking yields a worse Kendall's $\tau$-distance (see, \emph{rhs} of Figure~\ref{Hodge_vs_Dilation_Upsets03}), however the number of upsets in its top part is reduced.
\end{itemize}

\subsubsection{Corrupted comparisons and inconsistencies}
We recall that the normalized score $|v_0^{(g)}|$ is determined in order to minimize the objective function 
\[
\sum_{\{i,j\in V|[i,j]\in E\}}\Big[D_gv^{(g)}_0\Big]_{ij}^2 = \lambda_0^{(g)},
\]
with the deformed gradient $[D_gv^{(g)}_0]_{ij} = e^{ga_{ij}/2}[v_{0}^{(g)}]_j-e^{-ga_{ij}/2}[v_{0}^{(g)}]_i$ for all $[i,j]\in E$. The absolute values of the deformed gradient on the edges of the graph effectively rank the inconsistent comparisons.
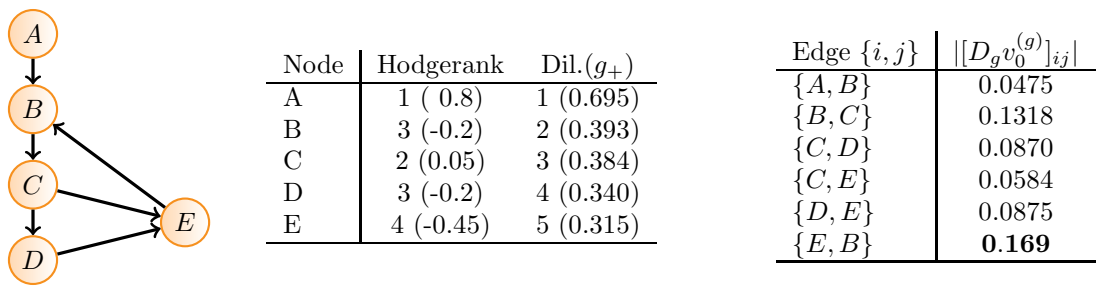
\begin{figure}[h]
\centering
\begin{minipage}{0.2\textwidth}
\begin{tikzpicture}[auto, thick]
\foreach \place/\name in {{(1,4)/a}, {(1,3)/b}, {(1,2)/c}, {(1,1)/d},{(3,1.5)/e}}
    \node[Node2] (\name) at \place {};
\foreach \source/\dest in {a/b,b/c,c/d,d/e,e/b,c/e}
    \path[->,very thick] (\source) edge (\dest);
\draw (1,4) node{$A$};
\draw (1,3) node{$B$};
\draw (1,2) node{$C$};
\draw (1,1) node{$D$};
\draw (3,1.5) node{$E$};
\end{tikzpicture}
\end{minipage}
\begin{minipage}{0.4\textwidth}
\begin{tabular}{l |  c c}
   Node &  Hodgerank &  Dil.($g_+$)\\
   \hline
A   &  1 ( 0.8)  & 1 (0.695)\\
B   &  3 (-0.2)  & 2 (0.393)\\
C   &  2 (0.05)  & 3 (0.384)\\
D   &  3 (-0.2)  & 4 (0.340)\\
E   &  4 (-0.45)  & 5 (0.315)\\
  \hline  
\end{tabular}
\end{minipage}
\begin{minipage}{0.25\textwidth}
\begin{tabular}{l | c }
Edge $\{i,j\}$ & $|[D_gv^{(g)}_0]_{ij}|$\\
\hline
$\{A,B\}$  & $0.0475$\\
$\{B,C\}$  & $0.1318$\\
$\{C,D\}$  & $0.0870$\\
$\{C,E\}$  & $0.0584$\\
$\{D,E\}$  & $0.0875$\\
$\{E,B\}$  & ${\bf 0.169}$ \\
\hline
\end{tabular}
\end{minipage}
\caption{\label{Fig:DirectComparison} On the left, rankings and scores (in parenthesis)  with an inconsistency. On the right, absolute value of the deformed gradient $|[D_g v^{(g)}_0]_{ij}|$ on the edges of the graph, for $g = 0.3$ (the largest value is in bold).The residual of the Hodgerank method is smaller than $10^{-14}$, while for the dilation Laplacian, we have $g= 0.3$ with $\lambda_0^{(g)} = 0.1054$.} 
\end{figure}

For a larger value of $g$, the dilation Laplacian allows to give more importance to the direct comparisons in the top part of the ranking. In this case, the dilation Laplacian method differs even more significantly from Hodgerank. This effect is illustrated in Figure \ref{Fig:DirectComparison}, where the edge between $B$ and $E$ may be seen as a corruption. In this example, Hodgerank gives an equal score to $B$ and $D$, and, furthermore, the score of $B$ is lower than the score of $C$. This is in contradiction with the direct measurements $B\succ C$, and $E\succ B$, i.e., there are two upsets. For the value $g = 0.3$, the score obtained by the dilation Laplacian emphasizes the direct comparison $B\succ C$, so that the score of $B$ is larger than the scores of $C$, $D$ and $E$. Indeed, in the graph of Figure~\ref{Fig:DirectComparison}, the absolute value of the deformed gradient is found to be the largest on the edge $\{E,B\}$. Namely, this example illustrates that the deformed gradient can rank the inconsistent edges. The least-squares method has a residual close to zero ($\leq 10^{-14}$), so that a similar treatment of inconsistencies can not be done in the same way.

As a matter of fact, inconsistencies are difficult to remove. Here, among the two alternative, it seems more intuitive to assume that $B$ is better than $E$ since $B\succ C \succ D\succ E$, so that the dilation Laplacian method gives a sensible result with less ties (the ranking has only one upset).
Finally, Sync-Rank gives a ranking where the only upset is the edge from $B$ to $C$, which is an admissible possibility.

In order to illustrate the methods in the presence of corruptions, an artificial graph with $50$ percent of missing comparisons (sampled uniformly at random) was built following the same methodology as before. An additional percentage of the remaining comparisons is then corrupted, i.e., the comparisons are reversed.
In Figure~\ref{Fig:Top10MissingCorrupted}, we observe that Serialrank gives a lower number of upsets in the top $10$ for a large number of corrupted comparisons. It is known that Serialrank is more robust to corruptions since its similarity measure relies on repeated comparisons. This feature may be desirable although Serialrank has a worse performance in the absence of corruptions. Furthermore, Serialrank can not rank the nodes in the absence of repeated comparisons, for instance, in the case of the line graph of Figure \ref{Fig:Line}. For a small number of corruptions, the dilation Laplacian method gives a smaller number of upsets.

\begin{figure}[h]\centering
\setlength\figureheight{4cm} 
\setlength\figurewidth{12cm}
%
%
\definecolor{mycolor1}{rgb}{1.00000,0.00000,1.00000}%
\begin{tikzpicture}

\begin{axis}[%
width=\figurewidth,
height=\figureheight,
at={(0.769in,0.47in)},
scale only axis,
xmin=-0.001,
xmax=0.1,
ymin=0,
ymax=0.45,
axis background/.style={fill=white},
ylabel={\% upsets in top $10$},
xlabel={\% corrupted},
axis x line*=bottom,
axis y line*=left
]
\addplot [color=red,solid,forget plot]
  table[row sep=crcr]{%
0	0.0676744851113757\\
0.005	0.100241308190334\\
0.01	0.111106619089193\\
0.015	0.15875015040375\\
0.02	0.153728685914404\\
0.025	0.183628812307589\\
0.03	0.190383437708916\\
0.035	0.21417155361483\\
0.04	0.22329484379132\\
0.045	0.234815043455361\\
0.05	0.240072027511914\\
0.055	0.25895555857038\\
0.06	0.278917233247743\\
0.065	0.281868360705961\\
0.07	0.283132641031526\\
0.075	0.298400864371453\\
0.08	0.273703994155798\\
0.085	0.315804348734972\\
0.09	0.292954991030654\\
0.095	0.327797042022173\\
0.1	0.318333033037643\\
};
\addplot [color=red,only marks,mark=x,mark size=3,mark options={solid},forget plot]
 plot [error bars/.cd, y dir = both, y explicit]
table[row sep=crcr, x index=0, y index=1, y error index=2]{%
0	0.0676744851113757	0.0435168024925222	0.0435168024925222\\
0.005	0.100241308190334	0.0560338396810968	0.0560338396810968\\
0.01	0.111106619089193	0.0528761659481602	0.0528761659481602\\
0.015	0.15875015040375	0.0653692651482659	0.0653692651482659\\
0.02	0.153728685914404	0.0737228028438698	0.0737228028438698\\
0.025	0.183628812307589	0.0687395886174842	0.0687395886174842\\
0.03	0.190383437708916	0.0810929119457592	0.0810929119457592\\
0.035	0.21417155361483	0.0836127365863232	0.0836127365863232\\
0.04	0.22329484379132	0.0840388440507911	0.0840388440507911\\
0.045	0.234815043455361	0.0831855154406194	0.0831855154406194\\
0.05	0.240072027511914	0.0879373366581172	0.0879373366581172\\
0.055	0.25895555857038	0.08593661453068	0.08593661453068\\
0.06	0.278917233247743	0.0960408456448024	0.0960408456448024\\
0.065	0.281868360705961	0.0806655452122433	0.0806655452122433\\
0.07	0.283132641031526	0.0911908275396911	0.0911908275396911\\
0.075	0.298400864371453	0.110898921022213	0.110898921022213\\
0.08	0.273703994155798	0.0804092818674277	0.0804092818674277\\
0.085	0.315804348734972	0.0919946792660921	0.0919946792660921\\
0.09	0.292954991030654	0.0829402316174874	0.0829402316174874\\
0.095	0.327797042022173	0.0886816158266547	0.0886816158266547\\
0.1	0.318333033037643	0.0898915166090912	0.0898915166090912\\
};
\addplot [color=green,solid,forget plot]
  table[row sep=crcr]{%
0	0.0930333184243449\\
0.005	0.12376839353361\\
0.01	0.129979363902791\\
0.015	0.169674258016222\\
0.02	0.168733935661949\\
0.025	0.200317388278976\\
0.03	0.200889634824745\\
0.035	0.227604189645991\\
0.04	0.221347252984056\\
0.045	0.240852644541034\\
0.05	0.254082305418905\\
0.055	0.271330261276629\\
0.06	0.280019448975651\\
0.065	0.272591551749135\\
0.07	0.277629129401842\\
0.075	0.299383369995899\\
0.08	0.275915208013246\\
0.085	0.317283831335091\\
0.09	0.299904933616241\\
0.095	0.325337487436939\\
0.1	0.312451775671816\\
};
\addplot [color=green,only marks,mark=square,mark size=3,mark options={solid},forget plot]
 plot [error bars/.cd, y dir = both, y explicit]
table[row sep=crcr, x index=0, y index=1, y error index=2]{%
0	0.0930333184243449	0.0465183044366574	0.0465183044366574\\
0.005	0.12376839353361	0.0561974514993484	0.0561974514993484\\
0.01	0.129979363902791	0.0567307336268805	0.0567307336268805\\
0.015	0.169674258016222	0.0625078832963222	0.0625078832963222\\
0.02	0.168733935661949	0.0850145678951732	0.0850145678951732\\
0.025	0.200317388278976	0.0779108611773723	0.0779108611773723\\
0.03	0.200889634824745	0.0834389662437914	0.0834389662437914\\
0.035	0.227604189645991	0.0904368324060344	0.0904368324060344\\
0.04	0.221347252984056	0.0887389093927328	0.0887389093927328\\
0.045	0.240852644541034	0.0847969257459724	0.0847969257459724\\
0.05	0.254082305418905	0.0920872694822257	0.0920872694822257\\
0.055	0.271330261276629	0.0902230687487305	0.0902230687487305\\
0.06	0.280019448975651	0.0931937755819433	0.0931937755819433\\
0.065	0.272591551749135	0.0724319823581288	0.0724319823581288\\
0.07	0.277629129401842	0.0923305289715321	0.0923305289715321\\
0.075	0.299383369995899	0.11688331303817	0.11688331303817\\
0.08	0.275915208013246	0.0918985104684488	0.0918985104684488\\
0.085	0.317283831335091	0.100587902548595	0.100587902548595\\
0.09	0.299904933616241	0.0943201171306574	0.0943201171306574\\
0.095	0.325337487436939	0.0863733536834554	0.0863733536834554\\
0.1	0.312451775671816	0.0861076276850679	0.0861076276850679\\
};
\addplot [color=mycolor1,solid,mark=triangle,mark size=3,mark options={solid},forget plot]
  table[row sep=crcr]{%
0	0.117463778337502\\
0.005	0.137508657848278\\
0.01	0.135806108081153\\
0.015	0.170784879368886\\
0.02	0.158271934034544\\
0.025	0.17939347630652\\
0.03	0.159680312208246\\
0.035	0.193737426642608\\
0.04	0.204162892360393\\
0.045	0.216460379591018\\
0.05	0.205377410887043\\
0.055	0.217292532881272\\
0.06	0.237391213918402\\
0.065	0.24730389528397\\
0.07	0.239057732740264\\
0.075	0.248124646029233\\
0.08	0.248032977148876\\
0.085	0.248150451829177\\
0.09	0.259236064576348\\
0.095	0.312939638969038\\
0.1	0.31116073865682\\
};
\addplot [color=mycolor1,only marks,mark=triangle,mark size=3,mark options={solid},forget plot]
 plot [error bars/.cd, y dir = both, y explicit]
table[row sep=crcr, x index=0, y index=1, y error index=2]{%
0	0.117463778337502	0.0596259362278871	0.0596259362278871\\
0.005	0.137508657848278	0.0629986095600844	0.0629986095600844\\
0.01	0.135806108081153	0.0629636520451169	0.0629636520451169\\
0.015	0.170784879368886	0.0676499219354841	0.0676499219354841\\
0.02	0.158271934034544	0.0691940782045909	0.0691940782045909\\
0.025	0.17939347630652	0.0729301863963973	0.0729301863963973\\
0.03	0.159680312208246	0.0728160891540194	0.0728160891540194\\
0.035	0.193737426642608	0.0788930617977241	0.0788930617977241\\
0.04	0.204162892360393	0.0915283808856313	0.0915283808856313\\
0.045	0.216460379591018	0.0835066536265645	0.0835066536265645\\
0.05	0.205377410887043	0.0763991614838886	0.0763991614838886\\
0.055	0.217292532881272	0.0990640295136563	0.0990640295136563\\
0.06	0.237391213918402	0.0903519676801555	0.0903519676801555\\
0.065	0.24730389528397	0.0773501787888409	0.0773501787888409\\
0.07	0.239057732740264	0.0743903810281514	0.0743903810281514\\
0.075	0.248124646029233	0.103848766716969	0.103848766716969\\
0.08	0.248032977148876	0.0859183983353929	0.0859183983353929\\
0.085	0.248150451829177	0.0802177943424586	0.0802177943424586\\
0.09	0.259236064576348	0.0878931907058181	0.0878931907058181\\
0.095	0.312939638969038	0.0974911718833788	0.0974911718833788\\
0.1	0.31116073865682	0.0882015071906943	0.0882015071906943\\
};
\addplot [color=blue,solid,mark=o,mark size=3,mark options={solid},forget plot]
  table[row sep=crcr]{%
0	0.0862928425387077\\
0.005	0.119141923565809\\
0.01	0.12322172695682\\
0.015	0.165131300257103\\
0.02	0.168504687556404\\
0.025	0.200054115481699\\
0.03	0.196073114979198\\
0.035	0.226871878902971\\
0.04	0.224794348157327\\
0.045	0.237647235604265\\
0.05	0.259405059809214\\
0.055	0.262690835057708\\
0.06	0.286585844000727\\
0.065	0.276391387088753\\
0.07	0.276982369713073\\
0.075	0.301065101600531\\
0.08	0.274329253617831\\
0.085	0.326686071144215\\
0.09	0.303173121640697\\
0.095	0.33349879326481\\
0.1	0.318188119874554\\
};
\addplot [color=blue,only marks,mark=o,mark size=3,mark options={solid},forget plot]
 plot [error bars/.cd, y dir = both, y explicit]
table[row sep=crcr, x index=0, y index=1, y error index=2]{%
0	0.0862928425387077	0.0470192322633289	0.0470192322633289\\
0.005	0.119141923565809	0.057677258067284	0.057677258067284\\
0.01	0.12322172695682	0.052684736449638	0.052684736449638\\
0.015	0.165131300257103	0.0614033917301455	0.0614033917301455\\
0.02	0.168504687556404	0.0809614253523052	0.0809614253523052\\
0.025	0.200054115481699	0.0773327250074487	0.0773327250074487\\
0.03	0.196073114979198	0.0821284665483353	0.0821284665483353\\
0.035	0.226871878902971	0.0802642916434888	0.0802642916434888\\
0.04	0.224794348157327	0.0830813306223071	0.0830813306223071\\
0.045	0.237647235604265	0.0814142095704228	0.0814142095704228\\
0.05	0.259405059809214	0.0945996776599062	0.0945996776599062\\
0.055	0.262690835057708	0.0917964537307278	0.0917964537307278\\
0.06	0.286585844000727	0.0960923352317236	0.0960923352317236\\
0.065	0.276391387088753	0.0751913931055894	0.0751913931055894\\
0.07	0.276982369713073	0.089638284848412	0.089638284848412\\
0.075	0.301065101600531	0.124326947386108	0.124326947386108\\
0.08	0.274329253617831	0.086021510225749	0.086021510225749\\
0.085	0.326686071144215	0.0933567614076073	0.0933567614076073\\
0.09	0.303173121640697	0.0864654070492151	0.0864654070492151\\
0.095	0.33349879326481	0.0871676209146329	0.0871676209146329\\
0.1	0.318188119874554	0.0797054917408921	0.0797054917408921\\
};
\end{axis}
\end{tikzpicture}%
\caption{Percentage of upsets in the top $10$ (the lower the better) in a set of $100$ objects with $50$ percent of missing comparisons and  various percentage of corrupted comparisons. The methods compared are least-squares (green squares \cite{YaoRanking}), Serialrank (magenta triangles \cite{FogelSpectralRanking}), Spectral Sync-Rank (blue circles \cite{Cucuringu}), dilation Laplacian (red crosses, $g = 0.3$). \label{Fig:Top10MissingCorrupted}}
\end{figure}
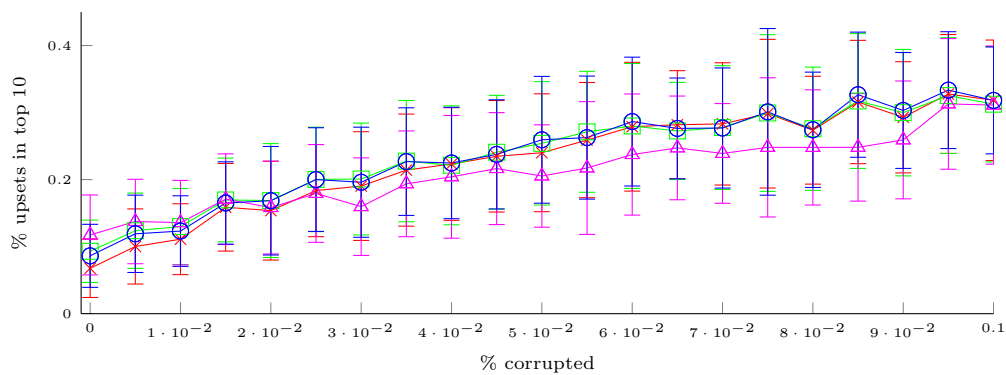

\subsection{Cardinal comparisons}

\subsubsection{Cardinal comparisons in the absence of noise}
As an application of the dilation Laplacian,  our method is used for ranking currencies based on the data taken from the Currency Converter Yahoo! Finance on $6$ November 2007, as considered in \cite{YaoRanking}.

\begin{table}[h]
\begin{minipage}{.5\linewidth}
\begin{tabular}{lccr}
   Currency & Dilation &  Hodgerank & Ratio\\
   \hline
USD   & 0.3118 &    1.7097 &  0.1823\\
JPY   & 0.0027 &    0.0149 &  0.1823  \\
EUR   & 0.4539 &    2.4889 &  0.1824  \\
CAD   & 0.3393 &    1.8610 &  0.1823  \\
GBP   & 0.6508 &    3.5690 &  0.1824  \\
AUD   & 0.2895 &    1.5878 &  0.1823  \\
CHF   & 0.2725 &    1.4946 &  0.1824  \\
  \hline
\end{tabular}
\centering\caption{Universal equivalent of currencies. \label{Currencies}}
\end{minipage}
\begin{minipage}{.5\linewidth}
\vspace{0.75cm}
\begin{tabular}{lcr}
   Official & SerialRank &  Dilation \\
   \hline
Man City & Man City   & Man City   \\
Liverpool & Chelsea    & Liverpool   \\
Chelsea   & Liverpool  & Chelsea   \\
Arsenal  & Arsenal  & Arsenal   \\
Everton  & Everton   & Everton   \\
   \hline
      &     &     \\
\end{tabular}
\centering\caption{Top 5 ranking of the English premier league 2013-2014 ($g = 0.1$). \label{League}}
\end{minipage}

\end{table}
The least-squares method of Hodgerank \cite{YaoRanking} provides a universal equivalent of the currencies. The dilation Laplacian is constructed by choosing $W$ as a binary weight matrix and the exchange rate between two currencies directly provides us with $s_{ij} = e^{a_{ij}}$ and  the inverse exchange rate $s_{ji} =e^{-a_{ij}}$. The eigenvector $v_0$ gives us directly the universal equivalent currency, so that the dilation Laplacian is here advantageous with respect to other spectral methods. Table \ref{Currencies} shows that the universal currencies found by our method and Hodgerank are proportional.
The ratio of the universal equivalent of each currency is used in Table \ref{Currencies} for the comparison of the methods. Because a universal equivalent currency is always defined up to a global factor, both methods give equivalent results. The accuracy of the method based on the dilation Laplacian is given by the smallest eigenvalue $\lambda_0 < 10^{-5}$. Typically, this financial network has no inconsistency in absence of arbitrage opportunities. Naturally, in this case, the eigenvector $v_0$ has to be normalized so that the product of its values on the objects is equal to unity.

The problem of ranking teams in the England premier league season 2013-2014 was also addressed using the dilation Laplacian.
There exist many different possible preprocessings of the data in order to construct the network of pairwise comparisons, and which influence the ranking \cite{Cucuringu}. However, we have chosen one of the simplest constructions, which does not take into account the number of goals scored. 
The comparisons between teams are defined as in \cite{FogelSpectralRanking}, as the average result (win (+1), loss (-1) or tie (0)) of home and away matches for each pair of teams (i.e., for instance, $a_{ij} =1/2$ results of the average of one victory of $i$ against $j$ and one tie).
Because each comparison has the same weight and since we use here $a_{ij}\in\{ -1,-1/2,0,1/2,1\}$, we choose to use the dilation Laplacian (\ref{DilationLaplacian}) with $g =0.1$.
In Table \ref{League}, we observe that the ranking obtained by the dilation Laplacian reproduces the official ranking for the top ranked teams, whereas it differs from it in the lower part of the ranking.
In Figure \ref{Fig:League}, the number of upsets in various top-$k$ for the ranking of the England Premier League 2013-2014 is displayed. The ranking obtained thanks to the dilation Laplacian often gives the smallest number of upsets.
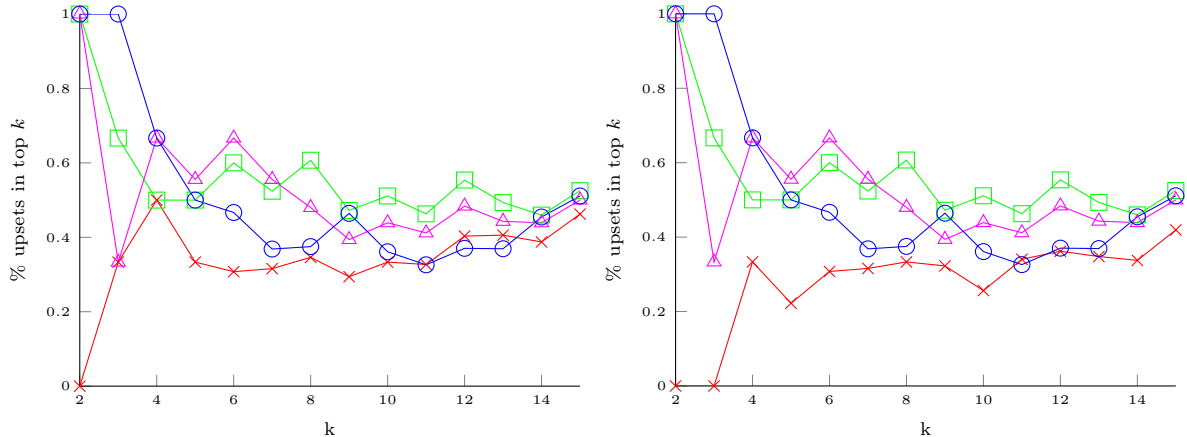
\begin{figure}[h]
\setlength\figureheight{0.3\textwidth} 
\setlength\figurewidth{0.4\textwidth}
%
%
\definecolor{mycolor1}{rgb}{1.00000,0.00000,1.00000}%
\begin{tikzpicture}

\begin{axis}[%
width=\figurewidth,
height=\figureheight,
at={(1.276in,0.909in)},
scale only axis,
unbounded coords=jump,
xmin=2,
xmax=15,
ymin=0,
ymax=1,
axis background/.style={fill=white},
title style={font=\bfseries},
ylabel={$\%$ upsets in top $k$},
xlabel={k},
axis x line*=bottom,
axis y line*=left
]
\addplot [color=red,solid,mark=x,mark size=3,mark options={solid},forget plot]
  table[row sep=crcr]{%
1	nan\\
2	0\\
3	0.333333333333333\\
4	0.5\\
5	0.333333333333333\\
6	0.307692307692308\\
7	0.315789473684211\\
8	0.346153846153846\\
9	0.294117647058824\\
10	0.333333333333333\\
11	0.326923076923077\\
12	0.403225806451613\\
13	0.405797101449275\\
14	0.3875\\
15	0.462365591397849\\
};
\addplot [color=green,solid,mark=square,mark size=3,mark options={solid},forget plot]
  table[row sep=crcr]{%
1	nan\\
2	1\\
3	0.666666666666667\\
4	0.5\\
5	0.5\\
6	0.6\\
7	0.523809523809524\\
8	0.607142857142857\\
9	0.472222222222222\\
10	0.511111111111111\\
11	0.462962962962963\\
12	0.553846153846154\\
13	0.493333333333333\\
14	0.459770114942529\\
15	0.525252525252525\\
};
\addplot [color=mycolor1,solid,mark=triangle,mark size=3,mark options={solid},forget plot]
  table[row sep=crcr]{%
1	nan\\
2	1\\
3	0.333333333333333\\
4	0.666666666666667\\
5	0.555555555555556\\
6	0.666666666666667\\
7	0.555555555555556\\
8	0.48\\
9	0.393939393939394\\
10	0.439024390243902\\
11	0.411764705882353\\
12	0.483870967741935\\
13	0.442857142857143\\
14	0.439024390243902\\
15	0.5\\
};
\addplot [color=blue,solid,mark=o,mark size=3,mark options={solid},forget plot]
  table[row sep=crcr]{%
1	nan\\
2	1\\
3	1\\
4	0.666666666666667\\
5	0.5\\
6	0.466666666666667\\
7	0.368421052631579\\
8	0.375\\
9	0.464285714285714\\
10	0.361111111111111\\
11	0.326086956521739\\
12	0.37037037037037\\
13	0.369230769230769\\
14	0.454545454545455\\
15	0.511111111111111\\
};
\end{axis}
\end{tikzpicture}%
%
%
\definecolor{mycolor1}{rgb}{1.00000,0.00000,1.00000}%
\begin{tikzpicture}

\begin{axis}[%
width=\figurewidth,
height=\figureheight,
at={(1.166in,0.565in)},
scale only axis,
unbounded coords=jump,
xmin=2,
xmax=15,
ymin=0,
ymax=1,
axis background/.style={fill=white},
ylabel={$\%$ upsets in top $k$},
xlabel={k},
axis x line*=bottom,
axis y line*=left
]
\addplot [color=red,solid,mark=x,mark size=3,mark options={solid},forget plot]
  table[row sep=crcr]{%
1	nan\\
2	0\\
3	0\\
4	0.333333333333333\\
5	0.222222222222222\\
6	0.307692307692308\\
7	0.315789473684211\\
8	0.333333333333333\\
9	0.32258064516129\\
10	0.256410256410256\\
11	0.340425531914894\\
12	0.362068965517241\\
13	0.347826086956522\\
14	0.3375\\
15	0.419354838709677\\
};
\addplot [color=green,solid,mark=square,mark size=3,mark options={solid},forget plot]
  table[row sep=crcr]{%
1	nan\\
2	1\\
3	0.666666666666667\\
4	0.5\\
5	0.5\\
6	0.6\\
7	0.523809523809524\\
8	0.607142857142857\\
9	0.472222222222222\\
10	0.511111111111111\\
11	0.462962962962963\\
12	0.553846153846154\\
13	0.493333333333333\\
14	0.459770114942529\\
15	0.525252525252525\\
};
\addplot [color=mycolor1,solid,mark=triangle,mark size=3,mark options={solid},forget plot]
  table[row sep=crcr]{%
1	nan\\
2	1\\
3	0.333333333333333\\
4	0.666666666666667\\
5	0.555555555555556\\
6	0.666666666666667\\
7	0.555555555555556\\
8	0.48\\
9	0.393939393939394\\
10	0.439024390243902\\
11	0.411764705882353\\
12	0.483870967741935\\
13	0.442857142857143\\
14	0.439024390243902\\
15	0.5\\
};
\addplot [color=blue,solid,mark=o,mark size=3,mark options={solid},forget plot]
  table[row sep=crcr]{%
1	nan\\
2	1\\
3	1\\
4	0.666666666666667\\
5	0.5\\
6	0.466666666666667\\
7	0.368421052631579\\
8	0.375\\
9	0.464285714285714\\
10	0.361111111111111\\
11	0.326086956521739\\
12	0.37037037037037\\
13	0.369230769230769\\
14	0.454545454545455\\
15	0.511111111111111\\
};
\end{axis}
\end{tikzpicture}%
\caption{Number of upsets in the top-$k$ (the lower the better) in the England Premier League 2013-2014 with: least-squares (green squares \cite{YaoRanking}), Serialrank (magenta triangles \cite{FogelSpectralRanking}), Spectral Sync-Rank (blue circles \cite{Cucuringu}), dilation Laplacian (red crosses, left $g=0.1$ and right $g = 0.5$) \label{Fig:League}}
\end{figure}

\subsubsection{Noisy cardinal comparisons}
In this section, we assume that, in the absence of noise, the edge flow is given by $a_{ij} = r_i-r_j$, for all $i$ and $j\in V$ and where $r_i$ and $r_j\in \{1,\dots,N\}$.
Following \cite{Cucuringu}, the robustness of the ranking methods are also studied for noisy cardinal comparisons. Two noise models are considered:
\begin{itemize}
\item Multiplicative Uniform Noise: MUN$(N,p,\eta)$. In this model, the measurements are given by $a_{ij} = r_i-r_j +\epsilon_{ij}$ with probability $p$ and $a_{ij} =0$ with probability $1-p$. The noise follows the discrete uniform distribution $\epsilon_{ij}\sim \mathcal{U}([-\eta (r_i-r_j),\eta (r_i-r_j)])$. Furthermore, the maximum absolute value of $a_{ij}$ is restricted to be $N-1$ (see \cite{Cucuringu}, for more details).
\item Erd\"os-R\'enyi Outliers: ERO$(N,p,\eta)$. In this model, the measurements are given by $a_{ij} = r_i-r_j $ with probability $p(1-\eta)$, while we have a random noise $a_{ij} \sim \mathcal{U}([-(N-1),(N-1)])$ (discrete uniform distribution) with probability $p\eta$ and no comparison $a_{ij} =0$ with probability $1-p$.
\end{itemize}
As mentioned in \cite{Cucuringu}, the MUN model is more realistic, while ERO model considers both perfect measurements and purely random comparisons.

The dilation Laplacian method is not expected to be very robust with respect to noise, whereas Sync-Rank relies on a robust loss function. We include in this section the ranking obtained thanks to the centrality score $\mu$ obtained from Algorithm~\ref{AlgRW}. The results of the algorithms are illustrated in Figure~\ref{Fig:MUN} and Figure~\ref{Fig:ERO}, where noisy pairwise cardinal comparisons of $N= 200$ objects were simulated.
In the case of MUN, we observe in Figure~\ref{Fig:MUN} that the dilation Laplacian method has a better performance for a low noise level, while Sync-Rank is more robust at a larger noise level. The centrality score $\mu$ is more accurate for a large $\eta$. The simulations were repeated $20$ times. The variance is large so that the error bars are not displayed. Notice that for a low noise Kendall's $\tau$-distance can simply vanish so that there is then a missing point in the logarithmic plot of Figure~\ref{Fig:MUN}.
\begin{figure}[h]
\setlength\figureheight{0.25\textwidth} 
\setlength\figurewidth{0.4\textwidth}
%
%
\begin{tikzpicture}

\begin{axis}[%
width=\figurewidth,
height=\figureheight,
at={(0.769in,0.47in)},
scale only axis,
unbounded coords=jump,
xmin=0,
xmax=0.2,
xlabel={$\eta$},
ymin=-7,
ymax=0,
ylabel={log Kendall-$\tau$ distance},
axis background/.style={fill=white}
]
\addplot [color=red,solid,mark=x,mark options={solid},forget plot]
  table[row sep=crcr]{%
0.001	-inf\\
0.011	-4.54448644679409\\
0.021	-2.86998549958596\\
0.031	-2.03779817603451\\
0.041	-1.4875538613793\\
0.051	-1.23832931999148\\
0.061	-1.10855223012801\\
0.071	-1.02728195018192\\
0.081	-0.980357002730401\\
0.091	-0.950032012214455\\
0.101	-0.924816293018017\\
0.111	-0.902519782126725\\
0.121	-0.899849991386203\\
0.131	-0.887662376512499\\
0.141	-0.878959244201816\\
0.151	-0.876484328268697\\
0.161	-0.867880515444042\\
0.171	-0.859735864173976\\
0.181	-0.857171824809628\\
0.191	-0.853318881734502\\
};
\addplot [color=blue,solid,mark=o,mark options={solid},forget plot]
  table[row sep=crcr]{%
0.001	-6.68964952169788\\
0.011	-2.77174485728846\\
0.021	-1.82483827852723\\
0.031	-1.71145394336908\\
0.041	-1.69734964894278\\
0.051	-1.65318241032861\\
0.061	-1.5842857902855\\
0.071	-1.48727586892778\\
0.081	-1.46770064460668\\
0.091	-1.43777280885462\\
0.101	-1.4425669717151\\
0.111	-1.44066564871415\\
0.121	-1.43156509760611\\
0.131	-1.44214178130311\\
0.141	-1.42804329265124\\
0.151	-1.42291143151867\\
0.161	-1.40951155513805\\
0.171	-1.41998626573822\\
0.181	-1.42341196073448\\
0.191	-1.41461112094521\\
};
\addplot [color=green,solid,mark=asterisk,mark options={solid},forget plot]
  table[row sep=crcr]{%
0.001	-3.60672070311374\\
0.011	-2.98071989715211\\
0.021	-1.69072563686301\\
0.031	-1.57390222766976\\
0.041	-1.51996359526135\\
0.051	-1.47967223816996\\
0.061	-1.45021397834062\\
0.071	-1.42922808964435\\
0.081	-1.4069945770579\\
0.091	-1.38384008887922\\
0.101	-1.36681201370179\\
0.111	-1.35031624003947\\
0.121	-1.33798768303811\\
0.131	-1.33194894527769\\
0.141	-1.31084665508044\\
0.151	-1.31136404233303\\
0.161	-1.29365563134172\\
0.171	-1.27825602698524\\
0.181	-1.29086541267218\\
0.191	-1.25813541447464\\
};
\end{axis}
\end{tikzpicture}%
%
%
\begin{tikzpicture}

\begin{axis}[%
width=\figurewidth,
height=\figureheight,
at={(0.993in,0.671in)},
scale only axis,
unbounded coords=jump,
xmin=0,
xmax=0.2,
xlabel={$\eta$},
ymin=-6,
ymax=0,
ylabel={log Kendall-$\tau$ distance},
axis background/.style={fill=white}
]
\addplot [color=red,solid,mark=x,mark options={solid},forget plot]
  table[row sep=crcr]{%
0.001	-inf\\
0.011	-5.5776591070836\\
0.021	-3.35269440876408\\
0.031	-2.26294168217803\\
0.041	-1.54216091758197\\
0.051	-1.18834183804178\\
0.061	-1.00531074886033\\
0.071	-0.89976657863049\\
0.081	-0.854236449683414\\
0.091	-0.827273150484824\\
0.101	-0.81182037359541\\
0.111	-0.795806156035971\\
0.121	-0.788410879032484\\
0.131	-0.782478285784034\\
0.141	-0.777805002244956\\
0.151	-0.772848656677347\\
0.161	-0.768566797236066\\
0.171	-0.767110207671742\\
0.181	-0.766439482370801\\
0.191	-0.761638562594803\\
};
\addplot [color=blue,solid,mark=o,mark options={solid},forget plot]
  table[row sep=crcr]{%
0.001	-inf\\
0.011	-2.83669048784766\\
0.021	-1.91136671800115\\
0.031	-1.81021260548695\\
0.041	-1.81388176981183\\
0.051	-1.81354275019167\\
0.061	-1.63589023824113\\
0.071	-1.51631719948624\\
0.081	-1.49736629831715\\
0.091	-1.47750091203692\\
0.101	-1.47043002312973\\
0.111	-1.46333821975565\\
0.121	-1.45553133992912\\
0.131	-1.46875330211843\\
0.141	-1.45881642488912\\
0.151	-1.46730822698209\\
0.161	-1.45830326225663\\
0.171	-1.45397079869665\\
0.181	-1.43795270493866\\
0.191	-1.46788600647173\\
};
\addplot [color=green,solid,mark=asterisk,mark options={solid},forget plot]
  table[row sep=crcr]{%
0.001	-inf\\
0.011	-2.8526766570667\\
0.021	-1.77402815499674\\
0.031	-1.81497666324621\\
0.041	-1.83746933214899\\
0.051	-1.80152052151313\\
0.061	-1.75770767758268\\
0.071	-1.74015918679028\\
0.081	-1.71851834791429\\
0.091	-1.68033496471091\\
0.101	-1.67712375977026\\
0.111	-1.63565807304165\\
0.121	-1.62864122186967\\
0.131	-1.59170042721076\\
0.141	-1.58484335588283\\
0.151	-1.56180898696983\\
0.161	-1.55231928163004\\
0.171	-1.53570420269834\\
0.181	-1.53037365003892\\
0.191	-1.52524457273612\\
};
\end{axis}
\end{tikzpicture}%
\caption{Natural logarithm of Kendall's $\tau$-distance (the lower the better) in the noise model MUN$(N,p,\eta)$, with $N=200$, $p = 0.2$ (left) and $p = 1$ (right). We compare Spectral Sync-Rank (blue circles), dilation Laplacian (red crosses, $g = 0.1/N$) and the random walk scores $\mu$ (\ref{eq:mu}) (green stars, $g = 1/N$).\label{Fig:MUN}}
\end{figure}
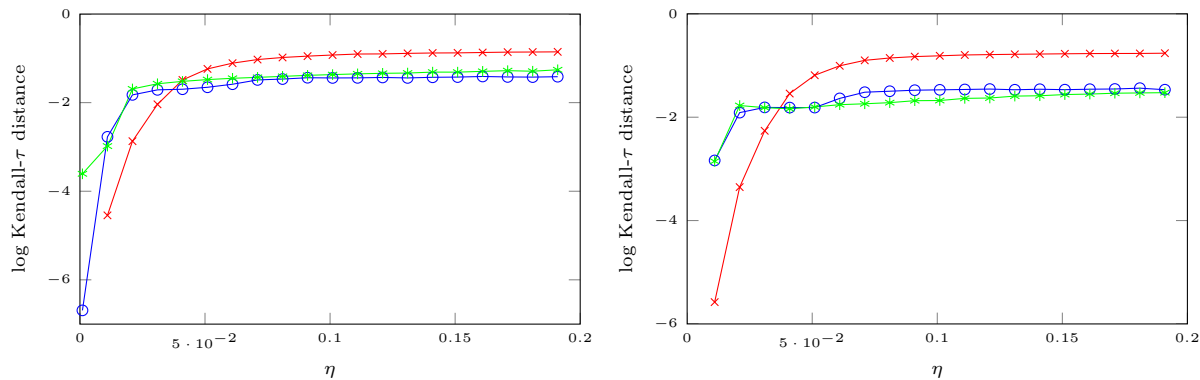
The same simulations were repeated $20$ times in the ERO noise model and the results are illustrated in Figure~\ref{Fig:ERO}, where, for clarity, the large error bars are not displayed. Obviously, Sync-Rank gives better results in this case. The methods based on the dilation Laplacian are less robust in the case of purely random comparisons. The centrality score $\mu$ is here very sensitive to the outliers.

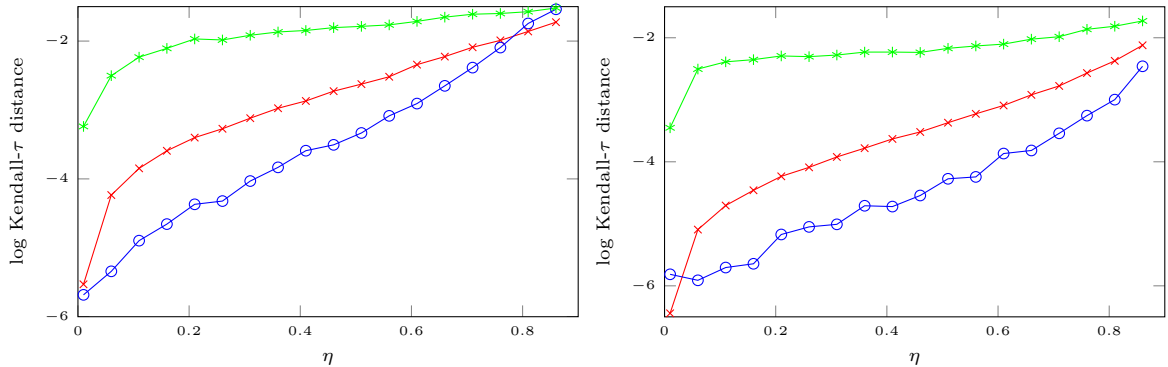
\begin{figure}[h]
\setlength\figureheight{0.25\textwidth} 
\setlength\figurewidth{0.4\textwidth}
%
%
\begin{tikzpicture}

\begin{axis}[%
width=\figurewidth,
height=\figureheight,
at={(0.769in,0.47in)},
scale only axis,
xmin=0,
xmax=0.9,
xlabel={$\eta$},
ymin=-6,
ymax=-1.5,
ylabel={log Kendall-$\tau$ distance},
axis background/.style={fill=white}
]
\addplot [color=red,solid,mark=x,mark options={solid},forget plot]
  table[row sep=crcr]{%
0.01	-5.52966027001093\\
0.06	-4.23429999065116\\
0.11	-3.84691539374076\\
0.16	-3.59210848705085\\
0.21	-3.4006447705607\\
0.26	-3.27251914102407\\
0.31	-3.11892248963691\\
0.36	-2.97474796222075\\
0.41	-2.86916603617787\\
0.46	-2.72670298720666\\
0.51	-2.62458467608622\\
0.56	-2.51764270521226\\
0.61	-2.34162893736742\\
0.66	-2.22517049755446\\
0.71	-2.08999959484725\\
0.76	-1.98980296553579\\
0.81	-1.86047958538203\\
0.86	-1.72552718858768\\
};
\addplot [color=blue,solid,mark=o,mark options={solid},forget plot]
  table[row sep=crcr]{%
0.01	-5.68265055095277\\
0.06	-5.34144520005243\\
0.11	-4.89620796628684\\
0.16	-4.65593465963355\\
0.21	-4.36884753018394\\
0.26	-4.32176861862436\\
0.31	-4.03129898914795\\
0.36	-3.83254094967213\\
0.41	-3.59078648927438\\
0.46	-3.50588810493057\\
0.51	-3.33370311996732\\
0.56	-3.08511283218062\\
0.61	-2.90690721564522\\
0.66	-2.64970425982135\\
0.71	-2.38512043569093\\
0.76	-2.09409103703572\\
0.81	-1.74421200079046\\
0.86	-1.5386493342609\\
};
\addplot [color=green,solid,mark=asterisk,mark options={solid},forget plot]
  table[row sep=crcr]{%
0.01	-3.23870043733068\\
0.06	-2.50180235051334\\
0.11	-2.2333192269477\\
0.16	-2.10561011733987\\
0.21	-1.96899748671885\\
0.26	-1.98375683174658\\
0.31	-1.91558079924798\\
0.36	-1.86777261804754\\
0.41	-1.84598750744878\\
0.46	-1.80473779692543\\
0.51	-1.78640448604425\\
0.56	-1.76516646962119\\
0.61	-1.71405191729739\\
0.66	-1.65282810983872\\
0.71	-1.6101731072467\\
0.76	-1.59945026698781\\
0.81	-1.57347191028891\\
0.86	-1.51973961244391\\
};
\end{axis}
\end{tikzpicture}%
%
%
\begin{tikzpicture}

\begin{axis}[%
width=\figurewidth,
height=\figureheight,
at={(0.758in,0.481in)},
scale only axis,
xmin=0,
xmax=0.9,
xlabel={$\eta$},
ymin=-6.5,
ymax=-1.5,
ylabel={log Kendall-$\tau$ distance},
axis background/.style={fill=white}
]
\addplot [color=red,solid,mark=x,mark options={solid},forget plot]
  table[row sep=crcr]{%
0.01	-6.44294728450538\\
0.06	-5.09302056755637\\
0.11	-4.70496275853067\\
0.16	-4.45920076383828\\
0.21	-4.23438669460058\\
0.26	-4.08950703571284\\
0.31	-3.92187032651163\\
0.36	-3.78037781267124\\
0.41	-3.63127446217122\\
0.46	-3.5189458701766\\
0.51	-3.36938469356079\\
0.56	-3.22634713304886\\
0.61	-3.09170096158456\\
0.66	-2.92023796919208\\
0.71	-2.77667747180544\\
0.76	-2.56897125354918\\
0.81	-2.37226915807413\\
0.86	-2.12136338704551\\
};
\addplot [color=blue,solid,mark=o,mark options={solid},forget plot]
  table[row sep=crcr]{%
0.01	-5.81418078434398\\
0.06	-5.91088081843736\\
0.11	-5.70503954584625\\
0.16	-5.64428174754859\\
0.21	-5.17197253955697\\
0.26	-5.04957463980189\\
0.31	-5.0083138700842\\
0.36	-4.70885706108789\\
0.41	-4.72204283215467\\
0.46	-4.54342293118424\\
0.51	-4.27409655884365\\
0.56	-4.24362020967981\\
0.61	-3.86688871419874\\
0.66	-3.81745544940477\\
0.71	-3.54097942474184\\
0.76	-3.2557600737014\\
0.81	-2.99771917075069\\
0.86	-2.46004736399437\\
};
\addplot [color=green,solid,mark=asterisk,mark options={solid},forget plot]
  table[row sep=crcr]{%
0.01	-3.45180326360162\\
0.06	-2.50442760821387\\
0.11	-2.38822235740524\\
0.16	-2.35276818663845\\
0.21	-2.29316975770571\\
0.26	-2.30455940247007\\
0.31	-2.27857907467797\\
0.36	-2.23105941836239\\
0.41	-2.23087230574343\\
0.46	-2.23606596545741\\
0.51	-2.17042241301137\\
0.56	-2.13219583404048\\
0.61	-2.1023450521507\\
0.66	-2.02249141306967\\
0.71	-1.98313596890993\\
0.76	-1.86345520048917\\
0.81	-1.81419006923467\\
0.86	-1.72879174773344\\
};
\end{axis}
\end{tikzpicture}%
\caption{Natural logarithm of Kendall's $\tau$-distance (the lower the better) in the noise model ERO$(N,p,\eta)$, with $N=200$, $p = 0.2$ (left) and $p = 1$ (right). We compare Spectral Sync-Rank (blue circles), dilation Laplacian (red crosses, $g = 0.1/N$) and the random walk scores $\mu$ (\ref{eq:mu}) (green stars, $g = 1/N$).\label{Fig:ERO}}
\end{figure}

\FloatBarrier
\section{Conclusions}
In this paper, deformed Laplacians generalizing existing Laplacians were defined. As particular cases of this construction, we have emphasized the relevance of the dilation Laplacians for ranking the nodes of networks constructed from pairwise comparisons. Relations with discrete Hodge theory were highlighted.
Furthermore, our method is based on a real-valued analogue of the connection Laplacian used in the Sync-Rank method of \cite{Cucuringu}.

The ranking method proposed in this paper relies on the computation of the least eigenvector of a dilation Laplacian, which was proved to take the same sign on a whole connected component of the comparison graph. Hence, the absolute value of this eigenmode provides \emph{directly} a ranking score, since no choice of sign is required, in analogy with the computation of the Perron eigenvector of a stochastic matrix.

On the example considered, the ranking obtained thanks to the dilation Laplacian is shown empirically to have a lower number of upsets in the top-$k$ as shown by the numerical simulations. \\

{\bf Acknowledgements}\\
\noindent {\small The authors thank the following organizations. EU: The research leading to these results has received funding from the European Research Council under the European Union's Seventh Framework Programme (FP7/2007-2013) / ERC AdG A-DATADRIVE-B (290923). This paper reflects only the authors' views, the Union is not liable for any use that may be made of the contained information.
Research Council KUL: GOA/10/09 MaNet, CoE PFV/10/002 (OPTEC),
BIL12/11T; PhD/Postdoc grants
Flemish Government:
FWO: projects: G.0377.12 (Structured systems), G.088114N (Tensor based
data similarity); PhD/Postdoc grants
IWT: projects: SBO POM (100031); PhD/Postdoc grants
iMinds Medical Information Technologies SBO 2014
Belgian Federal Science Policy Office: IUAP P7/19 (DYSCO, Dynamical
systems, control and optimization, 2012-2017).}\\
\appendix

\section{The general construction of the deformed Laplacian and connections with other works \label{AppDefLap}}

\subsection{Connection Laplacian and $SO(d)$}
In order to recover the connection Laplacian of \cite{SingerWu}, it suffices to choose $T_{[i,j]}$ as a real orthogonal matrix of $SO(d)$. 
Then, the Laplacian becomes
\begin{equation}
\big[L_{T}v\big]_{i}=\sum_{j\in V}w_{ij}\Big(v_i-O_{ij}v_j\Big).\label{ConnectionLaplacian}
\end{equation}
where we have defined the orthogonal matrix
$O_{ij} = (T_{[i,j]})^2,
$ with $O_{ij}^\intercal = O_{ij}^{-1} =  O_{ji}$, so that (\ref{ConnectionLaplacian}) is the connection Laplacian of Singer and Wu. Given any $O_{ij}\in SO(d)$, a classical result of Lie theory is that one can write $O_{ij}=\exp (o_{ij})$ with $o_{ij}\in \mathfrak{so}(d)$ an antisymmetric matrix. Hence, we choose $T_{(ij)}=\exp (o_{ij}/2)$.
\subsection{Magnetic Laplacian and $U(1)$}
The so-called magnetic Laplacian \cite{Shubin,deVerdiere,Berkolaiko} is obtained by choosing
$T_{[i,j]} = \exp({\rm i}\theta \alpha_{ij}/2),$
where $\alpha_{ij}$ is a skew-symmetric function of the oriented edges, i.e. $\alpha_{ij} = -\alpha_{ji}$.
Therefore, the general Laplacian becomes
\[\big[L_{T}v\big]_i=\sum_{j\in V}w_{ij}\Big(v_i- e^{{\rm i}\theta \alpha_{ij}}v_j\Big).\label{MagneticLaplacian}
\]
\subsection{Signed Laplacian and $\mathbb{Z}_2$}
The signed Laplacian \cite{Kunegis} used in the case of signed networks is simply obtained as a follows. We denote $\sigma_{ij} = \pm 1$ the sign of each edge.
Let us define $\Sigma_{ij} = 0$ if $\sigma_{ij} = 1$ and $\Sigma_{ij} = 1$ if $s_{i,j} = -1$. Then, we define 
$T_{[i,j]} = \exp({\rm i}\pi \Sigma_{ij}/2),$
and using $\exp({\rm i}\pi\Sigma_{ij})=\sigma_{ij}$, we find 
\[\big[L_{T}v\big]_i=\sum_{j\in V}w_{ij}\Big(v_i-\sigma_{ij} v_j\Big).\label{SignedLaplacian}
\]
\subsection{Dilation Laplacian and $\mathbb{R}^{+}_{0}$}
The dilation Laplacian is obtained by choosing $T_{[i,j]}\in\mathbb{R}^+$ for each directed edge $[i,j]$. 
The dilation group is here $\mathbb{R}^+_{0}$ for the multiplication of positive numbers.
\section{Proofs of the main results\label{AppProof}}
\begin{proof}[Proof of Proposition~\ref{PropConvexOptim}]
The objective function is purely quadratic on the feasible space. The Hessian of this function is strictly positive definite on the feasible space since $f^\intercal v_0 = 0$ and $\mathcal{L}_0$ is strictly positive definite on $(v_0)^\perp$. Indeed, since the grapg $\mathcal{G}$ is connected,  $v_0$ is the only eigenvector of $\mathcal{L}_0$ with a zero eigenvalue. Hence, the first order optimality conditions $\mathcal{L}_0f = g\mathcal{W}^{\rm diff}v_0$ and $f^\intercal v_0 = 0$ are here sufficient conditions. The vector $f^\star = g\mathcal{L}_0^\dagger\mathcal{W}^{\rm diff}v_0$ is the unique solution since $v_0^\intercal g\mathcal{L}_0^\dagger\mathcal{W}^{\rm diff}v_0 = 0$.
\end{proof}

\begin{proof}[Proof of Lemma~\ref{LemmaMoorePenrose}]
A first order perturbative expansion of the infinitesimal dilation Laplacian gives a correction to the combinatorial Laplacian: 
\[[\mathcal{L}_{g}f]_i=[\mathcal{L}_{0}f]_i-g [\mathcal{W}^{\rm diff}f]_i + \mathcal{O}(g^2),\]
where $[\mathcal{W}^{\rm diff}f]_i=w^{\rm diff}_i f_i$ for all $f\in \mathbb{R}^N$ and all $i \in V$. We introduce an expansion in powers of $g$ of the lowest eigenvector and lowest eigenvalue \[v^{(g)}_{0} = v_{0} +gf^{\star}+ \mathcal{O}(g^2)\mbox{ and }\lambda^{(g)}_0 = \lambda^{(0)}_0 + g \lambda^\star +\mathcal{O}(g^2),\]
where $\lambda^{(0)}_0 = 0$, since it corresponds to the lesat eigenvalue of the combinatorial Laplacian. We further assume the orthogonality of the first order perturbation, i.e., $f^{\star \intercal}v_{0} = 0$. By substituting these expansions in the eigenvalue equation $\mathcal{L}_{g}v^{(g)}_{0} = \lambda^{(g)}_0 v^{(g)}_{0}$, we obtain
\[g\mathcal{L}_{0}f^{\star} -g \mathcal{W}v_0 + \mathcal{O}(g^2) = g \lambda^\star v_{0} + \mathcal{O}(g^2).\]
Using $v_{0}^\intercal f^{\star} = 0$ and $v_{0}^\intercal\mathcal{W}v_0= 0$, we find $\lambda^\star = 0$. Therefore, the perturbation $f^{\star}$ satisfies
$\mathcal{L}_{0}f^{\star} - \mathcal{W}v_0  = 0$.
By using again $ v_{0}^\intercal\mathcal{W}v_0 = 0$, we obtain 
$f^{\star} =\mathcal{L}_{0}^{\dagger}\mathcal{W}v_0$
as stated in (\ref{eq:SeriesEigenVec}).
\end{proof}
\begin{proof}[Proof of Proposition~\ref{Prop2}]
($\Leftarrow$) If, for all edges $\{i,j\}\in E_u$, $a_{ij} = h_i-h_j$, we can write the dilation Laplacian as 
\[\big[\mathcal{L}_{g} v\big]_i = \sum_{j \in V}w_{ij}e^{g h_j}\big(e^{-g h_i}v_i-e^{-g h_j}v_j\big),\]
so that a zero eigenvector is indeed $[v_{0}]_i =  c \times e^{g h_i}$.

($\Rightarrow$) Assume that the graph has only one connected component. If the graph has many connected components, each of them can be treated separately. Consider an eigenvector $v$ of eigenvalue zero. Then, we have
\[ v^\intercal \mathcal{L}_{g}v= \frac{1}{2} \sum_{i,j\in V}w_{ij}\Big(e^{g a_{ij}/2}v_j-e^{-g a_{ij}/2}v_i\Big)^{2} = 0.\]
Since each term of the sum has to vanish, then for all edges $[i,j]$, we have $v_i =  e^{-g a_{ij}}v_j$
Then, $v$ necessarily takes  the same sign on all the nodes in the connected graph, so that we can write $v_i > 0$ for all $i\in V$. Hence, we can find a vector $u$ such that $v_i= e^{-u_i}$ for all node $i$. Substituting this last relation in (\ref{ScalingRelation}), we find $ga_{ij} = u_{j}-u_{i}$. Then, we define $h_i = u_{i}/g$, which proves our result.
\end{proof}
\begin{proof}[Proof of Lemma~\ref{LemmaInconsistent}]
Let $a_{ij} = h_i-h_j+\epsilon_{ij}$  with $\epsilon_{ij} =-\epsilon_{ji}$. Then, by definition, we have
\[
[\diag(e^{gh})\mathcal{L}_g(a,W)\diag(e^{gh})v]_i = \sum_{j\in V}w_{ij}e^{g(h_i+h_j)}(e^{g\epsilon_{ji}}v_i-v_j)=[\mathcal{L}_g(\epsilon,\bar{W})v]_i,
\]
with $\bar{w}_{ij} = w_{ij}e^{g(h_i+h_j)}$.
\end{proof}
\begin{proof}[Proof of Theorem~\ref{Thm1}]
We first show that $v^\star_i\geq 0$ for each node $i\in V$.
The eigenvector of smallest eigenvalue is a solution to the following optimization problem: 
\begin{equation}
\mathop{\rm minimize} _{v\in \mathbb{R}_\star^N} F(v),\ {\rm with}\ F(v) = \frac{\sum_{\{i,j\}\in E}w_{ij}\Big(s_{ji}v_j-s_{ij}v_i\Big)^{2}}{\sum_{i \in V}v_i^{2}}.\label{eq:minimizeF}
\end{equation}
Assume that the normalized solution  $v^\star$ of~(\ref{eq:minimizeF}) does not have the same sign everywhere. Then, using the inequality  
$\big|x-y \big|\geq\big||x|-|y|\big|$, we have
\[\Big(s_{ji}v^\star_j-s_{ij}v^\star_i\Big)^{2}\geq\Big(s_{ji}\big|v^\star_j\big|-s_{ij}\big|v^\star_i\big|\Big)^{2}\]
and hence, $|v^\star|$ has a lower objective value (the absolute value is taken entry-wise). Moreover, $|v^\star|$ and $v^\star$ have the same norm. This contradicts the assumption that the solution $v^\star_i$ does not have the same sign everywhere.

We show now that $v^\star_i> 0$ for all nodes $i\in V$. Firstly, assume that there exists one node $i_\star\in V$ such that $v^\star_{i_\star} = 0$. We separate the terms involving $i_\star$ in the numerator of the objective in (\ref{eq:minimizeF}), leading to
\[\sum_{\{\{i,j\}\in E |i\neq i_\star \neq j\}}w_{ij}\Big(s_{ji}v^\star_j-s_{ij}v^\star_i\Big)^{2} + \sum_{j\in V}w_{i_\star j}s^2_{ji_\star}(v^\star_j)^2.\]
By choosing a $0<u \leq (s_{j i_\star}/s_{i_\star j})v^\star_j$ for all $j\in V$ such that $\{i_\star,j\}\in E_u$, we have
\[\sum_{j\in V}w_{i_\star j}s^2_{ji_\star}(v^\star_j)^2\geq  \sum_{j\in V}w_{i_\star j}\Big(s_{ji_\star}v^\star_j-s_{i_\star j}u\Big)^{2}, \]
and,
\[\sum_{i \in V}(v^\star_i)^{2}\leq \sum_{\{i \in V| i\neq i_\star\}}(v^\star_i)^{2} + u^{2}.\]
Therefore, defining $\tilde{v}^\star_i = v^\star_i$ for all $i\in V$ such that $i\neq i_\star$, and $\tilde{v}^\star_{i_\star} = u$, we have $F(v^\star)\geq F(\tilde{v}^\star)$, which contradicts the assumption that $v^\star$ is a solution of the minimization problem.
Assume now that there is a finite number of nodes $i_\star^{(1)}, \dots, i_\star^{(m)}\in V$ such that $v^\star_{i_\star^{(\ell)}}=0$ for all $\ell\in\{1,\dots,m\}$. Then, by choosing $\tilde{v}^\star_i = v^\star_i$ for all $i\in V$ such that  $i\notin \{i_\star^{(1)}, \dots, i_\star^{(m)}\}$ and 
\[\tilde{v}^\star_{i_\star^{(\ell)}} = \min_{\ell\in\{1,\dots,m\}}\min_{j\in V}\frac{s_{[j,i^{(\ell)}_\star]}v^\star_j}{s_{[i^{(\ell)}_\star,j]}},\]
for all $\ell\in\{1,\dots,m\}$, we have $F(v^\star)\geq F(\tilde{v}^\star)$, contradicting again our initial assumption.

Finally, let us prove that the solution $v^\star$, with $\|v^\star\|_2 = 1$, is unique. If the eigenspace associated to the minimal eigenvalue has a dimension larger than $1$, one can find $\tilde{v}^\star$ in this eigenspace satisfying $\sum_{i\in V}\tilde{v}^\star_i v^\star_i= 0$. Clearly, if $\tilde{v}^\star$ has the same sign everywhere, this is impossible. Therefore, the dimension of the eigenspace is $1$.
\end{proof}
\begin{proof}[Proof of Proposition~\ref{CheegerEqual}]
 We have naturally that $\lambda_0^{(g)}\leq \eta^{[2]}_{\mathcal{G}}$, since $(\mathbb{R}_\star^+)^N\subset \mathbb{R}^N_\star$. Following Theorem \ref{Thm1}, it exists an eigenvector $v^{(g)}_0$ of the dilation Laplacian of smallest eigenvalue, satisfying $[v^{(g)}_0]_i>0$ for all $i\in V$. Therefore, choosing $v =  v^{(g)}_0\in (\mathbb{R}_\star^+)^N$, we have \[\lambda_0^{(g)} = \eta^{[2]}(v_0)\geq  \min_{v\in (\mathbb{R}_\star^+)^N}\eta(v) = \eta^{[2]}_{\mathcal{G}},\]
and therefore, $\lambda_0^{(g)} = \eta^{[2]}_{\mathcal{G}}$.
\end{proof}
\begin{proof}[Proof of Lemma~\ref{Lem:frustration}]
This is a simple consequence of Proposition~\ref{Prop2}.
\end{proof}
\begin{proof}[Proof of Theorem~\ref{Thm2}]
(i) We first prove the {\it lhs} of (\ref{eq:frustrationCheeger}). Denote $\|v\|_2^2 = \sum_i v_i^2$. For $v\in \mathbb{R}_\star^+$, and for all $i\in V$, we have
$\|v\|_2\geq v_i>0$, which yields 
$\frac{1}{N}\sum_{i\in V} v_i\leq  \|v\|_2$. Furthermore, for all $v\in \mathbb{R}_\star^+$, we find
\[
|s^{1/2}_{ij}v_j  - s^{-1/2}_{ij}v_i|\leq s^{1/2}_{ij}v_j  + s^{-1/2}_{ij}v_i\leq \|v\|_2+ s^{1/2}_{\rm max}\|v\|_2,
\] 
since, without loss of generality, if $s^{1/2}_{ij}\geq 1$, then $s^{1/2}_{ij}\leq s^{1/2}_{\rm max}$ and $(s^{1/2}_{ij})^{-1}\leq 1$.
Therefore, we obtain, for all $v\in \mathbb{R}_\star^+$,
\[
\frac{\sum_{i,j\in V}w_{ij}|s^{1/2}_{ij}v_j  - s^{-1/2}_{ij}v_i|^2}{\sum_{i\in V} v_i^2}\leq\|v\|_2(1+ s^{1/2}_{\rm max})\frac{\sum_{i,j\in V}w_{ij}|s^{1/2}_{ij}v_j  - s^{-1/2}_{ij}v_i|}{\|v\|_2\sum_{i\in V} v_i/N}.
\]
The result is obtained by taking the minimum over $v\in \mathbb{R}_\star^+$ of both sides in the above inequality.

(ii)  We now prove the {\it lhs} of (\ref{eq:frustrationCheeger}). Let $v$ be the eigenvector with eigenvalue $\lambda_0$. By Cauchy-Schwartz, we have 
\[\sum_{i,j\in V}w_{ij}|s^{1/2}_{ij}v_j  - s^{-1/2}_{ij}v_i|\leq\sqrt{\sum_{i,j\in V}w_{ij}} \sqrt{\sum_{i,j\in V}w_{ij}|s^{1/2}_{ij}v_j  - s^{-1/2}_{ij}v_i|^2},\]
and thanks to the triangle inequality $(\sum_{i\in V} v_i^2)^{1/2}\leq\sum_{i\in V}|v_i|$. Therefore, combining these inequalities, we obtain
\[\frac{\sum_{i,j\in V}w_{ij}|s^{1/2}_{ij}v_j  - s^{-1/2}_{ij}v_i|}{\sum_{i\in V}|v_i|}\leq\sqrt{{\rm vol}(\mathcal{G})} \sqrt{\frac{\sum_{i,j\in V}w_{ij}|s^{1/2}_{ij}v_j  - s^{-1/2}_{ij}v_i|^2}{\sum_{i\in V} v_i^2}}.\]
By dividing both sides by ${\rm vol}(\mathcal{G})$, we find (\ref{eq:frustrationCheeger}).

\end{proof}

\begin{proof}[Proof of Proposition~\ref{UpperboundLeastEigHodge}]
This is an application of the Rayleigh quotient formula, the fact that $\lambda_0^{(g)}\leq \eta_{a,W}^{[2]}(e^{gh})$ and of the inequality
\[
e^{g(h_i+h_j)}\leq \frac{1}{2}\sum_{k\in V} e^{2gh_k},
\]
for all $i,j\in V$, following from $2xy\leq x^2 + y^2$ for all $x,y\in \mathbb{R}$.
\end{proof}

\begin{proof}[Proof of Lemma~\ref{LemmaMajorization}]
Since $a_{ij} = h_i-h_j + \epsilon_{ij}$, we have
\[
\frac{v^\intercal\mathcal{L}_g(a,W)v}{v^\intercal v}= \frac{\sum_{i,j\in V}w_{ij}e^{g(h_i+h_j)}\Big(e^{g\epsilon_{ji/2}}(v_i e^{-gh_i})-e^{g\epsilon_{ij/2}}(v_je^{-g h_j})\Big)^2}{\sum_{i\in V}e^{2gh_i}(v_i e^{-gh_i})^2}.\]
Then, using the inequality
\[
\sum_{i\in V}e^{2gh_i}(v_i e^{-gh_i})^2\leq \sum_{i\in V}e^{2gh_i}\sum_{j\in V}(v_j e^{-gh_j})^2,
\]
and denoting $\bar{v}_g=\diag(e^{-gh})v$, we find
\[
\frac{v^\intercal\mathcal{L}_g(a,W)v}{v^\intercal v}\geq \frac{\sum_{i,j\in V}w_{ij}e^{g(h_i+h_j)}\Big(e^{g\epsilon_{ji/2}}[\bar{v}_g]_i-e^{g\epsilon_{ij/2}}[\bar{v}_g]_j\Big)^2}{\sum_{i\in V}e^{2gh_i}\sum_{j\in V}[\bar{v}_g]_j^2} = \frac{\bar{v}^\intercal_g\mathcal{L}_g(\epsilon,\bar{W})\bar{v}_g}{\bar{v}_g^\intercal\bar{v}_g},\]
with $\bar{W} = \diag(\frac{e^{gh}}{\|e^{gh}\|_2})W\diag(\frac{e^{gh}}{\|e^{gh}\|_2})$. This proves our statement.
\end{proof}
\begin{proof}[Proof of Theorem~\ref{Thm3}]
For all $v\in \mathbb{R}^N$ such that $v\neq 0$, the generalized Rayleigh quotient reads
\begin{equation}
R(v) = \frac{\sum_{i,j\in V}w_{ij}\Big(e^{-g a_{ij}/2}v_i - e^{g a_{ij}/2}v_j\Big)^{2}}{\sum_{i,j\in V}w_{ij}\Big(e^{-g a_{ij}/2}v_i\Big)^{2}}.\label{eq:GenRayleigh}
\end{equation}
Following the same arguments as in the proof of Theorem~\ref{Thm1}, if $v_i> 0$ for all $i\in S\subset V$ and  $v_i\leq 0$ for all $i\in V\setminus S$, the numerator can always be made smaller by choosing $\tilde{v}_i> 0$ for all $i\in V\setminus S$. It is also obvious that the denominator of  (\ref{eq:GenRayleigh}) is larger if $v_i> 0$ for all $i\in V$.
\end{proof}

\section{Relationship with Witten's deformed Laplacian\label{AppWitten}}
Actually, we can write the dilation Laplacian
\[[\mathcal{L}_g v ]_{i} = \sum_{j\in V}w_{ij}\big(e^{-g a_{ij}}v_i-v_j\big),
\]
as a Schr\"odinger operator \cite{CyconBook,ColindeVerdiereBook}
 $\mathcal{L}_g = \mathcal{L}_{0} + V_g$,
where $\mathcal{L}_0$ is the combinatorial Laplacian and  with the potential 
$[V_g]_{ii} = \sum_{\{ j \in V\}}w_{ij}(e^{-g a_{ij}}-1)$.
In the case where $a_{ij} = [{\rm d}h]_{ij}\equiv h_j-h_i$, this deformed Laplacian takes an interesting form.
This observation leads us to present the following analogy.
In a famous paper \cite{WittenMorse}, Witten defined a deformed Laplacian in order to prove the Morse inequalities of  differential topology. Subsequently, Forman defined a discrete Morse theory for cell complexes \cite{FormanWittenMorse,FormanMorse}. Here, we follow the construction of Forman and exhibit a deformation which corresponds to a deformed Laplacian presented in the current paper. However, we do not need to consider the discrete Morse theory in its generality.

The idea of Forman is to deform the Hodge Laplacian by first defining a real-valued function $\sigma$ on the simplicial complex, i.e. $\sigma$ should take a real value on each vertex, each edge, each triangle, etc. A Morse function is then supposed to satisfy certain properties that we do not need here.
For our purpose, let us define the function $\sigma$ only on the nodes and edges, as follows: $\sigma(i) = h_i$ and $\sigma(\{i,j\}) = (h_i+h_j)/2$. Let us notice that we do not require here $\sigma$ to be a Morse function in the sense of Forman, because we are not interested here in the topology of the network. The Witten deformed gradient is defined by
${\rm d}_{t} = e^{-t\sigma}{\rm d}e^{t\sigma}$.
Using the notations of \cite{FormanWittenMorse,FormanMorse}, we obtain the correspondence
\[{\rm d}_{t}^*{\rm d}_{t} = \mathcal{L}_{g}, \quad {\rm with } \ a = {\rm d}h\ {\rm and } \ g =t.
\]
Contrary to Forman, we only defined here the deformed exterior derivative for the node space. Defining the deformed exterior derivative for the simplicial complex would require a more complete definition of $\sigma$.

\section*{References}
\bibliography{References}
\bibliographystyle{unsrt}

\end{document}